\documentclass[conference]{IEEEtran}
\usepackage{graphicx}

\usepackage{cite}
\def\BibTeX{{\rm B\kern-.05em{\sc i\kern-.025em b}\kern-.08em
		T\kern-.1667em\lower.7ex\hbox{E}\kern-.125emX}}
	\makeatletter
	\def\endthebibliography{%
		\def\@noitemerr{\@latex@warning{Empty `thebibliography' environment}}%
		\endlist
	}
	\makeatother

\usepackage{enumitem}
\usepackage{color}

\usepackage{graphicx}
\usepackage{array}
\usepackage{subcaption}
\usepackage[export]{adjustbox}
\usepackage[abs]{overpic}

\setlist[itemize]{noitemsep, topsep=0pt}
\setlist[enumerate]{noitemsep, topsep=0pt}

\usepackage{makecell}

\usepackage{caption} 

\usepackage{amsthm}
\usepackage{amsmath}
\newtheorem{definition}{Definition} 
\newtheorem{lemma}{Lemma}
\newtheorem{example}{Example}[section]

\usepackage[linesnumbered,ruled,vlined]{algorithm2e}

\newcommand{\ouralgorithmEnd}{GeoPrune}
\newcommand{\ouralgorithmCapital}{GeoPrune }
\newcommand{\ouralgorithm}{GeoPrune }

\begin{document}
	
	
	\title{GeoPrune: Efficiently Finding Shareable Vehicles Based on Geometric Properties}

	%
	
	
	
	\author{
		
		\IEEEauthorblockN{ 	Yixin Xu, Jianzhong Qi, Renata Borovica-Gajic, Lars Kulik}
		\IEEEauthorblockA{\textit{School of Computing and Information Systems} \\
			\textit{The University of Melbourne, Australia} \\
			yixinx3@student.unimelb.edu.au, \{jianzhong.qi,renata.borovica,lkulik\}@unimelb.edu.au}
	}
	
	\maketitle
	
	\begin{abstract}
		On-demand ride-sharing is rapidly growing.
		Matching trip requests to vehicles efficiently is critical for the service quality of ride-sharing. To match trip requests with vehicles, a prune-and-select scheme is commonly used. The pruning stage identifies feasible vehicles that can satisfy the trip constraints (e.g., trip time). The selection stage selects the optimal one(s) from the feasible vehicles. The pruning stage is crucial to reduce the complexity of the selection stage and to achieve efficient matching. We propose an effective and efficient pruning algorithm called GeoPrune. GeoPrune represents the time constraints of trip requests using circles and ellipses, which can be computed and updated efficiently. Experiments on real-world datasets show that GeoPrune reduces the number of vehicle candidates in nearly all cases by an order of magnitude and the update cost by two to three orders of magnitude compared to the state-of-the-art.
	\end{abstract}
	
	
	\section{Introduction}
	Ride-sharing is becoming a ubiquitous transportation means in our daily lives.
	In August 2018, there were 436,000 Uber rides and 122,000 Lyft rides  daily in New York~\cite{Market}.
	The growing number of rides calls for efficient algorithms to match numerous trip requests to optimal vehicles in real-time.
	
	Matching trip requests to vehicles is commonly referred to as the \emph{dynamic ride-sharing matching} problem~\cite{ma2013t},\cite{tong2018unified}.
	The goal is to assign each trip request to a vehicle such that a given optimization objective is achieved while satisfying the service constraints of trip requests (such as the waiting time and detour time).
	Various optimization goals have been proposed in the literature, such as
	minimizing the total travel distance of vehicles~\cite{ma2013t},\cite{thangaraj2017xhare},\cite{tong2018unified},\cite{huang2014large}, 
	maximizing the number of served requests~\cite{alonso2017demand}, and maximizing the system profit~\cite{asghari2016price},\cite{ zheng2018order}.
	
	To find matches for trip requests, existing algorithms typically employ two stages: \textbf{pruning} and \textbf{selection}.
	The pruning stage filters out infeasible vehicles that cannot meet the service constraints of trip requests, e.g., vehicles that are too far away.
	From the remaining vehicles, the selection stage selects the optimal vehicles and adds the new trip requests to their routes.
	The computation time of the selection stage largely depends on the effectiveness of the pruning stage (i.e., the number of remaining vehicles) as it usually requires exhaustive checks on all remaining vehicles with respect to the optimization goal. 
	The pruning algorithm is thus crucial for both the efficiency of the selection stage and the overall matching efficiency.
	
	%
	
	In this work, we study how to efficiently prune infeasible vehicles for fast matching. We focus on finding vehicles that satisfy the service constraints of trip requests rather than any particular optimization goal. Thus, our solution is generic and can be easily integrated into existing selection algorithms  for various optimization goals. We consider two service constraints of trip requests: 
	the \textit{latest pickup time} and the \textit{latest drop-off time}.
	Vehicles violating these constraints are infeasible matches and are filtered out in the pruning stage.
	
	Pruning infeasible vehicles in real-time is challenging in many aspects.
	First, ride-sharing is a highly dynamic process. 
	New requests are arriving frequently and vehicles are moving continuously.
	A pruning algorithm has to not only effectively prune infeasible vehicles but also quickly update any information needed for future pruning. 
	Second, the pruning process needs to consider the constraints of not only the new trip request but also the trip requests that are currently being served by the vehicles. Checking all these constraints poses significant challenges to the algorithm efficiency.  
	
	
	Existing pruning algorithms maintain dynamic indices over the road network. A simple pruning strategy is to partition the road network space into grid cells and dynamically record the grid cell where each vehicle resides. To match a trip request, only the vehicles in the nearby grid cells of the trip request source location need to be examined~\cite{tong2018unified}. 
	Such a strategy finds vehicles that satisfy the latest pickup time constraint but does not consider the drop-off time. Thus, it may return many infeasible vehicles.
	To obtain a higher efficiency, two approximate algorithms were proposed, namely, \emph{Tshare}~\cite{ma2013t} and \emph{Xhare}~\cite{thangaraj2017xhare}.
	Tshare precomputes pair-wise distances between grid cells and records the cells on the route of each vehicle.
	To match a trip request, Tshare~\cite{ma2013t} checks the cells within a distance threshold of the request source/destination and retrieves vehicles passing these cells in a certain time range. 
	Xhare on the other hand clusters the road network and records reachable clusters for vehicles given the time constraints. To match a trip request, Xhare returns all vehicles that can make a detour to the cluster where the request source/destination resides. 
	Both algorithms may fail to find all feasible vehicles due to approximation errors such as in distance estimation, and their indices may have high storage and update costs.
	
	
	\begin{figure}[t]
		\centering
		\begin{overpic}[scale=0.8, trim=10 40 10 8,clip]{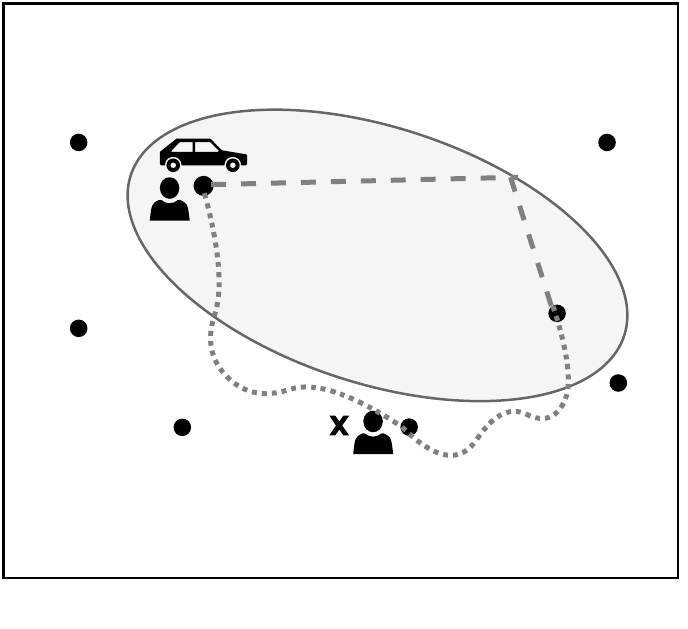}
			\put(28,55){$r_1$}
			\put(75,2){$r_2$}
			\put(45,80){$c$}
			\put(45,62){\large $s$}
			\put(92,12){\large $p'$}
			\put(125,39){\large $e$}
			\put(112, 78){\large $p$}
		\end{overpic}
		\caption{An example of our main idea.}
		\label{fig:example}
	\end{figure}
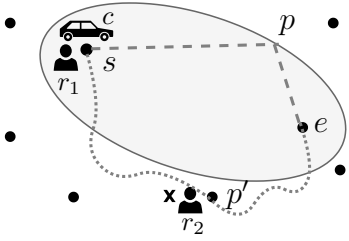

	To overcome the limitations above, we propose new pruning strategies based on geometric properties of service constraints. Our strategies are built upon the following intuition. As shown in Figure~\ref{fig:example}, consider a vehicle $c$ that has been assigned to a request $r_1$ with a trip from point $s$ to point $e$. 
	Vehicle $c$ is now at $s$ and needs to reach $e$ within time $t_1$  (e.g., $t_1$ minutes), as constrained by $r_1$'s latest drop-off time. 
	To satisfy the time constraint $t_1$, $c$ can visit a point $p$ on its way to $e$ only if
	$dist(s, p)/v_{max} + dist(p, e)/v_{max} \leq t_1$, where \textit{dist} is the Euclidean distance between two points, and $v_{max}$ is the maximum vehicle speed.
	Obviously, the vehicle may need to travel longer than the Euclidean distance as its movement is constrained by the road network.
	Also, it may not be able to always travel at the maximum speed.
	Thus, even if point $p$ satisfies this inequality,  vehicle $c$ may still be unable to visit $p$.
	On the other hand, \emph{if point $p$ does not satisfy this inequality, $c$ must not visit $p$.}
	The above inequality defines an ellipse as shown in the figure. Any point outside this ellipse will violate the inequality and must not be visited by $c$.
	Therefore, if there is another ride-sharing request $r_2$ 
	from a different user at point $p'$, 
	we can safely prune $c$ from consideration if $p'$ 
	is not in the ellipse of $c$. This forms the basis of our pruning strategies.

	Following the idea above, we propose an efficient 
	\emph{geometry-based pruning} algorithm (\textbf{GeoPrune}) for ride-sharing.
	Our algorithm represents the service constraints of vehicles and requests using geometric objects.
	The regular and closed shape of these objects further enables us to index them using efficient data structures such as R-trees for fast search and updates.
	For every new trip request, we return the pruning results by applying several point/range queries on the R-trees. Among the candidates, the optimal one is computed and returned with a separate selection algorithm satisfying the optimization goal.
	Once a trip request is assigned to a vehicle, we insert its source and destination to the vehicle route.
	Experimental results show that \ouralgorithm can prune most infeasible vehicles, which substantially reduces the computational costs of the selection stage and improves the overall matching efficiency.
	
	Our main contributions are as follows:
	\begin{itemize}
		\item We propose novel pruning strategies to filter  infeasible vehicles for trip requests.
		Our pruning strategies are based on geometric properties,  which eliminate expensive precomputation and update costs,  making them suitable for large networks and highly dynamic scenarios.
		
		\item Based on the pruning strategies, we propose an algorithm named \ouralgorithmEnd\ that can filter out most infeasible vehicles. It   significantly reduces the computational costs of the selection stage and the overall matching process.

		\item 
		Our theoretical analysis shows that the running time of \ouralgorithm is $O(\sqrt{|S||C|} +  |S||C|\log(|S||C|))$, where $|S|$ is the maximum number of stops of the vehicle schedules and $|C|$ is the number of vehicles.
		\ouralgorithmCapital takes $O(|S|\log^2(|S||C|))$ time to update the states for a newly assigned trip request. During every time slot, \ouralgorithm takes $O(|S|\log(|S||C|) + |C|\log^2|C|)$ time to update for moving vehicles. 
		
		\item Experiments on real datasets show that \ouralgorithm improves the matching efficiency by reducing the number of potential vehicles in nearly all cases by an order of magnitude and reducing the update time by two to three orders of magnitude compared with the state-of-the-art.
	\end{itemize}
	
	\section{Preliminaries}
	We first present basic concepts and a problem formulation 
	for our targeted ride-sharing matching problem. 
	
	\subsection{Definitions}\label{subsec:definition}
	
	We consider ride-sharing on 
	a road network that is represented as a directed 
	graph $G=\langle N, E  \rangle$, where $N$ is a set of vertices and $E$ is a set of edges. 
	Each edge $e(n_i,n_j)$ is associated with weight $d(n_i,n_j)$ that indicates the travel distance between vertices $n_i$ and $n_j$. 
	We denote an estimated shortest travel time between $n_i$ and $n_j$ as $t(n_i, n_j)$ (which may be calculated based on the shortest-path distance or fetched from a navigation service such as Google Maps).
	
	\textbf{Trip request.}  A \textit{trip request} $r_i=\langle t, s,e,w,\epsilon,\eta \rangle$ consists of six elements: the issue time $t$, the source location $s$, the destination location $e$, the maximum waiting time $w$, the maximum detour ratio $\epsilon$, and the number of passengers $\eta$. A set of trip requests is represented as $R=\{r_1, r_2,...,r_n\}$.
	
	For a trip request $r_i$, the issue time $r_i.t$ records the time when the trip request is sent. 
	The maximum waiting time $r_i.w$ limits the \emph{latest pickup time} of the request to be $r_i.lp=r_i.t + r_i.w$. The maximum detour ratio $r_i.\epsilon$ limits the extra detour time of the request. Together with the maximum waiting time, it constraints the \emph{latest drop-off time} of the request to be $r_i.ld = r_i.t + r_i.w + t(s,e) \times (1 + \epsilon)$. 
	Alternatively, a request can directly set the latest pickup and latest drop-off times. 
	The difference between the latest drop-off time and the 
	issue time, i.e., $r_i.ld-r_i.t$, is the maximum allowed travel time of $r_i$.
	\begin{example}
		Assume two trip requests $r_1=\langle$9:00\,am, $s_1$, \\$e_1$, 5\,min, 0.2, 1$\rangle$ and $r_2=\langle$9:07\,am, $s_2$, $e_2$, 5\,min, 0.2, 1$\rangle$ in\\ Figure~\ref{fig:vehicle_schedule}.
		The shortest travel times from $s_1$ to $e_1$ and 
		from $s_2$ to $e_2$, i.e., $t(s_1, e_1)$ and $t(s_2, e_2)$, are both 15\,min. Then, the time constraints of $r_1$ and $r_2$ are:
		$r_1.lp$=9:00\,am+5\,min=9:05\,am,
		$r_2.lp$=9:07\,am+5\,min=9:12\,am,
		$r_1.ld$=9:05\,am+15\,min$\times$1.2\\=9:23\,am,  $r_2.ld$=9:12\,am+15\,min$\times$1.2=9:30\,am.
	\end{example}

	\begin{table}[t]	
		\centering
		\caption{Frequently Used Symbols}
		\label{tab:notations}
		\begin{tabular}{|l|l|}
			\hline
			Notation & Description \\ \hline
			$G=\langle N, E  \rangle$ & \makecell[tl]{a road network with a set of\\ vertices $N$ and a set of edges $E$} \\ \hline
			$t(n_i,n_j)$ & \makecell[tl]{the  estimated shortest travel time between \\ vertices  $n_i$ and $n_j$} \\ \hline
			$R=\{r_i\}$ & a set of trip requests        \\ \hline
			$C=\{c_j\}$ & a set of vehicles       \\ \hline
			$r_i=\langle t, s,e,w,\epsilon,\eta \rangle$ & \makecell[tl]{a trip request issued at time $t$ with source $s$, \\destination $e$, maximum waiting time $w$,\\ maximum detour ratio $\epsilon$ and $\eta$ passengers}      \\\hline			
			$r_i.lp, r_i.ld$  & the latest pickup and  drop-off times of $r_i$ \\ \hline
			$r_i.wc, r_i.rd$ & the waiting circle and the detour ellipse of $r_i$ \\ \hline
			$c_j=\langle l, S, u, v \rangle$ & \makecell[tl]{a vehicle at $l$ with planned trip\\ schedule $S$, capacity $u$ and traveling speed $v$}  \\ \hline
			$(p^{k-1}, p^{k})$ & the segment between $p^{k-1}$ and $p^{k}$      \\ \hline
			$vd[k]$ & the detour ellipse of $(p^{k-1}, p^{k})$ \\ \hline
		\end{tabular}
	\end{table}

	\textbf{Vehicle.} A \textit{vehicle} $c_i$ is represented as $c_i=\langle l,S,u, v \rangle$, where $l$ denotes the location of the vehicle, $S$ represents the \emph{trip schedule} of the vehicle (which will be detailed in the next subsection), $u$ is the vehicle capacity, and $v$ is the travel speed.
	We use $C=\{c_1,c_2,...,c_n\}$ to denote a set of vehicles.
	
	We track the occupancy status of the vehicles, which is updated at every system time point.
	A vehicle is \emph{empty} if it has not been assigned to any trip requests. Otherwise, the vehicle is \textit{non-empty} and needs to follow their trip schedules to serve trip requests assigned to them.
	
	\subsection{Vehicle schedule}
	
	\textbf{Trip schedule.} The trip schedule of a vehicle $c_i$,   $c_i.S=\{p^0,p^{1},p^2,...,p^{m}\}$, is a sequence of source or destination locations (points on the road network) of trip requests, except for $p^0$ which records the current location of the vehicle, i.e., $p^0=c_i.l$.  We call a source or destination location on a trip schedule a \emph{stop}, and the path between every two adjacent stops $p^{k-1}$ and $p^{k}$ a \emph{segment}, denoted as $(p^{k-1}, p^{k})$.
	
	\begin{example}
		Figure~\ref{fig:vehicle_schedule} shows an example trip schedule for a vehicle.
		The current time is 9:00\,am and the vehicle is at $l$.
		There are two trip requests, $r_1$ and $r_2$, assigned to the vehicle and the vehicle schedule is $(l, r_1.s, r_2.s, r_1.e, r_2.e)$. 
	\end{example}
	
	
	\textbf{Trip schedule recorder.}
	We follow a previous study~\cite{tong2018unified} and record the \emph{estimated arrival time}, \emph{latest arrival time}, and \emph{slack time} of $c_i.S$ with three arrays $aar[]$, $ddl[]$, and $slk[]$:	
	\begin{itemize}
		\item 	Estimated arrival time 
		$arr[k]$ records the estimated arrival time to stop $p^k$ via the trip schedule.
		
		\item Latest arrival time $ddl[k]$ records the latest acceptable arrival time at the stop $p^{k}$. 
		If $p^k$ is the pickup point of a request $r_j$, $ddl[k]$ is the latest pickup time of $r_j$, i.e., $ddl[k] = r_j.lp$.
		If $p^k$ is the drop-off point of $r_j$,  $ddl[k]$ is the latest drop-off time of $r_j$, i.e., $ddl[k] = r_j.ld$.
		
		\item Slack time $slk[k]$ records the maximum extra travel time allowed between $(p^{k-1},p^k)$ to satisfy the latest arrival time of $p^k$ and all stops scheduled after $p^k$.
		For stop $p^i$, it only allows $ddl[i]-arr[i]$ detour time to ensure its latest arrival time $ddl[i]$.
		A detour between $p^{k-1}$ and $p^{k}$ will  not only affect the arrival time of $p^{k}$ but also that of all stops scheduled after $p^k$.
		Therefore, a detour between $p^{k-1}$ and $p^{k}$ must guarantee the latest arrival time of $p^k$ and all stops scheduled after $p^k$, i.e., $slk[k] = min\{ddl[i]-arr[i]\}, i = k,...,m$.
		$slk[k]$ can be calculated by referring to $slk[k+1]$, i.e., $slk[k]=min\{(ddl[k]-arr[k]), slk[k+1]\}$.
		The \textit{maximum allowed travel time} between ($p^{k-1}$, $p^k$) is thus $arr[k]-arr[k-1]+slk[k]$.
		
	\end{itemize}
	
	\begin{example}
		The arrays of the trip schedule in Figure~\ref{fig:vehicle_schedule} are shown in Table~\ref{tab:vehicleScheduleArray}.
		The estimated arrival time of the stops is computed based on the arrival time of previous stops and the travel time between stops, e.g., $arr[1]$=9:00\,am+3\,min= 9:03\,min, $arr[2]$
		=9:03\,am+5\,min= 9:08\,min.
		The latest arrival time of the stops is determined by the corresponding trip requests.
		The latest arrival time of $p^1$ is the latest pickup time of $r_1$, i.e., $ddl[1] = r_1.lp$=9:05\,am.
		The latest arrival time of $p^3$ is the latest drop-off time of $r_1$, i.e., $ddl[3] = r_1.ld$=9:23\,am.
		$ddl[k]-arr[k]$ represents the allowed detour time before visiting $p^k$ to ensure $ddl[k]$, e.g.,  $p^1$ allows  9:05\,am-9:03\,am=2\,mins detour before it and $p^2$ allows 9:12\,am-9:08\,am
		=4\,mins detour before it.
		$slk[k]$ records the minimum allowed detour time of $p^k$ and all stops after $p^k$, e.g., 
		a detour before $p^3$ will not only affect the arrival time of $p^3$ but also that of $p^4$. Thus, $slk[3]=min\{$5\,min,4\,min$\}$=4\,min.
		
	\end{example}

	\textbf{Valid trip schedule.} 
	To form a \emph{valid trip schedule}, the following trip constraints need to be satisfied:
	\begin{itemize}
		\item \textit{Point order constraint:} 
		Trip schedule $c_i.S$ must visit the pickup location $r_j.s$ before the drop-off location $r_j.e$, 
		for any trip request $r_j$ assigned to vehicle $c_i$. 
		
		\item  \textit{Time constraint.} 
		Trip schedule $c_i.S$ must meet the service constraints for every trip request $r_j$ assigned to vehicle $c_i$, i.e., 
		$r_j$ needs to be picked up before $r_j.lp$ and be dropped off before $r_j.ld$.
		
		\item  \textit{Capacity constraint.}
		At any time when $c_i$ is traveling with trip schedule $c_i.S$, 
		the number of passengers in the vehicle must be within the vehicle capacity.
	\end{itemize}
	
	\begin{figure}[t]
		\centering
		\includegraphics[width=0.95\linewidth]{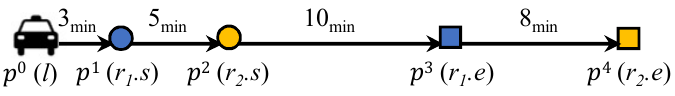}
		\caption{A vehicle schedule example at 9:00\,am.}
		\label{fig:vehicle_schedule}
	\end{figure}
	
	\begin{table}[t]
		\centering
		\caption{The arrays of the trip schedule in  Figure~\ref{fig:vehicle_schedule}.}
		\label{tab:vehicleScheduleArray}
		\begin{tabular}{|l|l|l|l|l|}
			\hline
			$p^k$ & $arr[k]$ & $ddl[k]$ &  $ddl[k]-arr[k]$ & $slk[k]$ \\ \hline
			$p^1$ & 9:03\,am   & 9:05\,am   & 2\,min & 2\,min      \\ \hline
			$p^2$ & 9:08\,am   & 9:12\,am   & 4\,min & 4\,min      \\ \hline
			$p^3$ & 9:18\,am   & 9:23\,am   & 5\,min & 4\,min       \\ \hline
			$p^4$ & 9:26\,am   & 9:30\,am   &   4\,min &  4\,min      \\ \hline
		\end{tabular}
	\end{table}
	
	\textbf{Feasible match.} Given a new trip request $r_n$, matching vehicle $c_i$ with $r_n$ (i.e., assigning $c_i$ to serve $r_n)$ is \emph{feasible} if adding $r_n$ into the trip schedule of $c_i$  yields a valid trip schedule. Vehicle $c_i$ is then a \emph{feasible vehicle} for $r_n$.
	
	Similar to the previous studies~\cite{tong2018unified},\cite{cheng2017utility},\cite{Xu2019Efficient}, we assume 
	the source and destination of the new trip request are inserted or appended to the current schedule of the matching vehicle.
	
	\subsection{Matching objective}
	
	\textbf{Problem definition.} 
	Given a road network $G$, a set of vehicles $C$, a set of trip requests $R$, and an optimization objective $O$, we aim to match every request $r \in R$ with a feasible vehicle $c \in C$ to optimize $O$.
	
	We examine a popular optimization objective, \emph{minimizing the total increased travel distance (time)}~\cite{ma2013t},\cite{thangaraj2017xhare},\cite{tong2018unified},\cite{huang2014large},\cite{alonso2017demand}.
	Suppose that the total travel time of the trip schedules of all vehicles is $T$ before assigning trip requests in $R$ and the total travel time  becomes $T'$ after assigning vehicles in $C$ to serve requests in $R$, our optimization goal $O$ is to minimize $T' - T$. 
	
	Minimizing the total increased distance for all trip requests is  NP-complete~\cite{ma2013t} and the future trip requests are unknown in advance.
	A common solution is to greedily assign each trip request to an optimal vehicle~\cite{ma2013t},\cite{tong2018unified},\cite{Xu2019Efficient},\cite{chen2018price}. 
	The trip requests are processed ordering by their issue time. 
	For every trip request, we assign it with a feasible vehicle such that the increased distance of the vehicle trip schedule is minimized.

	\subsection{Pruning and selection}
	We take a two-stage approach to solve the problem: 
	\begin{enumerate}
		\item  \textbf{Pruning.} Given a new request $r_n$, the pruning stage filters out infeasible vehicles and returns a set of vehicle candidates $C'$ for $r_n$.

		\item \textbf{Selection.} Given a set of vehicle candidates $C'$, the selection stage finds the optimal feasible vehicle in $C'$.
	\end{enumerate}

	In what follows, we focus on developing algorithms for the pruning stage.
	Observing that empty vehicles can be pruned by applying existing spatial network algorithms~\cite{liu2006dynamic},\cite{abeywickrama2016k},\cite{shen2017v},
	\cite{papadias2003query},
	we distinguish non-empty vehicles and empty vehicles and focus on pruning non-empty vehicles.
	
	
	\section{Geometric-based pruning}
	\label{sec:pruningStrategies}
	
	
	
	
	When a new trip request arrives, we find an optimal feasible vehicle and add the source and destination of the new trip request to the vehicle trip schedule. As discussed before, the trip schedule of the vehicle must satisfy the service constraints of all trip requests assigned to it including the new trip request. This is the basis of our pruning strategies. 
	
	There are two possibilities to add a stop to a trip schedule, either inserting it into a segment of the schedule or appending it to the end.
	For example, to add a new stop $p$ to the trip schedule in Figure~\ref{fig:vehicle_schedule}, we can either insert it to a segment to form a new schedule such as $(p^0, p, p^1, p^2, p^3, p^4)$ (we cannot insert before $p^0$ because  $p^0$ is the current location of the vehicle) or append it to the end where the schedule becomes $(p^0, p^1, p^2, p^3, p^4, p)$.
	We say that a stop is \textit{added} to a schedule if it is either inserted or appended to the schedule and the adding is \textit{valid} if it still generates a valid trip schedule.
	
	We first detail the criteria to determine whether adding the source or the destination of a new trip request is valid based on constraints of existing trip requests in the trip schedule and constraints of the new trip request, respectively. Then, we summarize these criteria into three pruning rules.


	\subsection{Constraints based on existing trip requests}\label{subsec:vehicleServiceConstraint}
	
	
	Given a segment $(p^{k-1}$, $p^{k}$), 
	if we insert a new stop $p$ to it, the path from $p^{k-1}$ to $p^{k}$ becomes ($p^{k-1} , p , p^{k}$). The travel time from $p^{k-1}$ to $p^{k}$ becomes $t(p^{k-1}, p) + t(p, p^{k})$, which must be no larger than the maximum allowed travel time of the segment $arr[k]-arr[k-1]+slk[k]$ to satisfy the constraints of exiting trip requests.
	
	The maximum allowed travel time limits the area that the vehicle can reach between $p^{k-1}$ and $p^{k}$. 
	Our key observation is that such a reachable area can be bounded using an \textbf{ellipse} $vd[k]$, and we call it the \textit{detour ellipse}  of the segment.
	
	\begin{definition}
		The \emph{detour ellipse $vd[k]$ of a segment} $(p^{k-1}, p^{k} )$ is an ellipse with $p^{k-1}$ and $p^{k}$ as its two focal points, and the major axis length $vd[k].major$ equals to the maximum allowed travel time multiplied by the vehicle speed $v$, i.e., 
		$vd[k].major = (arr[k] - arr[k-1] + slk[k]) \cdot v$
	\end{definition}

	\begin{figure}[t]
		\centering
		\includegraphics[width=0.9\linewidth]{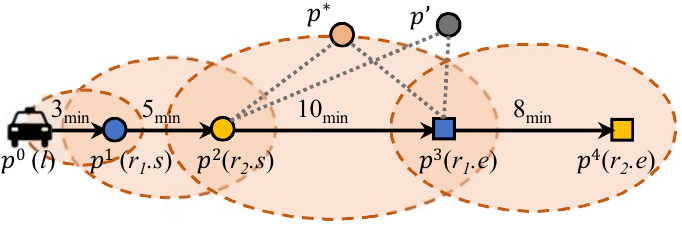}
		\caption{Detour ellipses of the trip schedule in Figure~\ref{fig:vehicle_schedule}.}
		\label{fig:vehicle_ellipse}
	\end{figure}

	\begin{lemma}\label{lemma:segmentDetourConstraint}
		For a segment $(p^{k-1}, p^{k})$, if a point $p$ is outside of $vd[k]$, $t(p^{k-1}, p) + t(p, p^k)$ will exceed the maximum allowed travel time. The segment is therefore invalid for inserting $p$.
	\end{lemma}
	
	\vspace{-0.6em}
	\begin{proof}
		According to the definition of ellipses, if a point $p$ is outside of the ellipse, the sum of the Euclidean distances $|p^{k-1}p|+|pp^k|$ must be greater than $vd[k].major$.
		Since any road network distance between two points is no smaller than their Euclidean distance (triangle inequality), the sum of road network distances $d(p^{k-1},p)+ d(p,p^k)$ is at least as large as $|p^{k-1}p|+|pp^k|$ and thus must also be greater than $vd[k].major$.
		The time required to travel such a distance thus exceeds the maximum allowed travel time and violates the latest arrival time of existing stops.
	\end{proof}
	\vspace{-0.6em}

	\begin{example}
		Figure~\ref{fig:vehicle_ellipse} shows the detour ellipses of the trip schedule illustrated in Figure~\ref{fig:vehicle_schedule}.
		For segment $(p^2, p^3)$, the slack time  is 4\,min and thus the maximum allowed travel time from $p^2$ to $p^3$ is 10\,min + 4\,min = 14\,min.
		We make an ellipse with $p^2$ and $p^3$ as the two focal points and the major axis length being 14\,min  multiplied by the vehicle speed, i.e., $|p^2p^*|+|p^*p^3|=$ (14\,min  $\cdot v$) for a point $p^*$ on the ellipse.
		If a point $p'$ is outside this ellipse, then the Euclidean distance $|p^2p'|+|p'p^3| >$ (14\,min  $\cdot v$).
		Thus, the road network distance $d(p^2,p')+d(p',p^3)$ will also be greater than (14\,min  $\cdot v$) and the corresponding travel time with speed $v$ will exceed 14\,min, which violates the service constraint of exiting trip requests. Therefore, it is invalid to insert $p'$ between $(p^2, p^3)$.
	\end{example}
	
	Most existing ride-sharing matching algorithms do not consider the variations in traffic assuming a constant travel speed~\cite{pan2019ridesharing}.
	To remove this constraint, we replace the constant speed assumption with a \textit{maximum} speed when computing the ellipses for the vehicle and the requests. This enables our approach to avoid false negatives if vehicles travel at varying speeds: all feasible vehicles are kept (Lemma~\ref{lemma:segmentDetourConstraint} holds) as long as they do not exceed the maximum speed. 
	We later show that using the maximum speed still preserves pruning efficiency.
	
	The ellipse construction is independent of the vehicle trajectories.  It only relies on the maximum allowed travel time and the endpoints of its segment.
	We record the ellipses of vehicles and update them only if the corresponding segments change.
	Specifically, when a trip request is newly assigned to a vehicle, we update the trip schedule of the vehicle and recalculate its ellipses.
	Meanwhile, when the vehicles are moving, some stops (and their segments) may already be visited and become obsolete.
	We remove the ellipses of such obsolete segments.

	\subsection{Constraints based on the new request}\label{subsec:requestServiceConstraint}
	Next, we analyze the service constraints of new requests.

	\textbf{Latest pickup time constraint.}\label{subsubsec:requestWaitingConstraint}
	Recall that $r_n.w$ denotes the maximum waiting time to ensure the latest pickup time of the new request $r_n$. We define a \emph{waiting circle} with $r_n.w$.
	
	\begin{definition}
		The \emph{waiting circle of $r_n$}, denoted by $r_n.wc$, is a circle 
		centered at $r_n.s$ and with $r_n.w\cdot v$ as its radius.	
	\end{definition}

	\begin{lemma}\label{lemma:waitingConstraint}
		If it is valid to add $r_n.s$ after a stop $p^{k}$ in $c_i.S$, then $p^{k}$ and all stops before $p^{k}$ must be covered by  $r_n.wc$.
	\end{lemma}
	
	\vspace{-0.6em}
	\begin{proof}
		The waiting circle bounds the area that a vehicle can reach before picking up $r_n$ to ensure the latest pickup time of $r_n$.
		Points outside of $r_n.wc$ have Euclidean distances (and hence network distances) to $r_n.s$ greater than $r_n.w \cdot v$.
		If a vehicle is scheduled to visit a point outside of $r_n.wc$ before reaching $r_n.s$, the vehicle cannot pickup $r_n$ before the latest pickup time $r_n.lp$ and thus violates the constraint of $r_n$.
	\end{proof}
	\vspace{-0.6em}

	\begin{example}
		Figure~\ref{fig:circle_example} shows the waiting circle of a new request $r_n$. The source  
		$r_n.s$ can only be added after the stops in the waiting circle $r_n.wc$, i.e., $p^0$ or $p^1$.
		If the vehicle visits $p^2$ (outside of the waiting circle) before $r_n.s$, it will not pick up $r_n$ before the latest pickup time of $r_n$.
		Thus, it is invalid to add $r_n.s$ after $p^2$ or any   stops afterwards, i.e.,  $p^3$ and $p^4$.
	\end{example}
	
	\begin{figure}
		\centering
		\includegraphics[width=0.9\linewidth]{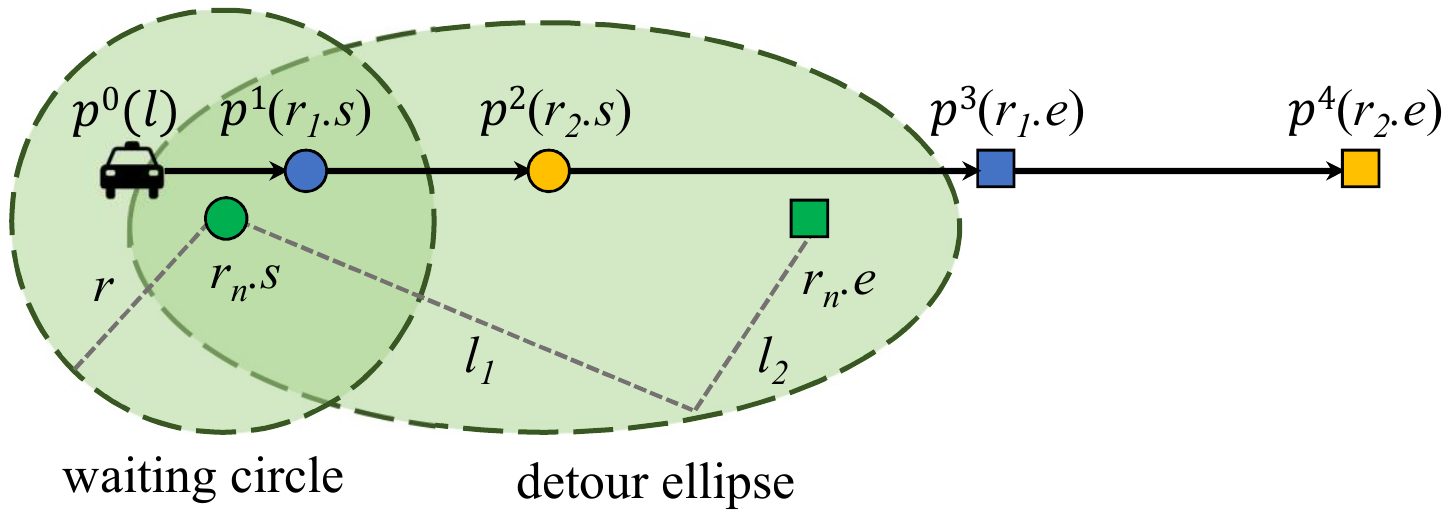}
		\caption{Waiting circle and detour ellipse of $r_n$, $r = r_n.w \cdot v$, $l_1+l_2 = (r_n.ld-r_n.t) \cdot v$.}
		\label{fig:circle_example}
	\end{figure}
	
	\textbf{Latest drop-off time constraint.}\label{subsubsec:requestDetourConstraint}
	Similar to the detour ellipses of segments, 
	we define a detour ellipse for a new request $r_n$ to ensure the latest drop-off time of $r_n$.
	
	\begin{definition}
		The \emph{detour ellipse $r_n.rd$ of a new trip request} $r_n$ is an ellipse with  $r_n.s$ and $r_n.e$ as the two focal points. The major axis length is the maximum allowed travel time of $r_n$ multiplied by the speed $v$, i.e., $r_n.rd.major = (r_n.ld-r_n.t) \cdot v$.
	\end{definition}
	
	The detour ellipse of $r_n$ restricts the area that a vehicle can visit while serving $r_n$.
	After picking up $r_n$ (reaching $r_n.s$), if the vehicle is scheduled to visit any stop outside of the detour ellipse of $r_n$, it will not be able to reach the destination $r_n.e$ before the latest drop-off time $r_n.ld$. 
	
	\begin{lemma}\label{lemma:detourConstraint}
		Let $r_n.s $ be added after stop $p^{s}$ in the trip schedule $c_i.S$ of a vehicle $c_i$. If it is valid to add $r_n.e$ after $p^k$ in $c_i.S$, then $p^k$ and all stops scheduled between $p^{s}$ and $p^k$ must be covered by $r_n.rd$.
	\end{lemma}
	
	\begin{example}
		The detour ellipse of $r_n$ is shown in Figure~\ref{fig:circle_example}.
		If $r_n.s$ is added after $p^0$, then $r_n.e$ can only be added after either $p^0$ or stops inside of the detour ellipse, i.e., $p^1$ and $p^2$.
		Adding $r_n.e$ after later stops (e.g., $p^3$) will violate the latest drop-off time of $r_n$.
	\end{example}

	\begin{figure}
		\centering
		\begin{subfigure}{0.4\textwidth}
			\centering
			\includegraphics[scale =0.5]{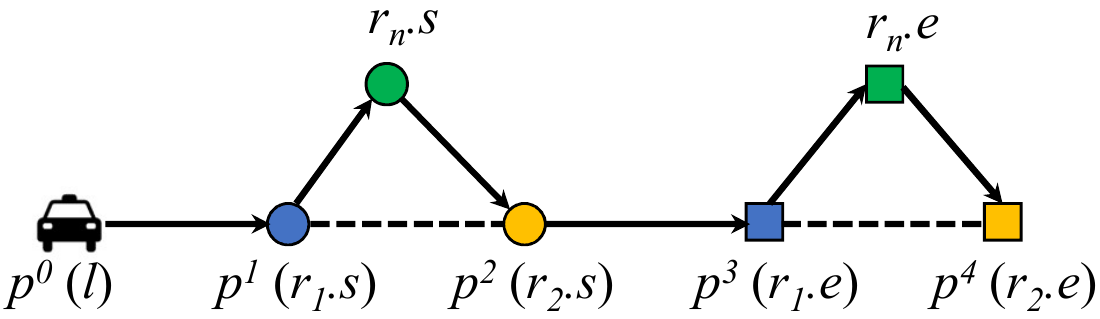}
			\caption{insert source and insert destination.}
			\label{fig:insert_insert}
		\end{subfigure}
		\vspace{1mm}
		\begin{subfigure}{0.4\textwidth}
			\centering
			\includegraphics[scale =0.5]{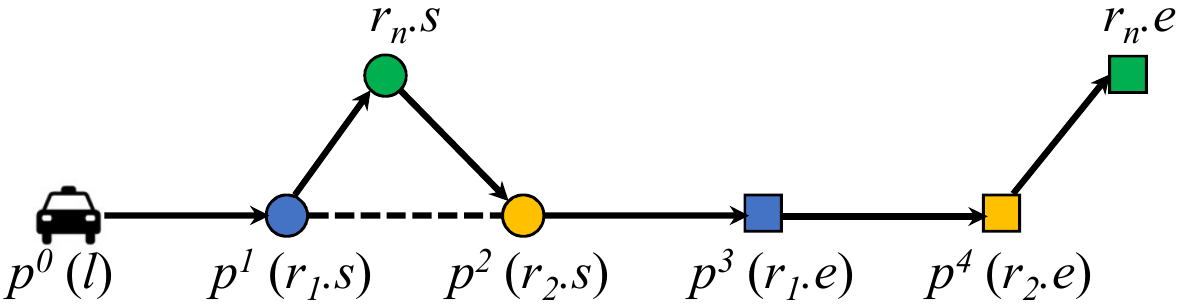}
			\caption{insert source and append destination.}
			\label{fig:insert_append}
		\end{subfigure}
		\vspace{1mm}
		\begin{subfigure}{0.4\textwidth}
			\centering
			\includegraphics[scale =0.5]{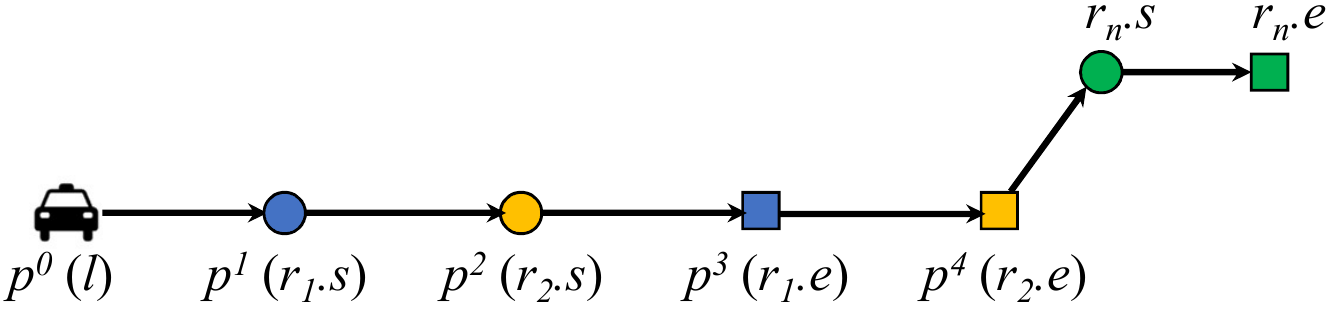}
			\caption{append source and append destination.}
			\label{fig:append_append}
		\end{subfigure}
		\caption{Cases to add a new trip request to a trip schedule}
		\label{fig:insertionCases}
	\end{figure}
	
	\subsection{Pruning Rules}
	\label{subsec:threeConditions}
	There are three cases shown in Figure~\ref{fig:insertionCases} when adding a new trip request $r_n$ to the trip schedule $c_i.S$
	of a vehicle $c_i$: 
	\begin{enumerate}
		\item \textbf{insert-insert:} insert $r_n.s$ into a segment of $c_i.S$ and insert $r_n.e$ into the same or another segment of $c_i.S$.
		\item \textbf{insert-append:} insert $r_n.s$ into a segment of $c_i.S$ and append $r_n.e$ to the end of $c_i.S$
		\item \textbf{append-append:} append both $r_n.s$ and $r_n.e$ to the end of  $c_i.S$. 
	\end{enumerate}
	
	We next analyze the conditions that $c_i$ needs to satisfy so that 
	adding $r_n$ to $c_i.S$ is valid for each case. 
	
	\textbf{Insert-insert.} Figure~\ref{fig:insert_insert} illustrates the insert-insert case, where both $r_n.s$ and $r_n.e$ are inserted into some segments of the trip schedule $c_i.S$. 
	According to Lemma~\ref{lemma:waitingConstraint}, a segment is valid for inserting a stop only if the stop is inside the detour ellipse of the segment. Therefore, both $r_n.s$ and $r_n.e$ must be inside the 
	detour ellipse of at least one segment of $c_i.S$.

	A special case is to insert both $r_n.s$ and $r_n.e$ to the same segment of $c_i.S$, as shown in Figure~\ref{fig:insert_two}. In this case, both $r_n.s$ and $r_n.e$ must be inside the detour ellipse of the segment.
	
	\begin{figure}[t]
		\centering
		\includegraphics[width=0.5\linewidth]{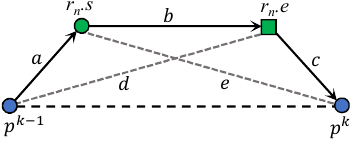}
		\caption{The special case of insert-insert.}
		\label{fig:insert_two}
	\end{figure}
	
	\begin{lemma}\label{lemma:insert_two_case}
		A segment $(p^{k-1}, p^{k})$ is valid to insert both $r_n.s$ and $r_n.e$ \emph{only if} $r_n.s$ and $r_n.e$ are both included in the detour ellipse of the segment $vd[k]$.
	\end{lemma}

	\begin{proof}
		We use Figure~\ref{fig:insert_two} to illustrate our proof, where the Euclidean distances among the stops are represented by $a, b, c, d, e$.
		Suppose that $(p^{k-1}, p^{k})$ is valid to insert both $r_n.s$ and $r_n.e$, and the trip schedule becomes  $(p^{k-1}, r_n.s, r_n.e, p^{k})$ after the insertion. 
		Traveling between $(p^{k-1}, p^{k})$ must satisfy the maximum allowed travel time constraint.
		Thus, $d(p^{k-1}, r_n.s) + d(r_n.s, r_n.e)+d(r_n.e, p^{k}) = (t(p^{k-1}, r_n.s) + t(r_n.s, r_n.e)+t(r_n.e, p^{k})) \cdot v \leq (arr[k] - arr[k-1] + slk[k]) \cdot v_{max} = vd[k].major$. 
		Since the Euclidean distance between two stops is no larger than their road network distance, $a+b+c \leq d(p^{k-1}, r_n.s) + d(r_n.s, r_n.e)+d(r_n.e, p^{k}) \leq vd[k].major$.
		According to the triangle inequality, $e<b+c$.
		Thus, $a+e<a+b+c \leq vd[k].major$. The  Euclidean distance sum from $r_n.s$ to $p^{k-1}$ and $p^{k}$ is smaller than $vd[k].major$ and $r_n.s$ must be inside $vd[k]$.
		Similarly, $d<a+b$, and $d+c$ $<a+b+c \leq vd[k].major$. $r_n.e$ must be inside $vd[k]$.
	\end{proof}

	The pruning rule for the insert-insert case is as follows:
	\begin{lemma}\label{lemma:insert_insert_condition}
		A vehicle $c_i$ may be matched with $r_n$ in the insert-insert case \emph{only if} it satisfies: 
		\begin{itemize}
			\item there exists a segment of $c_i.S$ with the detour ellipse that covers $r_n.s$, i.e., $r_n.s \in vd[k]$, $k = 1,...,m$; and 
			\item there exists a segment of $c_i.S$ with the detour ellipse that covers $r_n.e$, i.e., $r_n.e \in vd[k]$, $k = 1,...,m$.
		\end{itemize}
	\end{lemma}
	
	\textbf{Insert-append.} Figure~\ref{fig:insert_append} illustrates the insert-append case. 
	According to Lemma~\ref{lemma:segmentDetourConstraint}, to insert $r_n.s$, there must be 
	a segment in the trip schedule of $c_i$ that has a detour ellipse covering 
	$r_n.s$. Meanwhile, any stop between $r_n.s$ and $r_n.e$ needs to be covered by the detour ellipse of $r_n$
	(see Lemma~\ref{lemma:detourConstraint}).
	
	Checking all the stops between $r_n.s$ and $r_n.e$ against the detour ellipse of $r_n$ is non-trivial. 
	For fast pruning, we only check the ending stop of the current trip schedule: if the ending stop is outside of the detour ellipse of $r_n$, it is invalid for  appending $r_n.e$. 
	Take Figure~\ref{fig:insert_append} as an example. We only check if $p^4$ is inside the detour ellipse of $r_n$.
	This simplified rule may bring in a small number of infeasible vehicles, which will be filtered later as explained in the next paragraphs. 
	The pruning rule for the insert-append case is as follows:
	
	\begin{lemma}\label{lemma:insert_append_condition}
		A vehicle $c_i$ may be matched with $r_n$ in  the insert-append case \emph{only if} it satisfies: 
		\begin{itemize}
			\item there exists a segment of $c_i.S$ with the detour ellipse that covers $r_n.s$, i.e., $r_n.s \in vd[k]$, $k = 1,...,m$; and 
			\item the ending stop of the vehicle schedule, $p^{m}$, is covered by the detour ellipse of $r_n$, i.e., $p^{m} \in r_n.rd$.
		\end{itemize}
	\end{lemma}
	
	\textbf{Append-append.} Figure~\ref{fig:append_append} illustrates the append-append case, where we append both $r_n.s$ and $r_n.e$ to the end of the trip schedule.
	In this case, $r_n$ will not affect any exiting stops. Only the service constraints of $r_n$ need to be considered. 
	Furthermore, no stop is scheduled between $r_n.s$ and $r_n.e$, and hence the detour constraint of $r_n$ is satisfied already.
	The only constraint to check is the waiting time constraint of $r_n$.
	According to Lemma~\ref{lemma:waitingConstraint}, all stops scheduled before $r_n.s$ must be covered by the waiting circle of $r_n$.
	For example, in Figure~\ref{fig:append_append}, the vehicle needs to visit $p^0, p^1, p^2, p^3, p^4$ before picking up $r_n.s$.
	Therefore, all these stops should be covered by the waiting circle of $r_n$.
	Similar to the insert-append case, we only check the ending stop, as summarized in Lemma~\ref{lemma:append_append_condition}.
	\begin{lemma}\label{lemma:append_append_condition}
		A vehicle $c_i$ may be matched with $r_n$ in the append-append case \emph{only if}  the ending stop of its trip schedule, $p^{m}$, is covered by the waiting circle of $r_n$, i.e., $p^m \in r_n.wc$.
	\end{lemma}

	In our implementation, we use minimum bounding rectangles (MBRs) to represent ellipses and circles as they are easier to operate on and tightly bound the ellipses and circles.
	
	

	\subsection{Applying the Pruning Rules}\label{sec:applyPruning}
	When a new trip request $r_n$ arrives, we first compute the waiting circle and the detour ellipse of $r_n$. Then, we compute a set of vehicle candidates that may match $r_n$ based on the pruning rules above (Lemmas~\ref{lemma:insert_insert_condition}, \ref{lemma:insert_append_condition}, \ref{lemma:append_append_condition}).
	
	To facilitate the pruning, we compute sets of vehicles that:
	\begin{enumerate}
		\item have trip schedule segments with detour ellipses that cover $r_n.s$ (for the insert-insert and insert-append cases);
		\item have trip schedule segments with detour ellipses that cover $r_n.e$ (for the insert-insert case);
		\item have the ending stop of the trip schedule covered by $r_n.wc$ (for the append-append case); 
		\item have the ending stop of the trip schedule covered by $r_n.rd$ 
		(for the insert-append case).
	\end{enumerate}
	
	To find vehicles that satisfy a pruning rule, we just need to join  
	the relevant sets of vehicles computed above. 
	For example, vehicles that may satisfy the insert-insert case are those 
	in both the first and the second sets above.
	
	\textbf{R-tree based pruning.} We build two R-trees~\cite{guttman1984r} to accelerate the computation process,
	although other spatial indices may also be applied.
	One R-tree stores the detour ellipses of all segments for all vehicle trip schedules, denoted by $T_{seg}$; 
	the other R-tree stores the location of the ending stops of all non-empty vehicles, denoted as $T_{end}$.
	We run four queries:
	\vspace{-0.5mm}
	\begin{enumerate}
		\item $Q_1=T_{seg}.pointQuery (r_n.s)$ is a point query that returns all segments whose detour ellipses cover $r_n.s$; each segment returned may be used to insert $r_n.s$.
		\item $Q_2=T_{seg}.pointQuery (r_n.e)$ is a point query that returns all segments whose detour ellipses cover $r_n.e$; each segment returned may be used to insert $r_n.e$.
		\item $Q_3=T_{end}.rangeQuery(r_n.wc)$ is a range query that returns all ending stops covered by $r_n.wc$; 	each ending stop returned may be used to append $r_n.s$ and $r_n.e$.
		\item $Q_4=T_{end}.rangeQuery(r_n.rd)$ is a range query that returns all ending stops covered by $r_n.rd$; 
		each ending stop returned may be used to append $r_n.e$.
	\end{enumerate}
	
	
	The returned segments and ending stops are further pruned based on their time and capacity constraints.
	Specifically, for each segment $(p^{k-1}, p^{k})$ returned for inserting $r_n.s$ ($r_n.e$), we check whether the insertion violates the latest arrival time  of $p^{k}$ and $r_n.s$ ($r_n.e$). 
	The schedule between $(p^{k-1}, p^{k})$ becomes 
	($p^{k-1}, r_n.s (r_n.e),p^{k}$) after the insertion.
	For the new schedule, the arrival time of $r_n.s$ ($r_n.e$) and $p^{k}$ is estimated based on the arrival time of $p^{k-1}$ plus the travel time between them.
	If the estimated arrival time of $r_n.s (r_n.e)$ or $p^k$ exceeds their latest arrival time, the segment is  discarded.
	For each ending stop $(p^{m})$ returned for appending $r_n.s$ ($r_n.e$), 
	we estimate the arrival time of $r_n.s$ ($r_n.e$) with the appended schedule by summing up the end stop arrival time and the travel time from the end stop to $r_n.s$ ($r_n.e$). If the estimated time exceeds the latest arrival time of $r_n.s$ ($r_n.e$),
	we also discard the ending stop. 
	Besides the time constraint, we also check the capacity constraint for segments to insert $r_n.s$. If a segment $(p^{k-1},p^{k})$ is returned for inserting $r_n.s$, we sum up the number of passengers carried in $(p^{k-1},p^{k})$  and that of $r_n$ and discard the segment if the sum exceeds the capacity.

	Let the sets of vehicles corresponding to the segments and ending stops 
	returned by the four queries above (after filtering) be 
	$O_1$, $O_2$, $O_3$ and $O_4$, respectively.
	The set of vehicles satisfying the three pruning cases are:  $F_1 = O_1 \cap O_2$ (insert-insert);
	$F_2 = O_1 \cap O_4$ (insert-append); $F_3=O_3$ (append-append).
	The union of these three sets, $F=F_1 \cup F_2 \cup F_3$, is returned as the candidate vehicles.

	\textbf{Processing empty vehicles.}
	\label{subsec:pruneEmpty}
	Empty vehicles do not have designated trip schedules yet. 
	We only need to check whether they are in the waiting circle of the new request. This can be done by a range query over all empty vehicles using 
	the  waiting circle as the query range.
	
	Since our objective is to minimize the system-wide travel time, the optimal empty vehicle is just the nearest one. We thus take a step further and  directly compute the optimal 
	empty vehicle with a network nearest neighbor algorithm named \emph{IER}~\cite{papadias2003query} which has been shown to be highly efficient~\cite{abeywickrama2016k} (other network nearest neighbor algorithms may also apply). 
	
	\section{The GeoPrune algorithm}
	\label{sec:algorithm}
	Next, we describe our algorithms to handle pruning using the pruning rules described in the previous section, including \emph{pruning}, \emph{match update}, and \emph{move update} algorithms.

	\begin{algorithm}[t]
		
		\DontPrintSemicolon 
		\KwIn{A new trip request $r_n$}
		\KwOut{a set of possible vehicles to serve $r_n$}
		
		\tcp{Pruning stage}
		$r_n.wc$ = the waiting circle of $r_n$\;
		$r_n.rd$ = the detour ellipse of $r_n$\;
		
		$Q_1 \leftarrow T_{seg}.pointQuery(r_n.s)$\;
		$Q_2 \leftarrow T_{seg}.pointQuery(r_n.e)$\;
		$Q_3 \leftarrow T_{end}.rangeQuery(r_n.wc)$\;
		$Q_4 \leftarrow T_{end}.rangeQuery(r_n.rd)$\;
		
		\For{ an element in $Q_1$, $Q_2$, $Q_3$, and  $Q_4$}{
			\If{the time or capacity constraint is violated}{
				remove the element\;
			}
		}
		
		Record the corresponding vehicles of the elements in $Q_1$, $Q_2$, $Q_3$, $Q_4$
		in $O_1$, $O_2$, $O_3$, $O_4$. 
		
		$F, F_1, F_2, F_3 \leftarrow \emptyset$\;
		$F_1 \leftarrow O_1 \cap O_2$\tcp*{insert-insert case}
		$F_2 \leftarrow O_1 \cap O_4$\tcp*{insert-append case}
		$F_3 \leftarrow O_3$\tcp*{append-append case}
		$F \leftarrow F_1 \cup F_2 \cup F_3 $\;
		
		\Return F

		\caption{Prune non-empty vehicles}
		\label{algo:insertion-prune}
	\end{algorithm}   
	\label{subsec:onlineAlgorithmSummary}

	\begin{algorithm}[t]
		\DontPrintSemicolon 
		\KwIn{A new trip request $r_n$ and the matched vehicle $c_i$}
		
		\If{$c_i$  empty }{
			$T_{ev}.remove(c_i)$\;
		}\Else{
			\For{$segment$ in the trip schedule of $c_i$ }{
				remove the ellipse of $segment$ from $T_{seg}$\;
			}
			$T_{end}$.remove(ending stop of $c_i$)\;
		}
		add $r_n.s$ and $r_n.e$ to the trip schedule of $c_i$\;
		\For{$segment$ in the trip schedule of $c_i$}{
			compute the detour ellipse of $segment$\;
			insert the ellipse of $segment$ into $T_{seg}$\;
		}    	
		$T_{end}$.insert(the end stop of $c_i$)\;
		
		\caption{Update index - match}
		\label{algo:insertion-update-match}
	\end{algorithm}
	
	\begin{algorithm}[t]
		\DontPrintSemicolon 
		\KwIn{A moving vehicle $c_i$}
		
		$P \leftarrow$ obsolete segments of $c_i$\;
		
		\For{$p \in P$}{
			$T_{seg}.remove(p$)\;
		}
		\If{$c_i$ reaches the ending stop}{
			$T_{end}.remove$({ending stop of} $c_i$)\;
			$T_{ev}.insert(c_i)$\;
		}
		\caption{Update index - move}
		\label{algo:insertion-update-move}
	\end{algorithm}

	\textbf{Pruning.}
	Algorithm~\ref{algo:insertion-prune} summarizes the pruning algorithm.    
	For every new trip request $r_n$, we first compute the waiting circle and the detour ellipse for $r_n$ (line 1 to line 2).
	Then, we apply four queries to compute four sets $Q_1$, $Q_2$, $Q_3$, 
	and $Q_4$ as described in Section~\ref{sec:applyPruning} (line 3 to line 6).
	Each returned segment and ending stop  is checked against 
	the capacity and time constraints as described in Section~\ref{sec:applyPruning} (line 7 to line 9).
	The vehicles corresponding to the remaining segments and ending stops are  fetched as candidates (line 11 to line 16).

	
	
	\textbf{Match update.}
	If a new trip request $r_n$ is matched with a vehicle $c_i$,
	we update the data structures as summarized in 
	Algorithm~\ref{algo:insertion-update-match}.
	If $c_i$ is an empty vehicle, the vehicle now becomes occupied.
	We remove the vehicle from an R-tree denoted by $T_{ev}$ that stores the empty vehicles for fast nearest empty vehicle computation (line 1 to line 2). Otherwise, we first remove the segments and the ending stop of $c_i$ from the two R-trees $T_{seg}$ and $T_{end}$ (line 4 to line 6).
	Then, we add the new trip request to the trip schedule of the matched vehicle $c_i$ (line 7).
	Based on the updated vehicle schedule, we recompute the detour ellipses and insert them into $T_{seg}$ (line 8 to line 10). 
	The new ending stop is also inserted into $T_{end}$ (line 11).

	
	\textbf{Move update.}
	We also update the data structures when the vehicles move.
	Algorithm~\ref{algo:insertion-update-move} summarizes this update procedure. 
	At every time point, we check if a vehicle has reached a stop in its trip schedule. If yes, the segments before the reached stop become obsolete
	and their detour ellipses are removed from  $T_{seq}$ (line 1 to line 3).
	When the vehicle reaches its ending stop, the vehicle becomes empty.
	We remove it from $T_{end}$ and insert it into $T_{ev}$ (line 4 to line 6).

	\subsection{Algorithm Complexity}
	We measure the complexity of our algorithm by the following two parameters $|S|$ and $C$ as they are key to our algorithm: $|S|$ is the maximum number of stops of the vehicle schedules (which is constrained by the vehicle capacity and is a small constant) and $|C|$ is the number of vehicles. We note, however, that instead of using $|S|$, we can use the number of requests $|R|$ instead because there is a linear relationship: $|S||C| \propto |R|$.
	
	\textbf{Pruning.} 
	It takes $O(1)$ time to compute the waiting circle and the detour ellipse of a new request.
	There are at most $|S||C|$ MBRs in $T_{seg}$ and $|C|$ entries in $T_{end}$.
	The point query on $T_{seg}$ returns at most $|S||C|$ results and hence the complexity is $O(\sqrt{|S||C|} + |S||C|)$~\cite{arge2008priority}\cite{manolopoulos2010r}.
	At most $|C|$ results will be returned from the range query on $T_{end}$ and the complexity is $O(\sqrt{|C|} + |C|)$.
	The time complexity of the queries on R-trees is thus $O(\sqrt{|S||C|} + |S||C|)$. 
	Checking the time and capacity constraints takes $O(|S||C| + |C|)$ time.
	
	It takes $O(|S||C|+|C|)$ time to retrieve the corresponding vehicles and at most $|C|$ vehicles will be returned in each set after sorting ( $O(|S||C|\log(|S||C|))$ time).
	The set intersection hence takes O($|C|$) time~\cite{ding2011fast}.
	The overall time complexity of  GeoPrune  is thus $O(\sqrt{|S||C|} +  |S||C|\log(|S||C|))$.
	
	\textbf{Update.}
	When a new trip request is assigned to a vehicle $c_i$, it takes $O(\log|C|)$ time to delete $c_i$ from $T_{ev}$ if $c_i$ was empty, $O(|S|\log(|S||C|))$ time to remove invalid segments from $T_{seg}$, and $O(\log|C|)$ time to remove the obsolete record in $T_{end}$~\cite{arge2008priority}\cite{manolopoulos2010r}.
	For the new schedule of $c_i$, there are at most $|S|$ new segments.
	It thus takes at $O(|S|)$ time to compute the new detour ellipses for these new segments and $O(|S|\log^2(|S||C|))$ time to insert the ellipses to $T_{seg}$~\cite{arge2008priority}\cite{manolopoulos2010r}.
	The overall update time for a new request is $O(|S|\log^2(|S||C|))$.
	
	When a vehicle moves, the number of obsolete scheduled stops is at most $|S|$.
	Therefore, the time to remove obsolete vehicle ellipses from $T_{seg}$ is $O(|S|\log(|S||C|))$.
	At most $|C|$ vehicles change their status while moving, hence the time to update $T_{end}$ and $T_{ev}$ is at most $O(|C|\log^2|C|)$. 
	Therefore, the overall update time for moving all vehicles in a time slot is $O(|S|\log(|S||C|) + |C|\log^2|C|)$.

	\section{Experiments}
	\label{sec:experiments}
	
	In this section, we study the empirical performance of our \ouralgorithm algorithm and compare it against the state-of-the-art pruning algorithms. 
	All algorithms are implemented in C++ and run on a 64-bit virtual node with a 1.8 GHz CPU and 128 GB memory from
	an academic computing cloud (Nectar~\cite{nectar}) running on OpenStack.
	The travel distance between points is computed by a shortest path algorithm on road networks~\cite{akiba2014fast}.

	\subsection{Experimental Setup}

	\textbf{Dataset.}
	We perform the experiments on real-world road network datasets, \emph{New York City} (NYC) and \emph{Chengdu} (CD, a capital city in China).
	These two road networks are extracted from OpenStreetMap~\cite{osm}.
	We transform the coordinates to Universal Transverse Mercator (UTM) coordinates to support pruning based on Euclidean distance.
	We use real-world taxi request data on the two road networks~\cite{NYTaxiData}\cite{ChengduTaxiData} and remove unrealistic trip requests, i.e.,  duration time less than 10 seconds or longer than 6 hours. 
	There are 448,128 taxi requests (April 09, 2016) for NYC and 259,423 (November 18, 2016) taxi requests for Chengdu.
	Every taxi request consists of a source location, a destination location and an issue time. We map the locations to their respective nearest road network vertices.
	Similar to previous studies~\cite{chen2018price}\cite{huang2014large}, we assume the number of passengers to be one per request. 

	\begin{table}[t]
		\setlength{\tabcolsep}{5pt}
		\caption{Datasets}
		\label{tab:roadNetworkDatasets}
		\centering
		\begin{tabular}{l r r r}
			\hline
			\textbf{Name}  & \textbf{\# vertices} & \textbf{\# edges} & \textbf{\# requests}   \\ \hline
			NYC   & 166,296       & 405,460     &    448,128   \\ 
			CD & 254,423 & 467,773 & 259,343  \\ \hline
		\end{tabular}
	\end{table}
	
	\begin{table}[t]
		\setlength{\tabcolsep}{12pt}
		\caption{Experiment parameters}
		\label{tab:expSettings}
		\centering
		\begin{tabular}{l r r }
			\hline
			\textbf{Parameters}  & \textbf{Values} & \textbf{Default}    \\ \hline
			Number of vehicles  & $2^{10}$ to $2^{17}$     & $2^{13}$         \\ 
						Capacity  & 2, 4, 6, 8, 10        & 4        \\ 
			Waiting time (min)  & 2, 4, 6, 8, 10      & 4  \\ 
			Detour ratio    & 0.2, 0.4, 0.6, 0.8     & 0.2     \\ 
			Number of requests & \makecell[tr]{20k to 100k} & 60k \\ 
			\makecell[tl]{Frequency of requests \\ (\# requests/second)}    &    1  to 10  &  \makecell[tr]{refer to\\ table~\ref{tab:roadNetworkDatasets}}\\ 
			\makecell[tl]{Transforming speed (km/h)} & \makecell[tr]{20 to 140} & 48 \\ \hline 
		\end{tabular}\\
	\end{table}
	\textbf{Implementation.}
	We run simulations following the settings of previous studies~\cite{tong2018unified},\cite{huang2014large}.
	The initial positions of vehicles are randomly selected from the road network vertices.
	Non-empty vehicles move on the road network following their trip schedules (shortest paths) while empty vehicles stay at their latest drop-off location until they are committed to new requests.
	Similar to previous studies~\cite{huang2014large},\cite{chen2018price}, we use a constant travel speed for all edges in the road network (48km/h).
	For the selection step, we apply the state-of-the-art insertion  algorithm~\cite{tong2018unified} to minimize the total travel distance for all methods compared. 
	If no satisfying vehicle is found for a new trip request, the trip request is ignored.
	
	Table~\ref{tab:expSettings} summarizes the parameters used in our experiments. By default,
	we simulate ride-sharing on $2^{13}$ vehicles with a capacity  of 4 and 60,000 trip requests, and the maximum waiting time and the detour ratio are 4\,min and 0.2.

	\textbf{Baselines.}
	We compare \ouralgorithm against the following state-of-the-art pruning algorithms.
	The parameter values are set according to the numbers reported in the original papers.
	\begin{itemize}
		\item 	\textbf{GreedyGrids~\cite{tong2018unified}.} This algorithm retrieves  all vehicles that are currently in the nearby grid cells.
		
		\item \textbf{Tshare~\cite{ma2013t}.} This is the single-side search algorithm of Tshare~\cite{ma2013t}. The grid cell lengths of both GreedyGrids and Tshare are set to 1\,km~\cite{tong2018unified}.
		
		\item \textbf{Xhare~\cite{thangaraj2017xhare}.} This algorithm only checks the non-empty vehicles. 
		To make it applicable for finding empty vehicles, we prune empty vehicles in Xhare using the same algorithm applied in our method (see Section~\ref{sec:applyPruning}). We optimize the update process by precomputing the pair-wise distance between clusters.
		The landmark size is set to 16,000 for NYC and 23,000 for Chengdu, and the grid cell length is set to 10\,m.
		The maximum distance between landmarks in a cluster is set to 1\,km.
	\end{itemize}

	
	\begin{figure}[t]
		\centering
		\begin{subfigure}[t]{0.24\textwidth}  
			\centering 
			\includegraphics[width=1.0\textwidth]{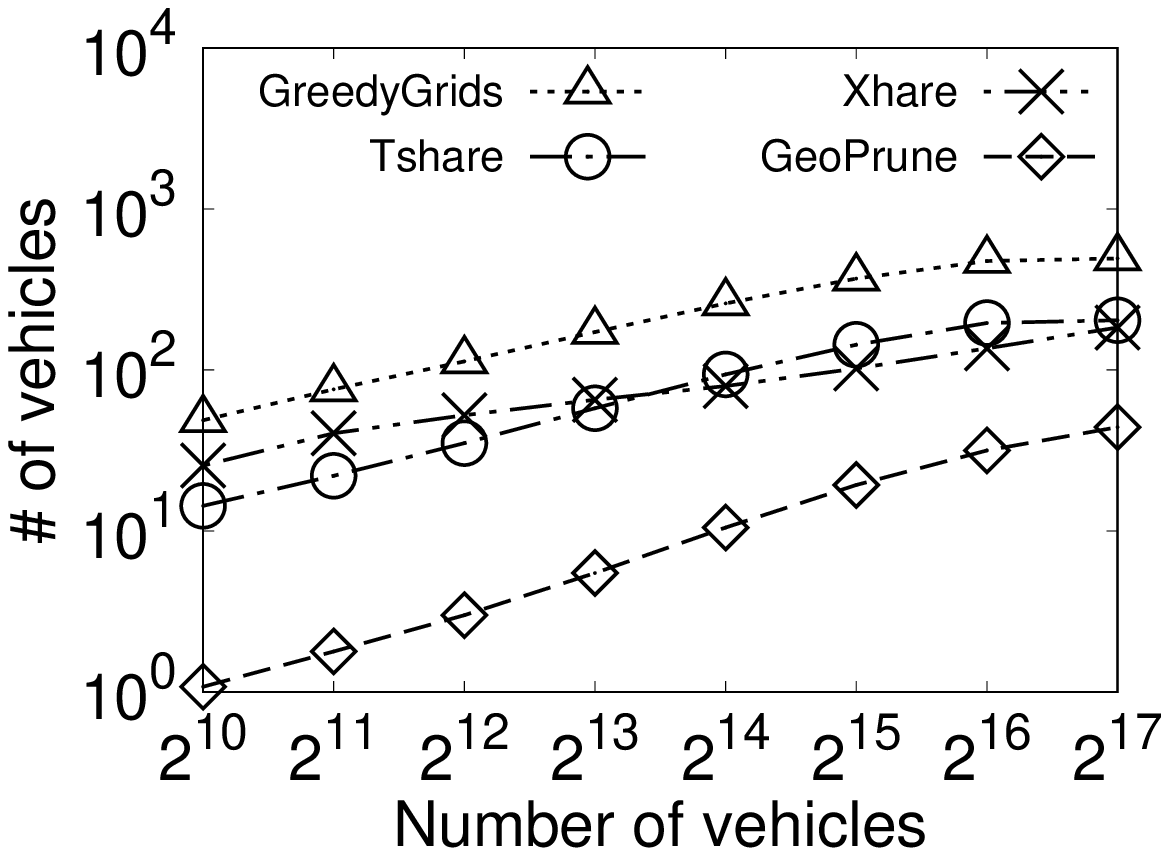}
			\vspace{-6mm}
			\caption{\# remaining  vehicles (NYC).}%
			\label{fig:NYVarynumVehicles_numPrune}
		\end{subfigure}
		\begin{subfigure}[t]{0.24\textwidth}  
			\centering 
			\includegraphics[width=1.0\textwidth]{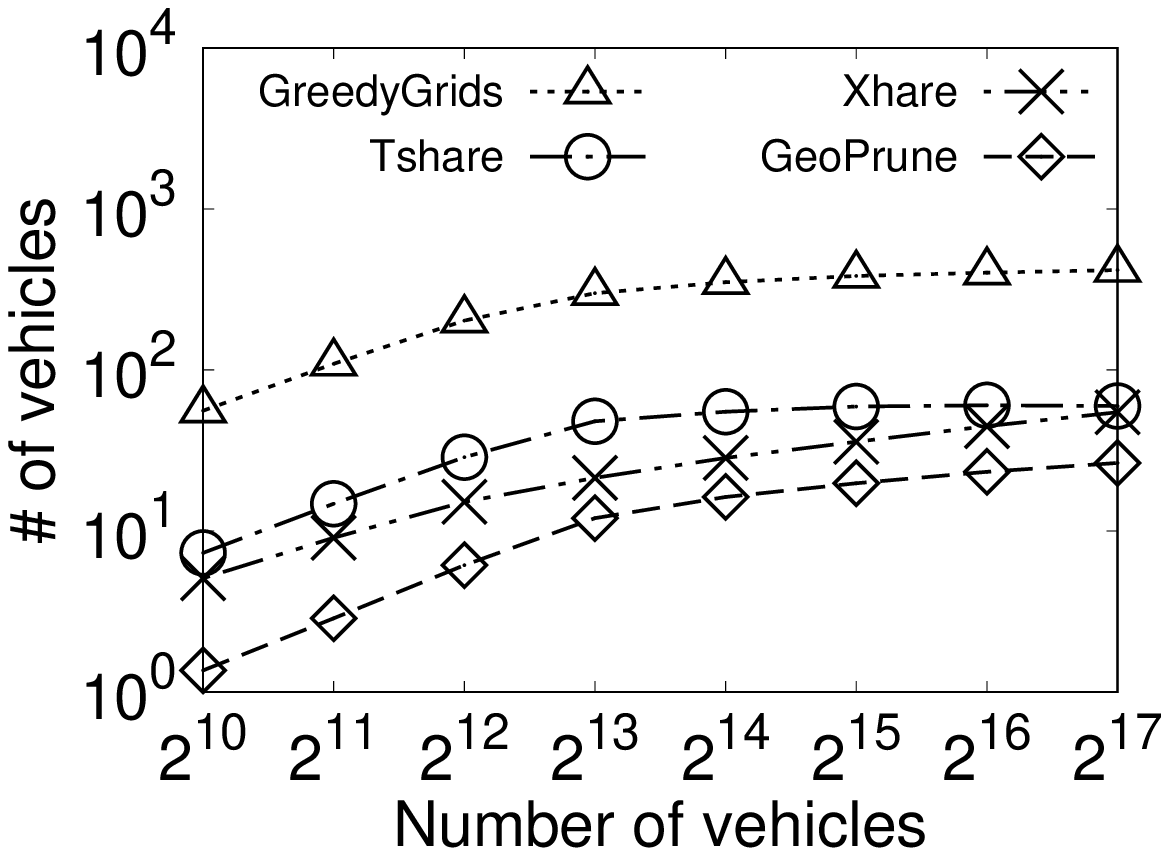}
			\vspace{-6mm}
			\caption{\# remaining vehicles (CD).}%
			\label{fig:chengduVarynumVehicles_numPrune}
		\end{subfigure}
		\vspace{3mm}
		\begin{subfigure}[t]{0.24\textwidth}  
			\centering 
			\includegraphics[width=1.0\textwidth]{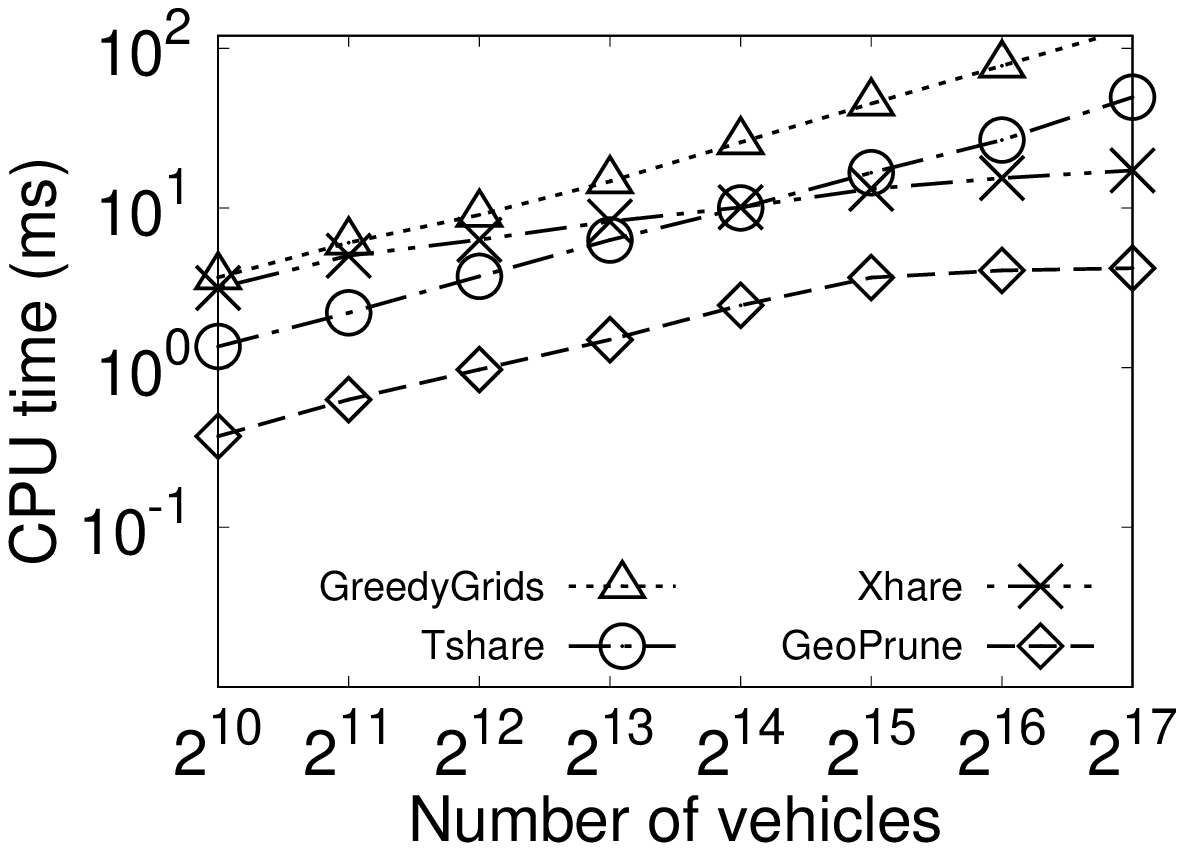}
			\vspace{-6mm}
			\caption{Overall match time (NYC).}%
			\label{fig:NYVarynumVehicles_matchTime}
		\end{subfigure}
		\begin{subfigure}[t]{0.24\textwidth}  
			\centering 
			\includegraphics[width=1.0\textwidth]{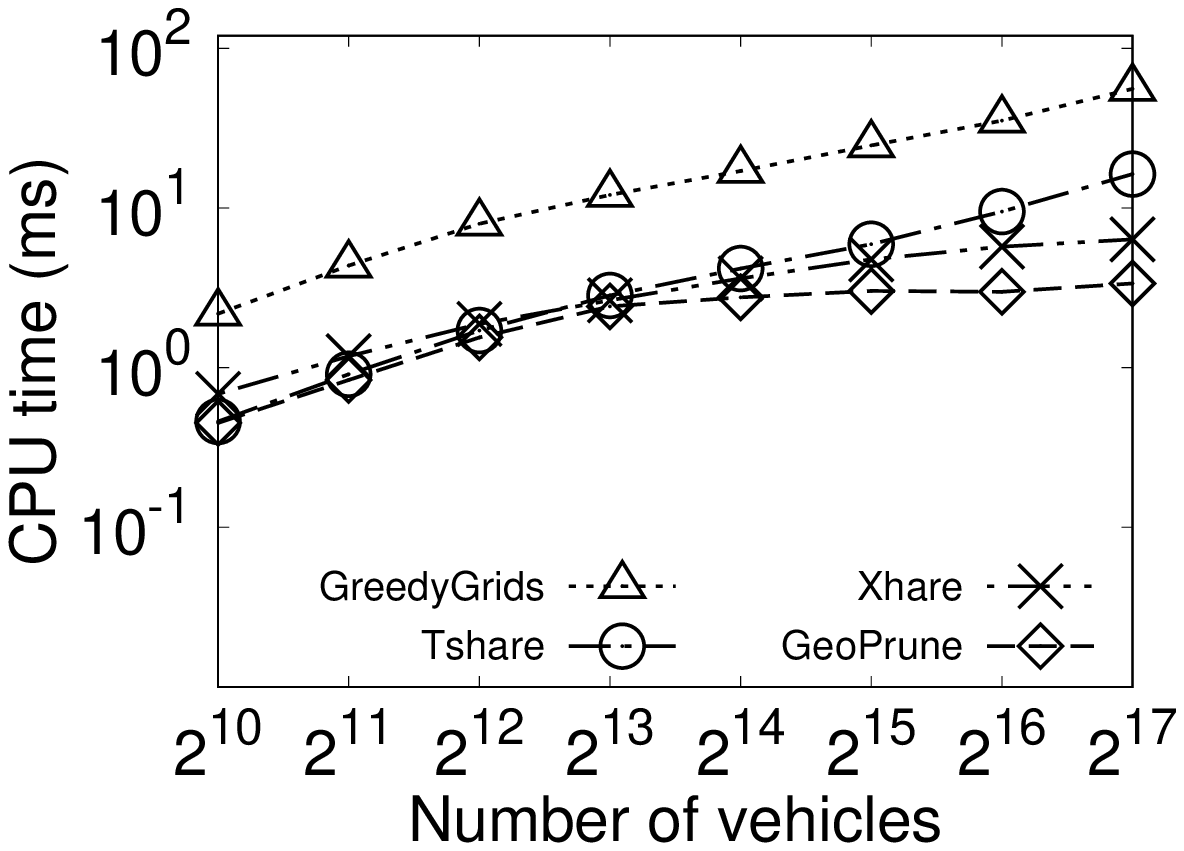}
			\vspace{-6mm}
			\caption{Overall match time (CD).}%
			\label{fig:chengduVarynumVehicles_matchTime}
		\end{subfigure}
		\vspace{3mm}
		\begin{subfigure}[t]{0.24\textwidth}  
			\centering 
			\includegraphics[width=\textwidth]{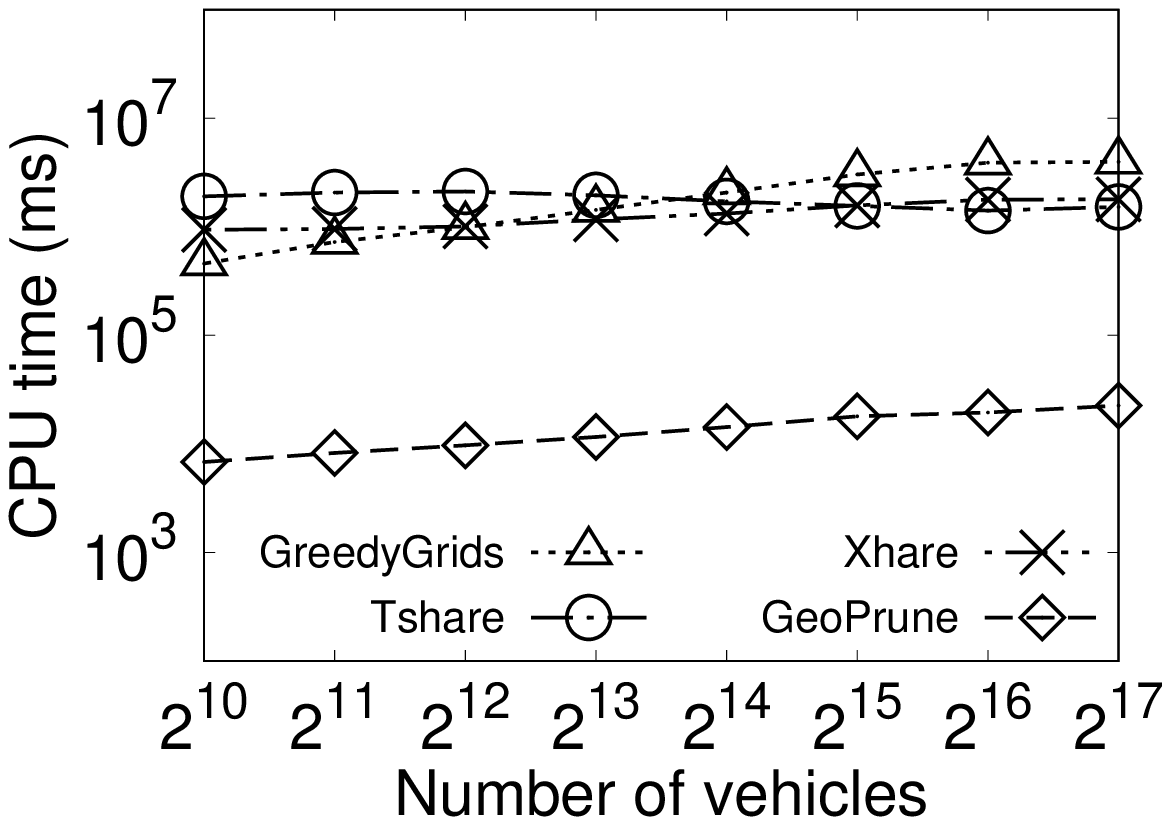}
			\vspace{-6mm}
			\caption{Overall update time (NYC).}%
			\label{fig:NYVarynumVehicles_overallUpdateTime}
		\end{subfigure}
		\begin{subfigure}[t]{0.24\textwidth}  
			\centering 
			\includegraphics[width=\textwidth]{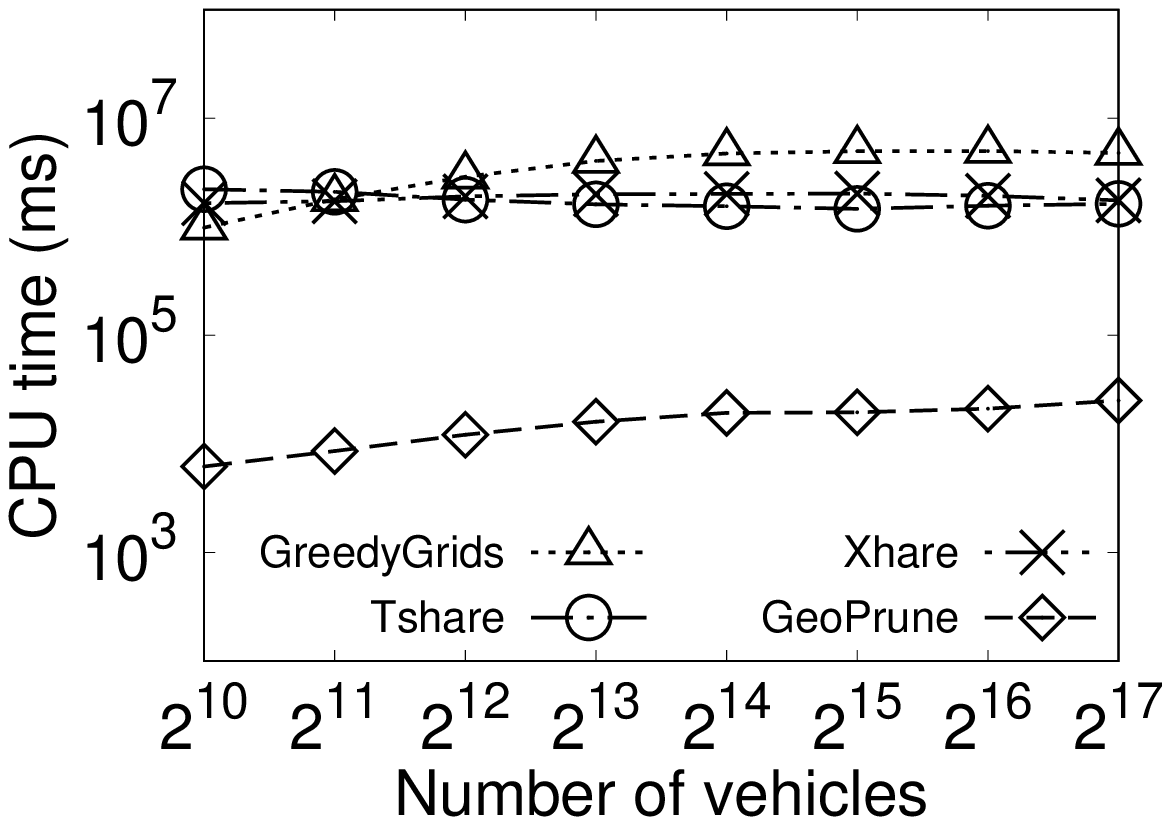}
			\vspace{-6mm}
			\caption{Overall update time (CD).}%
			\label{fig:chengduVarynumVehicles_overallUpdateTime}
		\end{subfigure}
		\caption{Effect of the number of vehicles.}
		\label{fig:varyNumVehicles}
	\end{figure}
	
				\begin{figure}[t]
		\centering
		\begin{subfigure}[t]{0.24\textwidth}
			\centering
			\includegraphics[width=\textwidth]{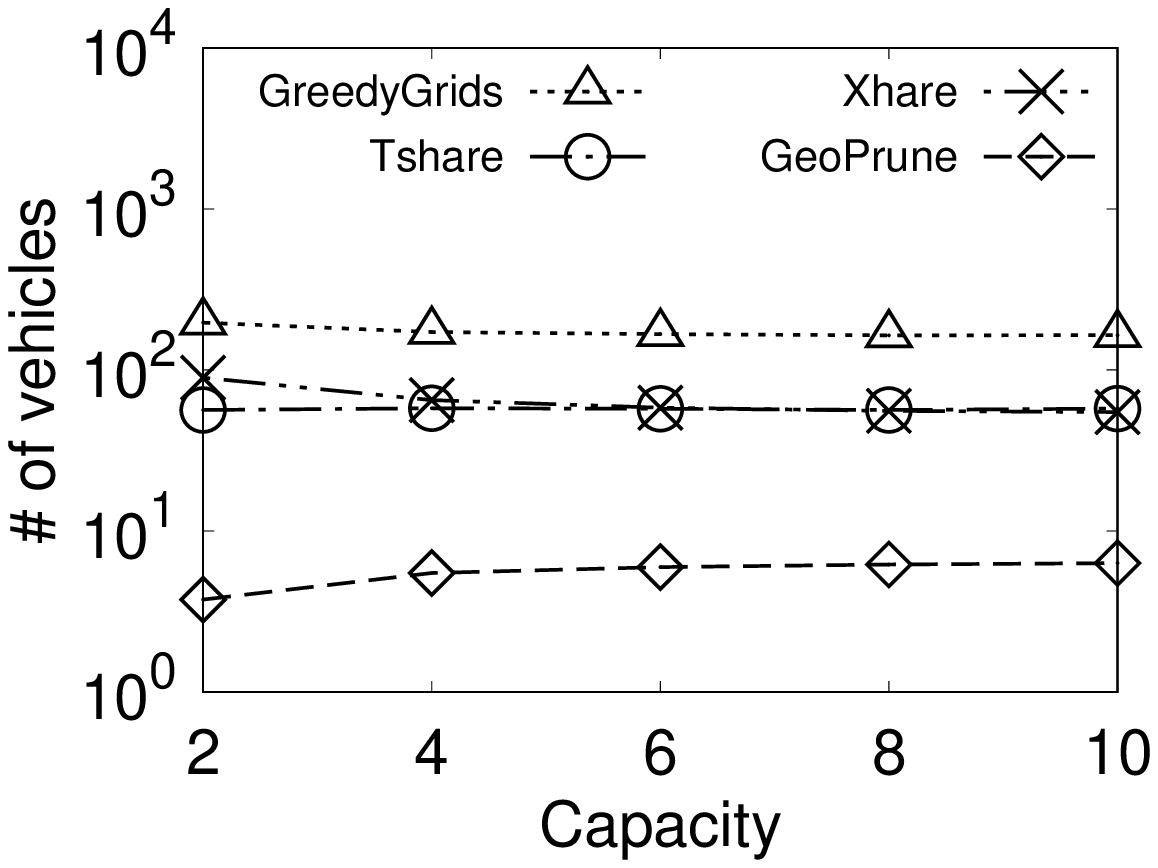}
			\vspace{-6mm}
			\caption{\# remaining vehicles (NYC).}%
			\label{fig:NYVaryCapacity_numPrune}
		\end{subfigure}
		\begin{subfigure}[t]{0.24\textwidth}
			\centering
			\includegraphics[width=\textwidth]{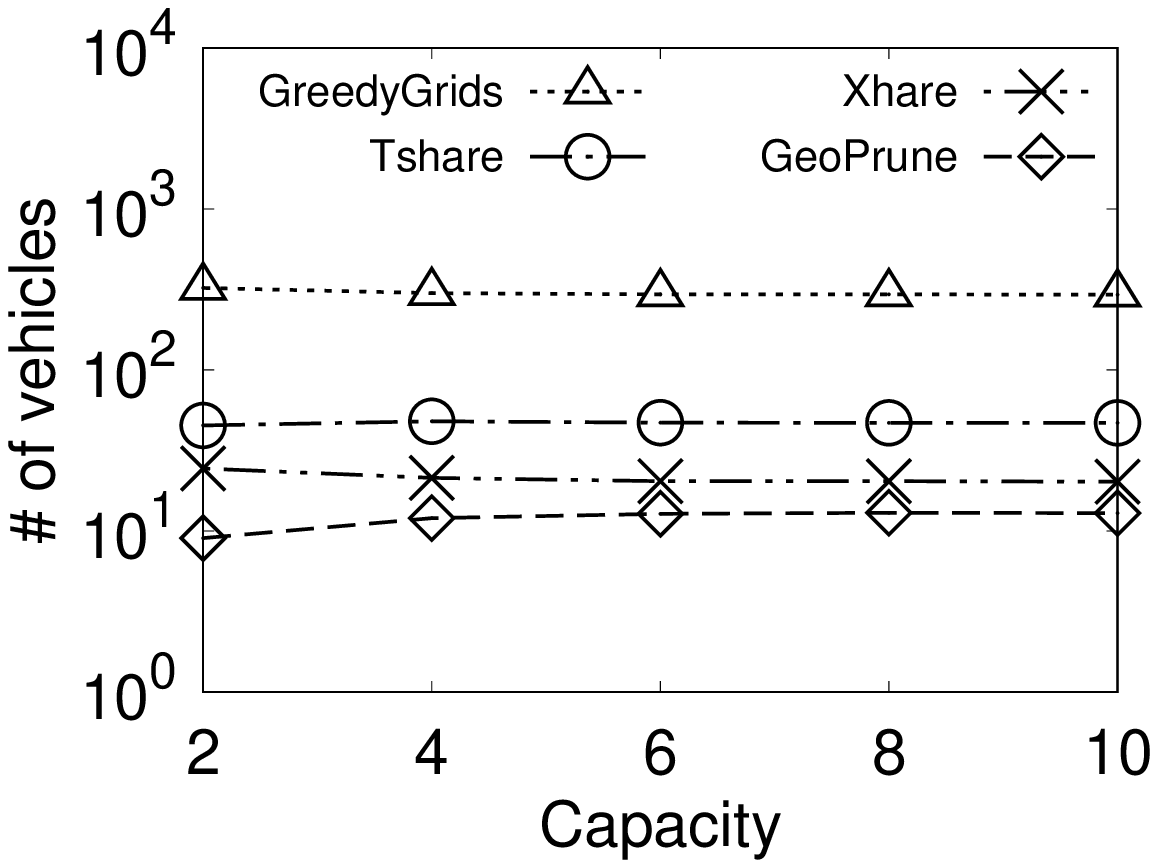}
			\vspace{-6mm}
			\caption{\# remaining vehicles (CD).}%
			\label{fig:ChengduVaryCapacity_numPrune}
		\end{subfigure}
		\vspace{3mm}
		\begin{subfigure}[t]{0.24\textwidth}
			\centering
			\includegraphics[width=\textwidth]{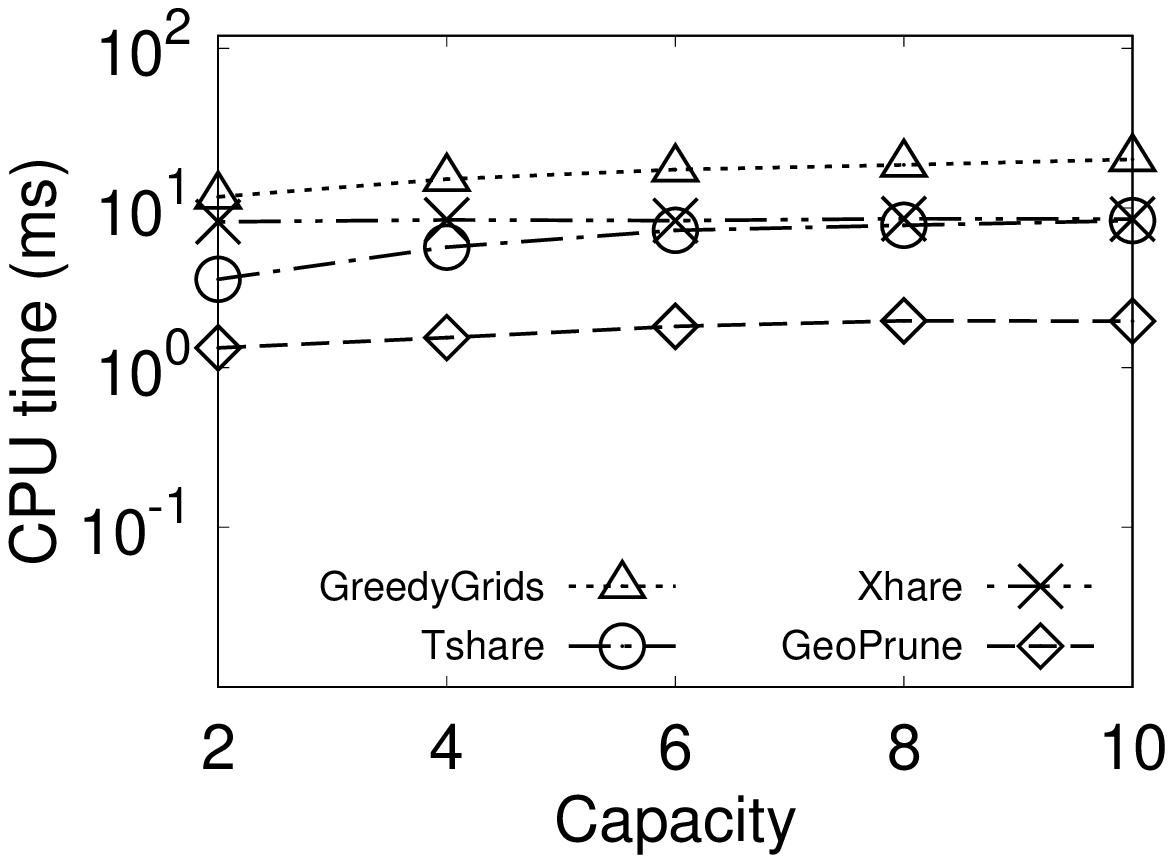}
			\vspace{-6mm}
			\caption{Overall match time (NYC).}%
			\label{fig:NYVaryCapacity_matchTime}
		\end{subfigure}
		\begin{subfigure}[t]{0.24\textwidth}
			\centering
			\includegraphics[width=\textwidth]{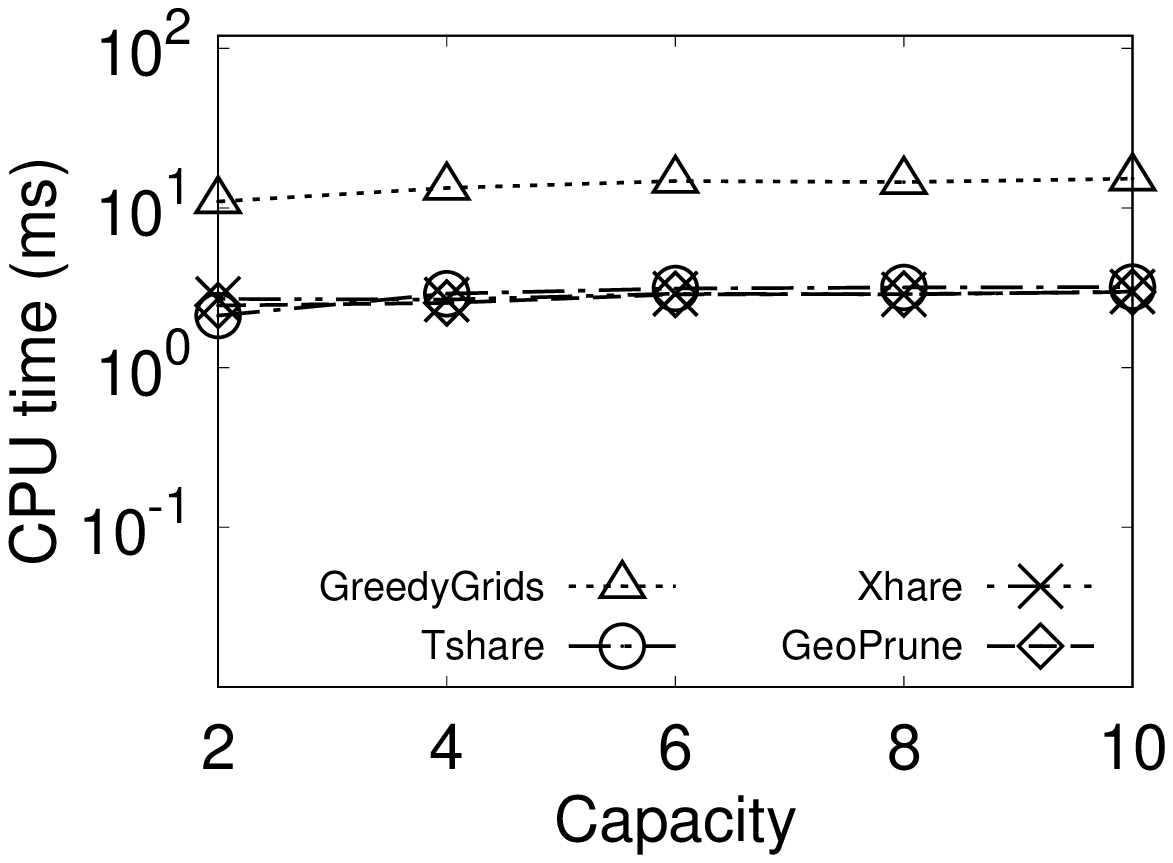}
			\vspace{-6mm}
			\caption{Overall match time (CD).}%
			\label{fig:ChengduVaryCapacity_matchTime}
		\end{subfigure}
		\vspace{3mm}
		\begin{subfigure}[t]{0.24\textwidth}  
			\centering 
			\includegraphics[width=\textwidth]{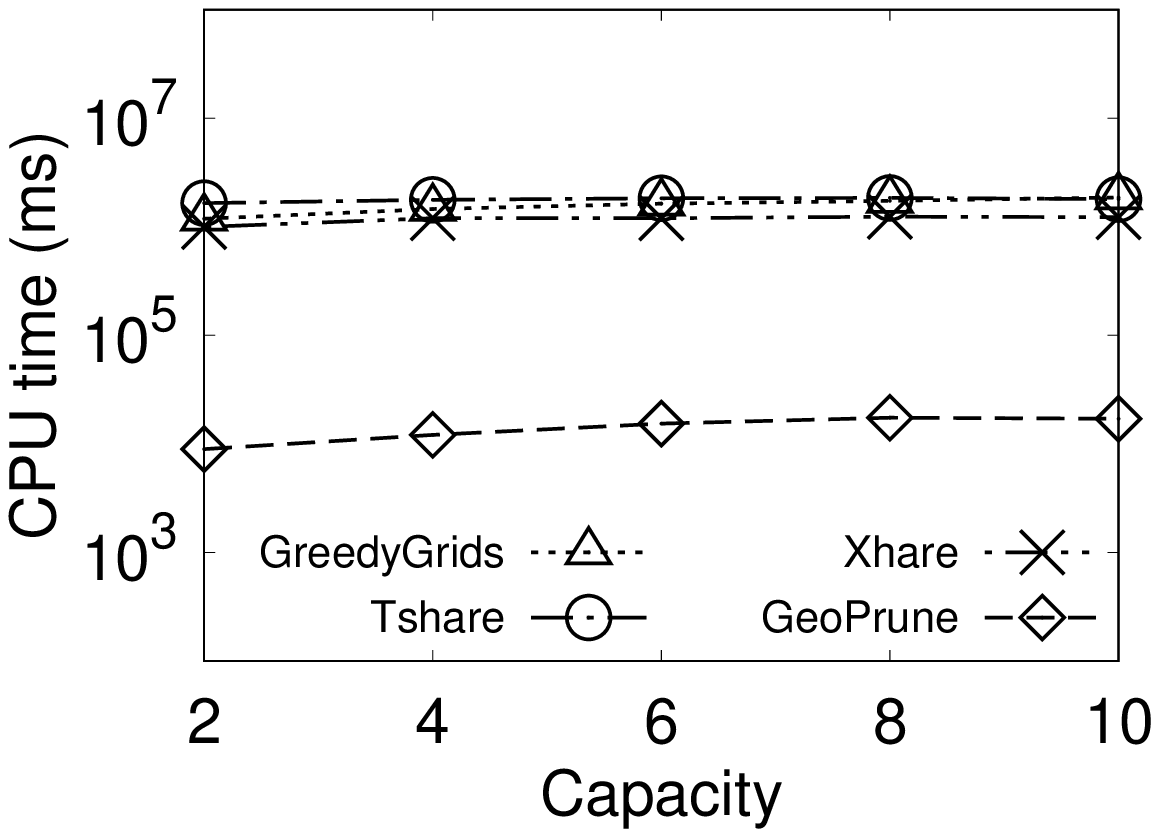}
			\vspace{-6mm}
			\caption{Overall update time (NYC).}%
			\label{fig:NYVaryCapacity_overallUpdateTime}
		\end{subfigure}
		\begin{subfigure}[t]{0.24\textwidth}  
			\centering 
			\includegraphics[width=\textwidth]{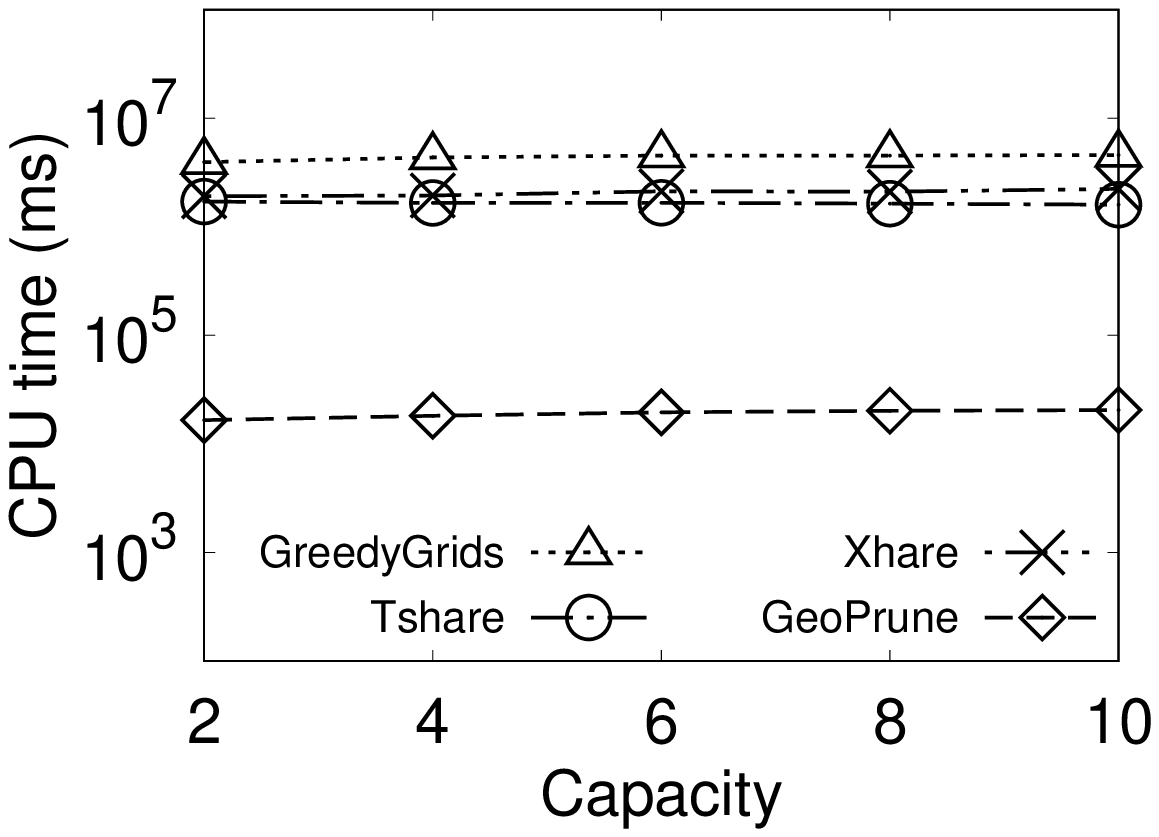}
			\vspace{-6mm}
			\caption{Overall update time (CD).}%
			\label{fig:ChengduVaryCapacity_overallUpdateTime}
		\end{subfigure}
		\caption{Effect of the capacity. }
		\label{fig:varyCapacity}
	\end{figure}

	\textbf{Metrics.}
	We measure and report the following metrics:
	\begin{itemize}
		\item \textit{Number of remaining vehicles} -- the number of remaining candidate vehicles after the pruning.
		Note that \ouralgorithm prunes empty vehicles and non-empty vehicles separately with different criteria and such a scheme is applied on Xhare to make it applicable.  
		GreedyGrids and Tshare, however, process the two types of vehicles together and return both types after pruning.
		For consistency, we only compare the number of remaining non-empty vehicles.
		
		\item 
		\textit{Match time} -- the total running time of the matching process, including both pruning and selection time.
		
		\item 	\textit{Overall update time} -- the overall match update and move update time.
		
		\item 
		\textit{Memory consumption} -- the memory cost of the data structures of an algorithm.
		
	\end{itemize}

	%
	%
	%
	
	\vspace{-1em}	
	\subsection{Experimental Results}
	
	\vspace{-0.4em}
	\subsubsection{Effect of the Number of Vehicles}
	
	\begin{figure}[t]
		\centering
		\begin{subfigure}[t]{0.24\textwidth}
			\centering
			\includegraphics[width=\textwidth]{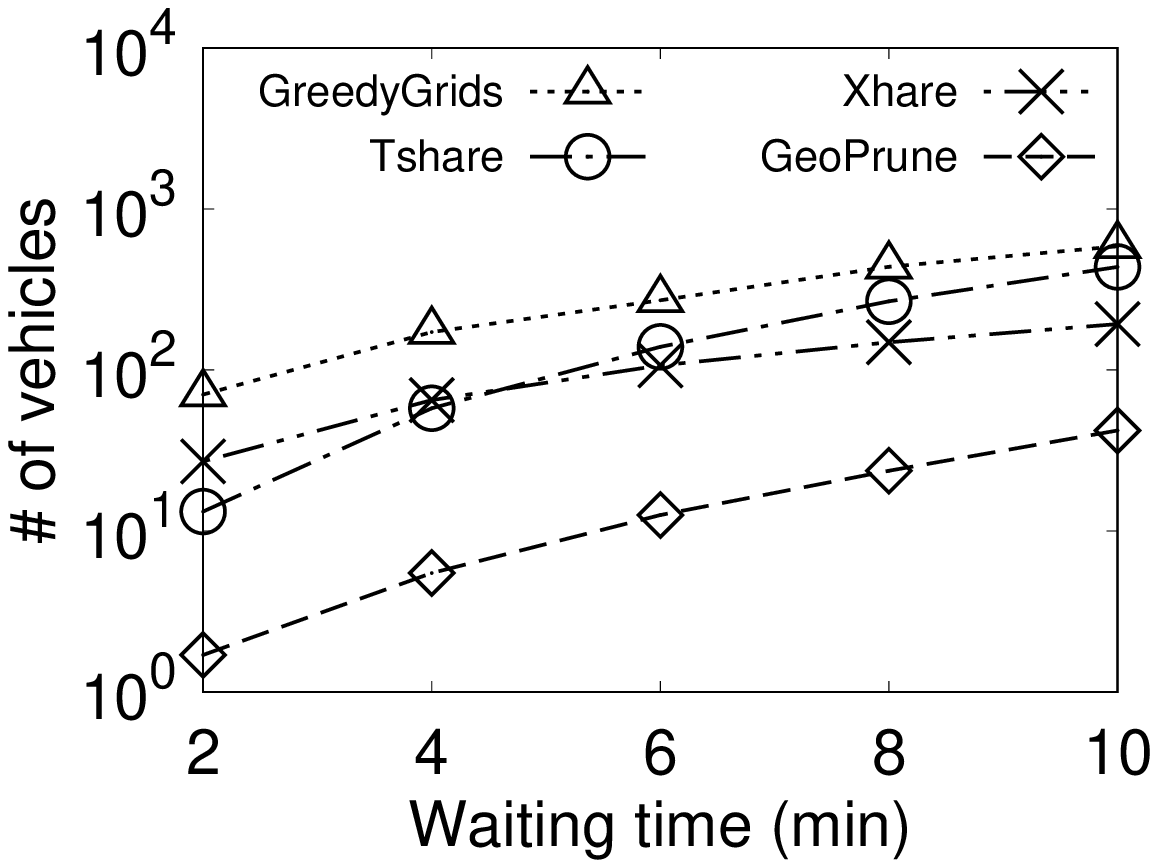}
			\vspace{-6mm}
			\caption{\# remaining vehicles (NYC).}%
			\label{fig:NYVaryWaitingTime_numPrune}
		\end{subfigure}
		\begin{subfigure}[t]{0.24\textwidth}
			\centering
			\includegraphics[width=\textwidth]{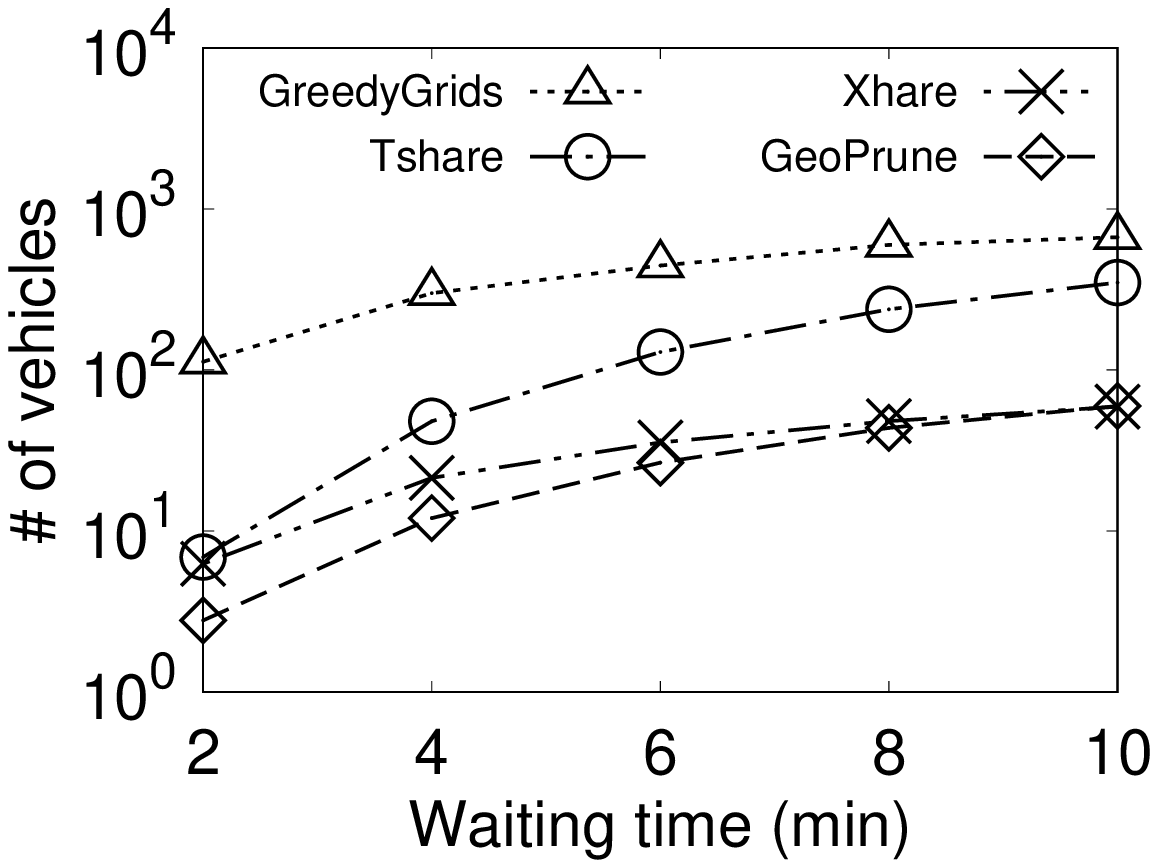}
			\vspace{-6mm}
			\caption{\# remaining vehicles (CD).}%
			\label{fig:ChengduVaryWaitingTime_numPrune}
		\end{subfigure}
		\vspace{3mm}
		\begin{subfigure}[t]{0.24\textwidth}
			\centering
			\includegraphics[width=\textwidth]{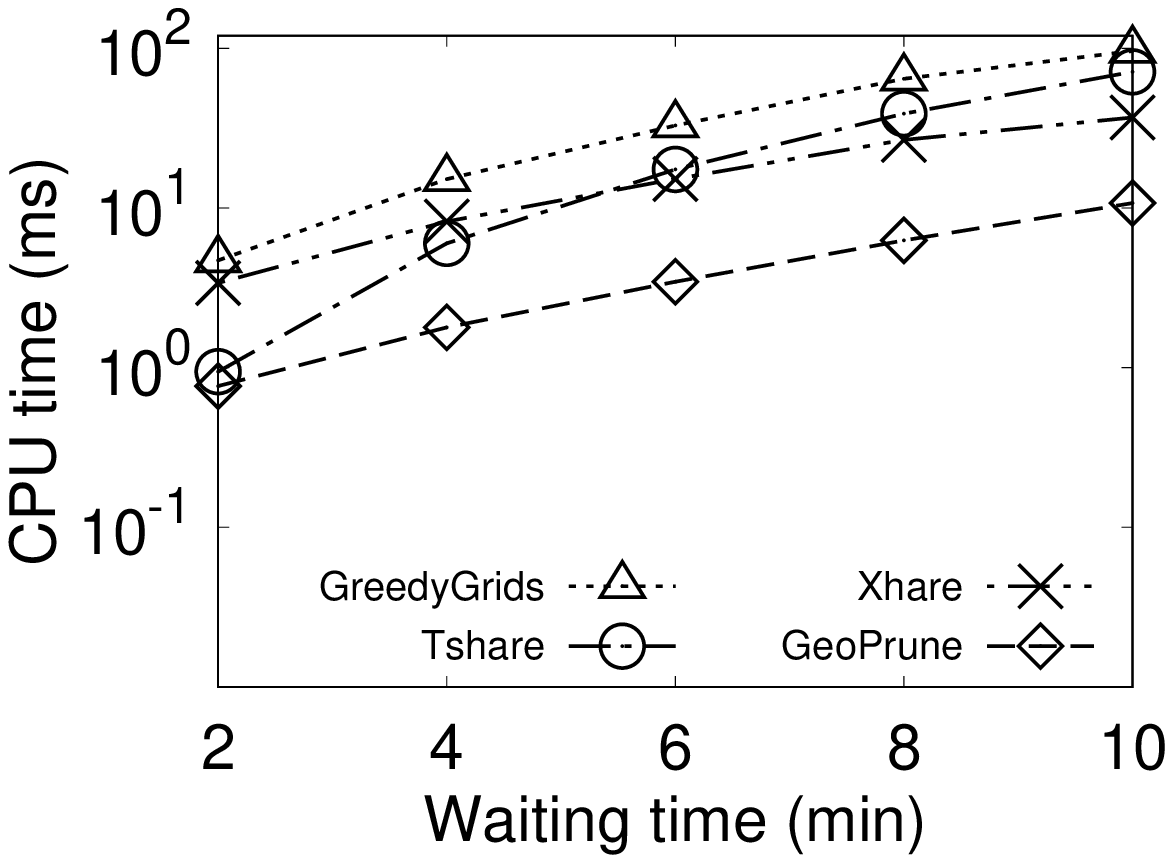}
			\vspace{-6mm}
			\caption{Overall match time (NYC).}%
			\label{fig:NYVaryWaitingTime_matchTime}
		\end{subfigure}
		\begin{subfigure}[t]{0.24\textwidth}
			\centering
			\includegraphics[width=\textwidth]{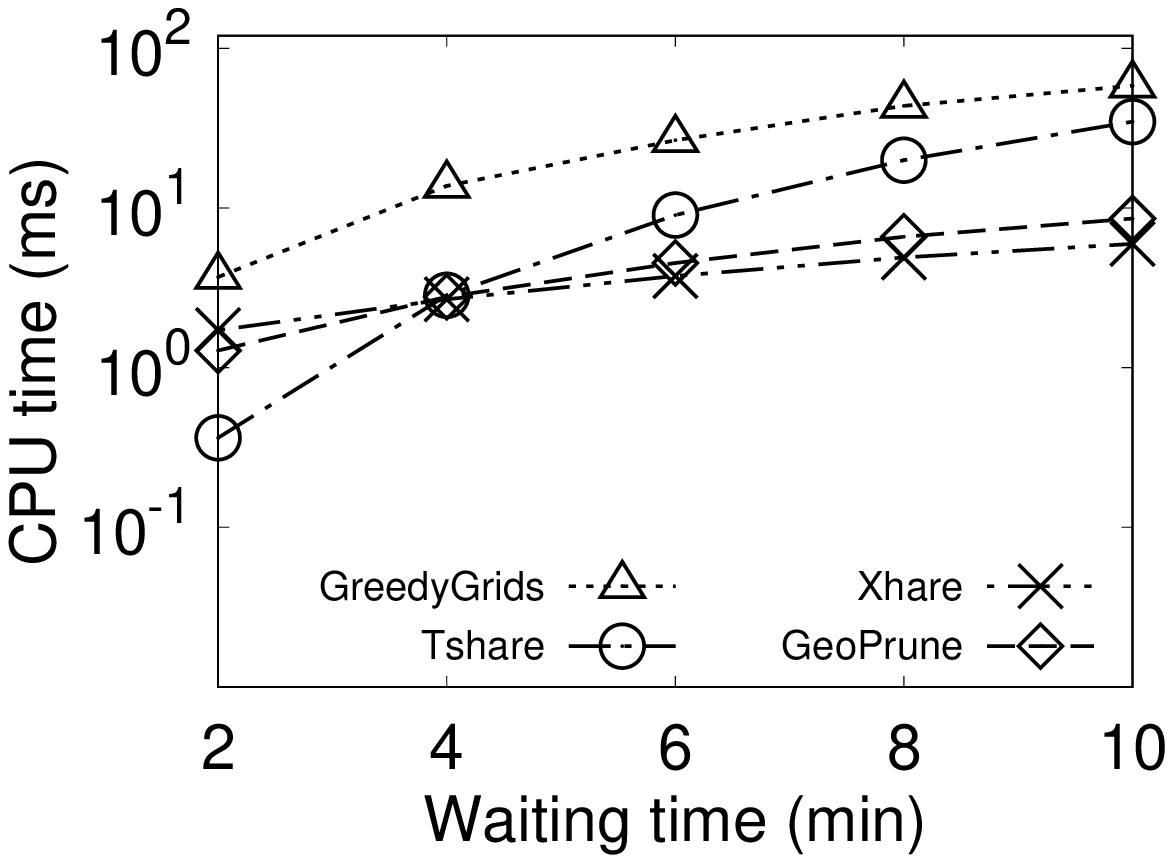}
			\vspace{-6mm}
			\caption{Overall match time (CD).}%
			\label{fig:ChengduVarywaitingTime_matchTime}
		\end{subfigure}
		\vspace{3mm}
		\begin{subfigure}[t]{0.24\textwidth}  
			\centering 
			\includegraphics[width=\textwidth]{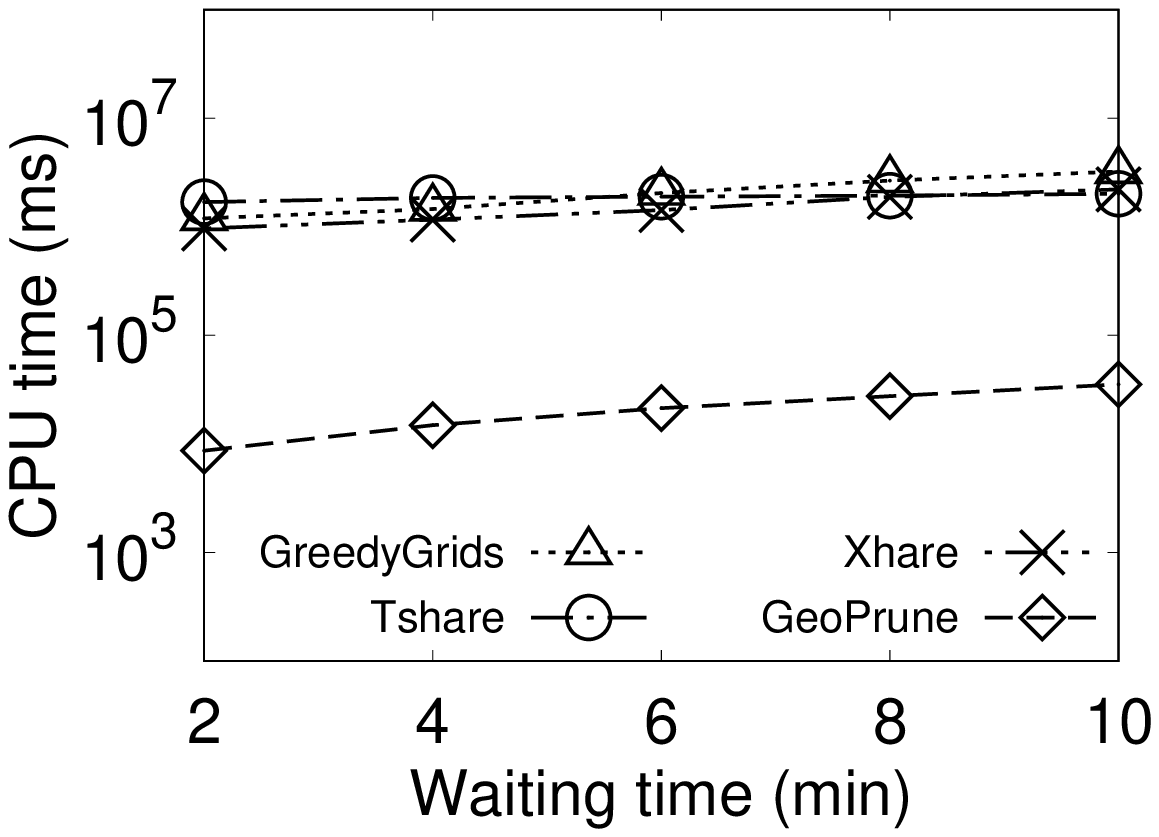}
			\vspace{-6mm}
			\caption{Overall update time (NYC).}%
			\label{fig:NYVaryWaitingTime_overallUpdateTime}
		\end{subfigure}
		\begin{subfigure}[t]{0.24\textwidth}  
			\centering 
			\includegraphics[width=\textwidth]{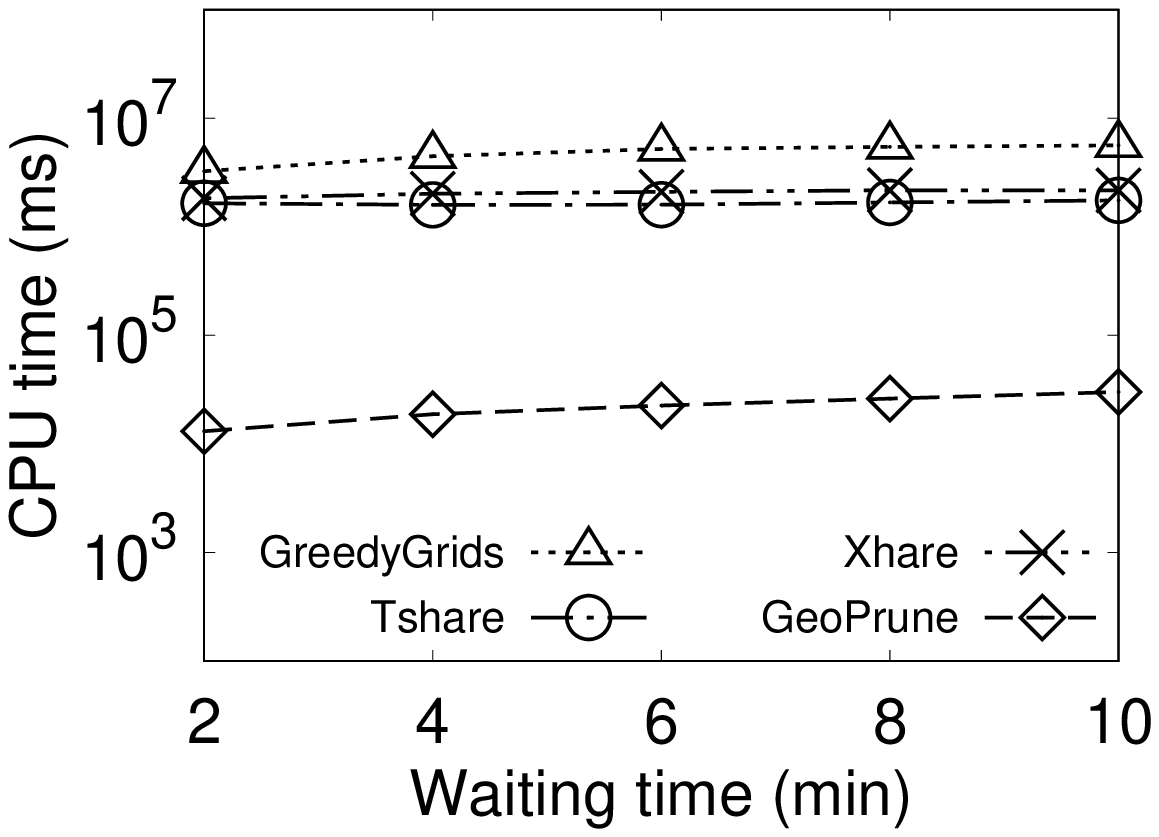}
			\vspace{-6mm}
			\caption{Overall update time (CD).}%
			\label{fig:ChengduVaryWaitingTime_overallUpdateTime}
		\end{subfigure}
		\caption{Effect of the waiting time. }
		\label{fig:NYVaryWaitingTime}
	\end{figure}
	
	\begin{figure}[t]
		\centering
		\begin{subfigure}[t]{0.24\textwidth}
			\centering
			\includegraphics[width=\textwidth]{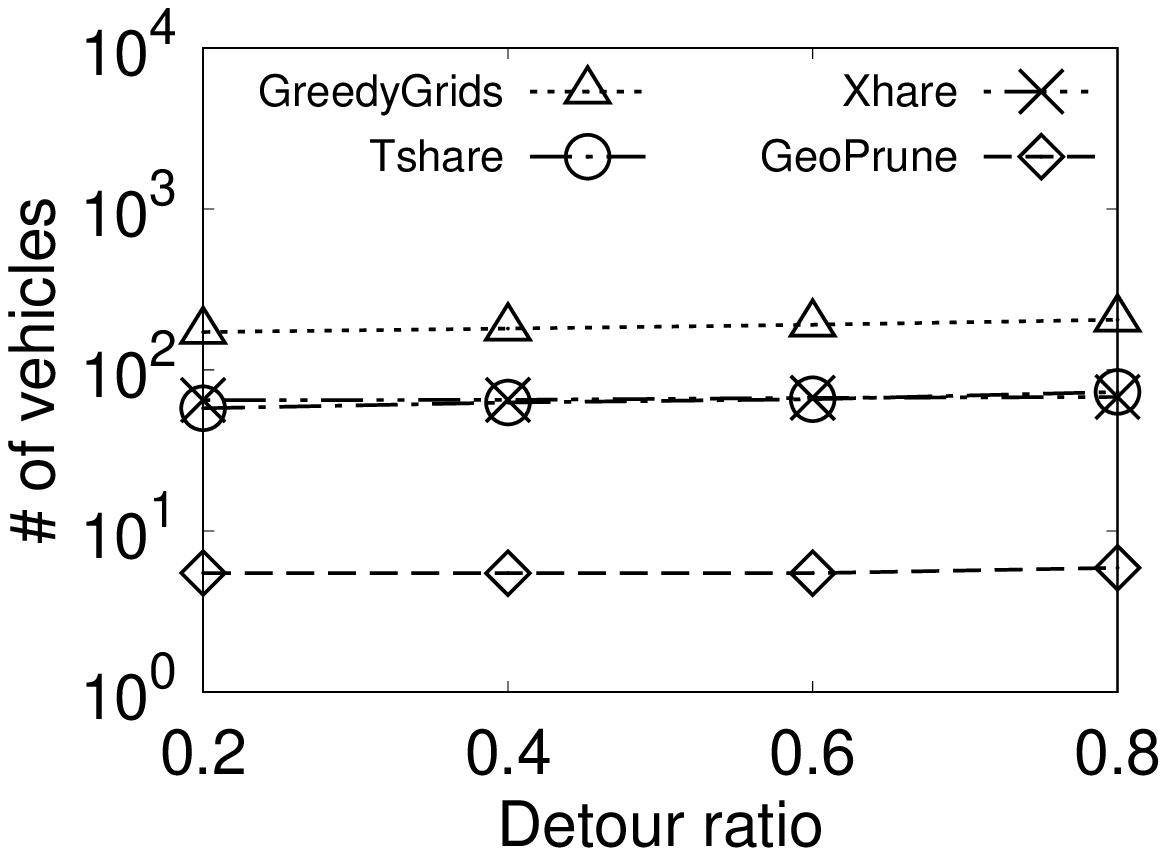}
			\vspace{-6mm}
			\caption{\# remaining vehicles (NYC).}%
			\label{fig:NYVaryDetourRatio_numPrune}
		\end{subfigure}
		\begin{subfigure}[t]{0.24\textwidth}
			\centering
			\includegraphics[width=\textwidth]{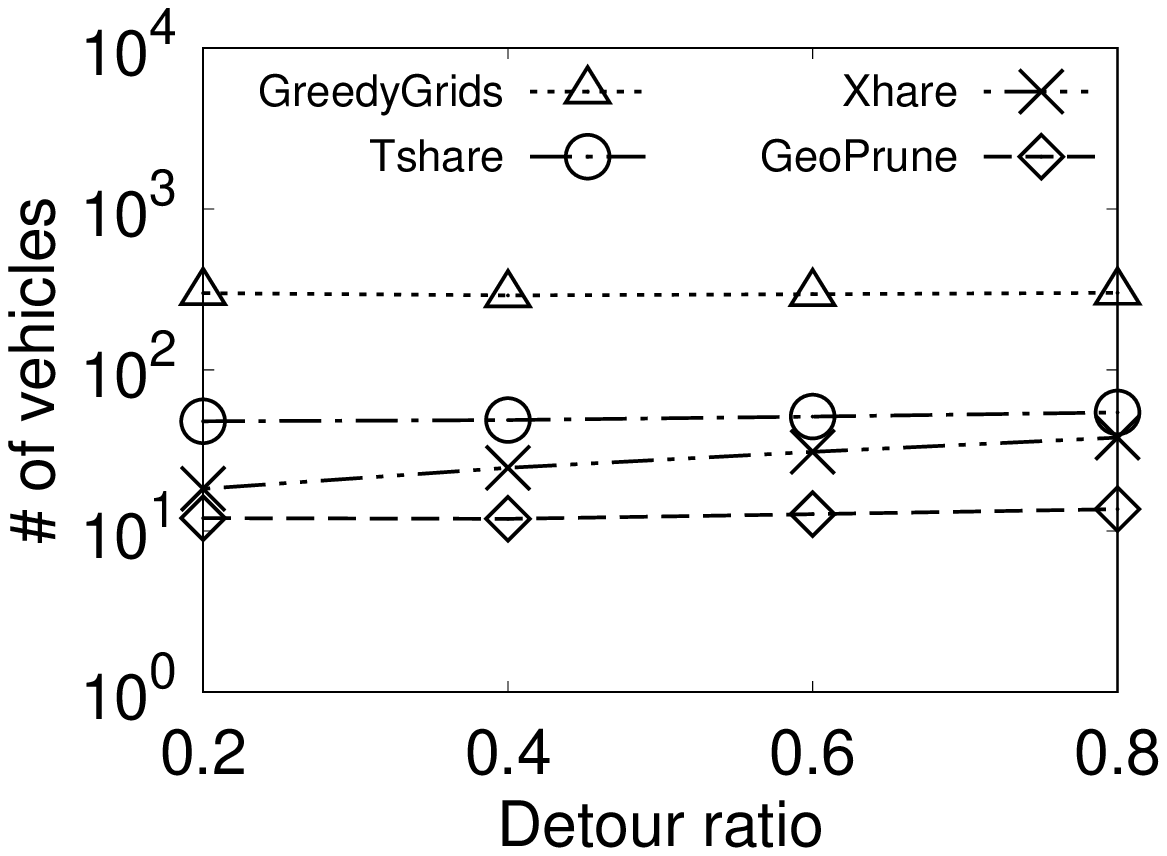}
			\vspace{-6mm}
			\caption{\# remaining vehicles (CD).}%
			\label{fig:ChengduVaryDetourRatio_numPrune}
		\end{subfigure}
		\vspace{3mm}
		\begin{subfigure}[t]{0.24\textwidth}
			\centering
			\includegraphics[width=\textwidth]{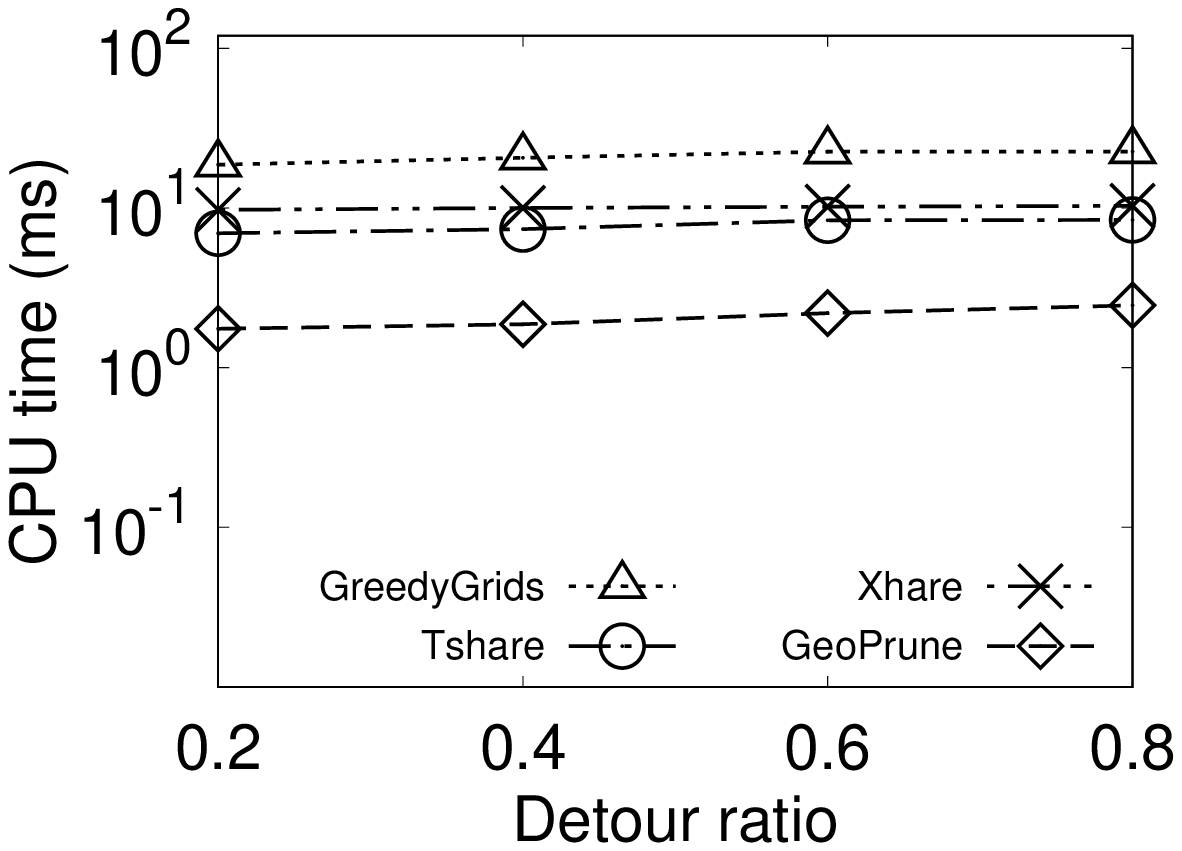}
			\vspace{-6mm}
			\caption{Overall match time (NYC).}%
			\label{fig:NYVaryDetourRatio_matchTime}
		\end{subfigure}
		\begin{subfigure}[t]{0.24\textwidth}
			\centering
			\includegraphics[width=\textwidth]{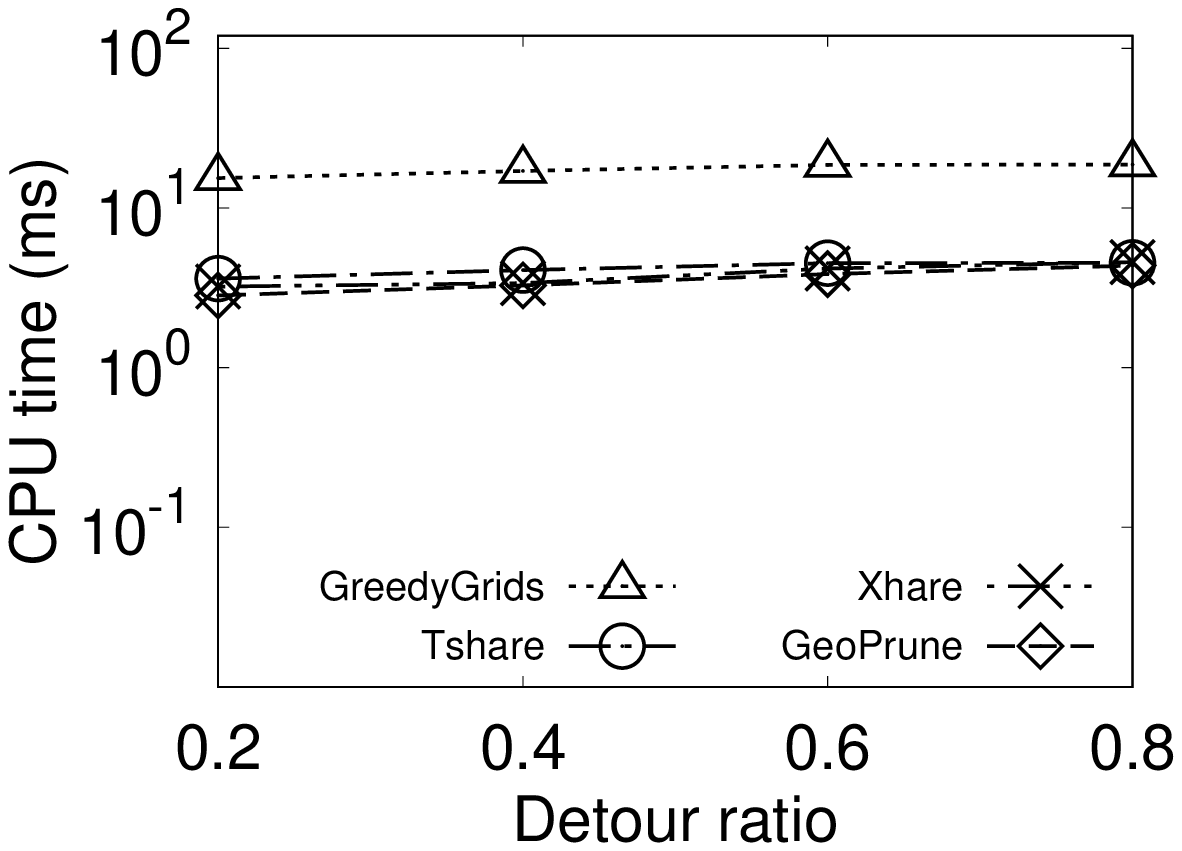}
			\vspace{-6mm}
			\caption{Overall match time (CD).}%
			\label{fig:ChengduVaryDetourRatio_matchTime}
		\end{subfigure}
		\vspace{3mm}
		\begin{subfigure}[t]{0.24\textwidth}  
			\centering 
			\includegraphics[width=\textwidth]{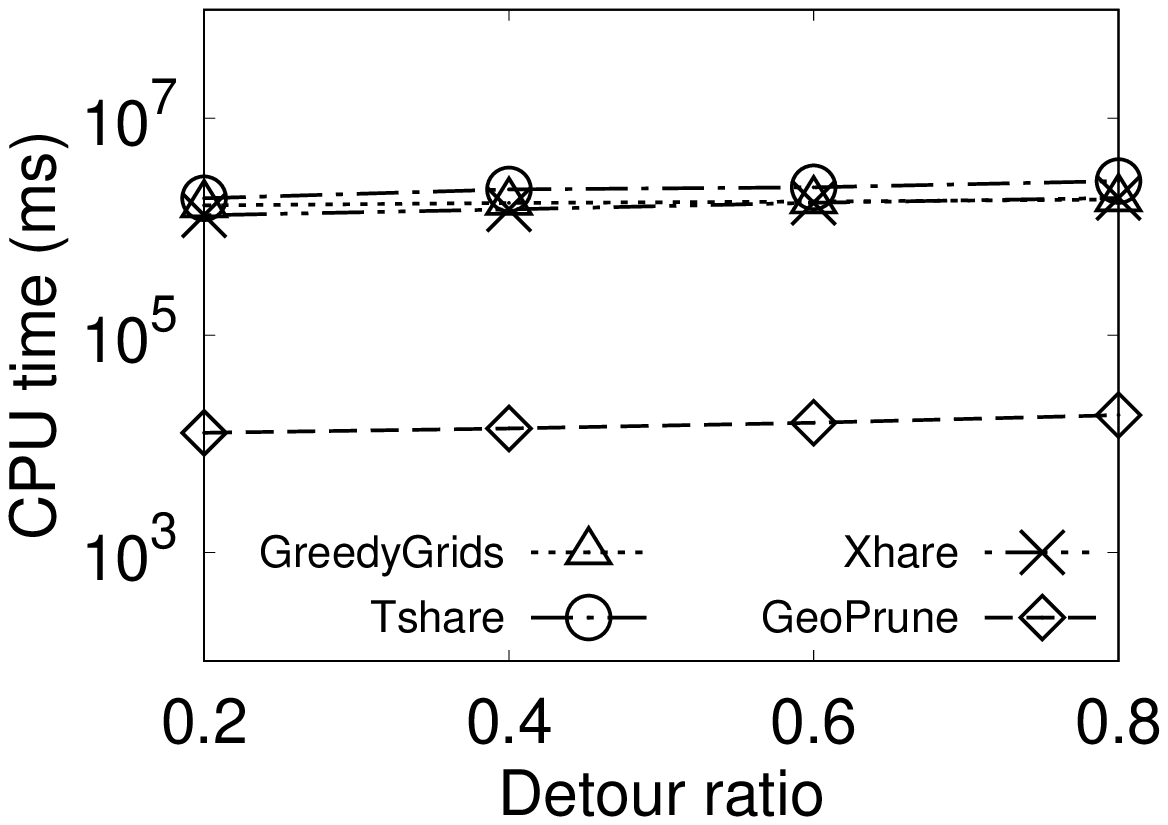}
			\vspace{-6mm}
			\caption{Overall update time (NYC).}%
			\label{fig:NYVaryDetourRatio_overallUpdateTime}
		\end{subfigure}
		\begin{subfigure}[t]{0.24\textwidth}  
			\centering 
			\includegraphics[width=\textwidth]{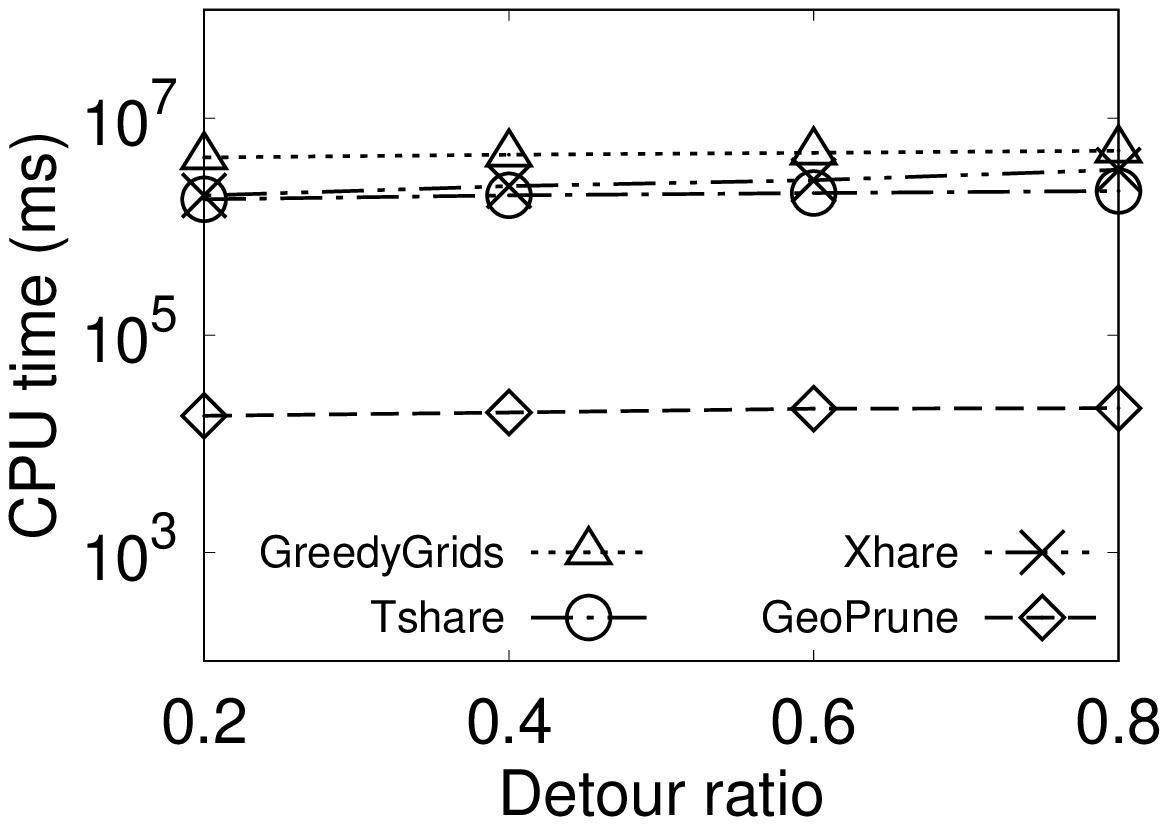}
			\vspace{-6mm}
			\caption{Overall update time (CD).}%
			\label{fig:ChengduVaryDetourRatio_overallUpdateTime}
		\end{subfigure}
		\caption{Effect of the detour ratio. }
		\label{fig:VaryDetourRatio}
	\end{figure}

	\begin{figure}[t]
		\centering
		\begin{subfigure}[t]{0.24\textwidth}
			\centering
			\includegraphics[width=\textwidth]{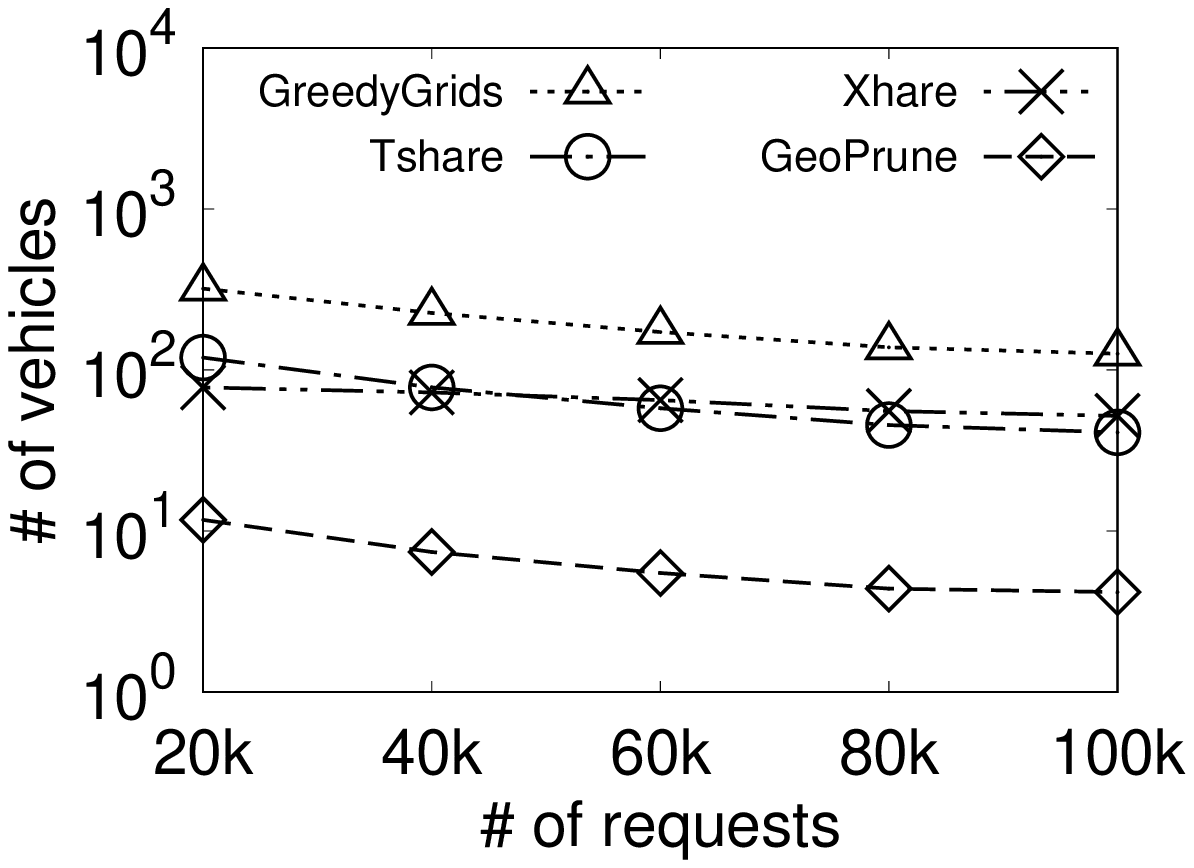}
			\vspace{-6mm}
			\caption{\# remaining vehicles (NYC).}%
			\label{fig:NYVaryNumRequests_numPrune}
		\end{subfigure}
		\begin{subfigure}[t]{0.24\textwidth}
			\centering
			\includegraphics[width=\textwidth]{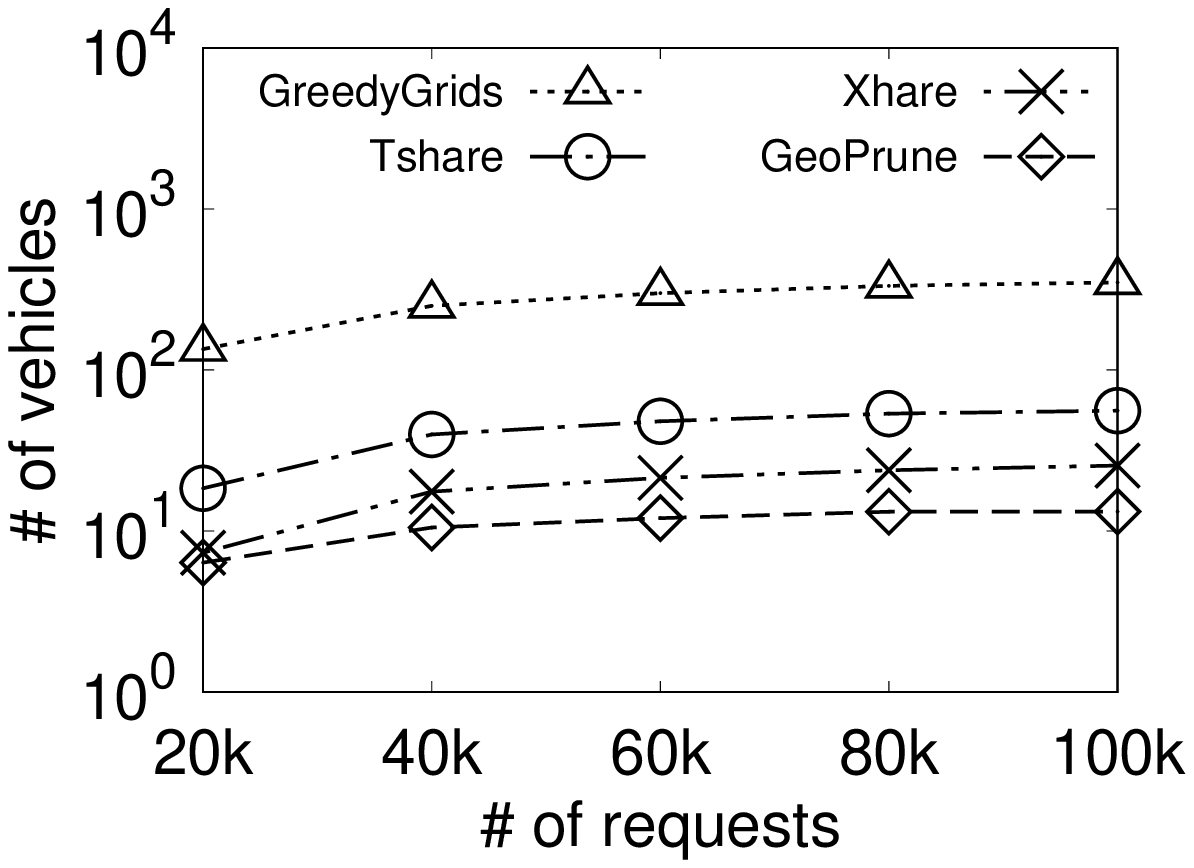}
			\vspace{-6mm}
			\caption{\# remaining vehicles (CD).}%
			\label{fig:ChengduVaryNumRequests_numPrune}
		\end{subfigure}
		\vspace{3mm}
		\begin{subfigure}[t]{0.24\textwidth}
			\centering
			\includegraphics[width=\textwidth]{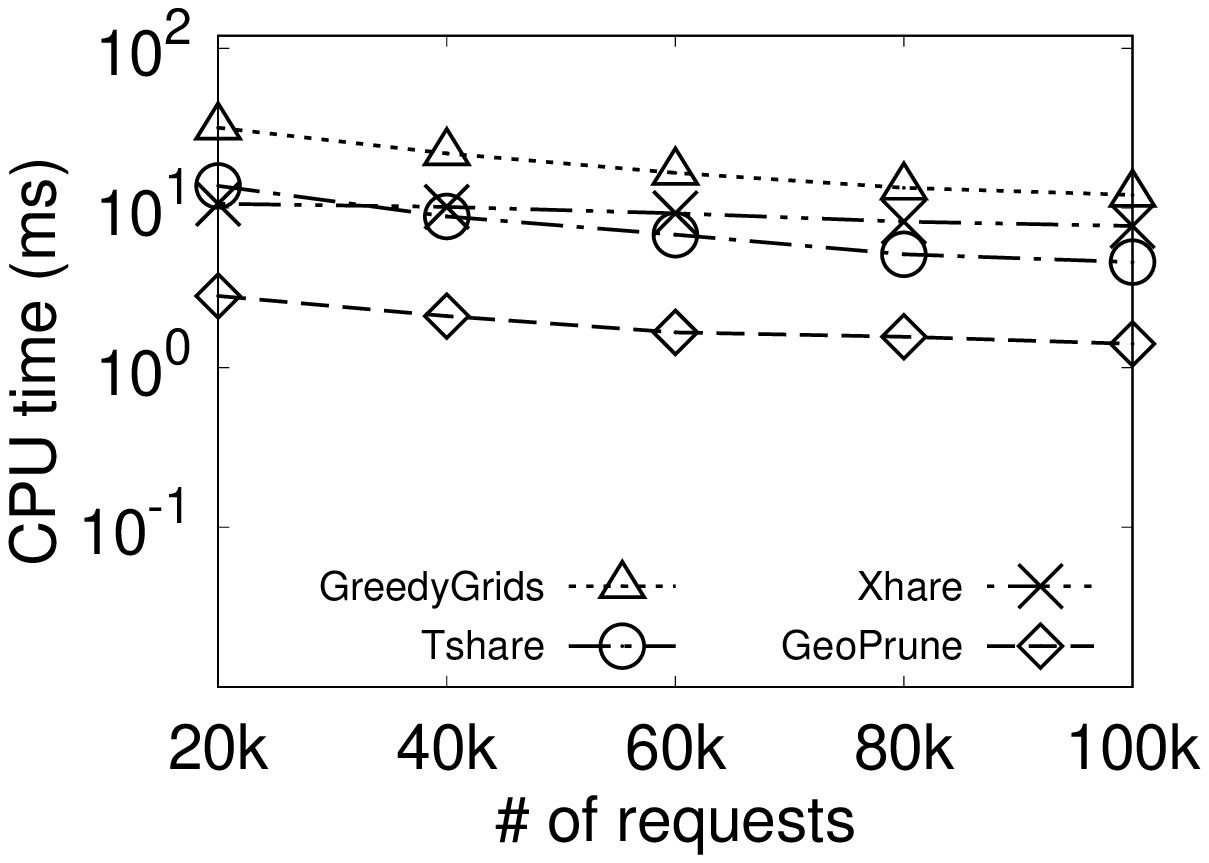}
			\vspace{-6mm}
			\caption{Overall match time (NYC).}%
			\label{fig:NYVaryNumRequests_matchTime}
		\end{subfigure}
		\begin{subfigure}[t]{0.24\textwidth}
			\centering
			\includegraphics[width=\textwidth]{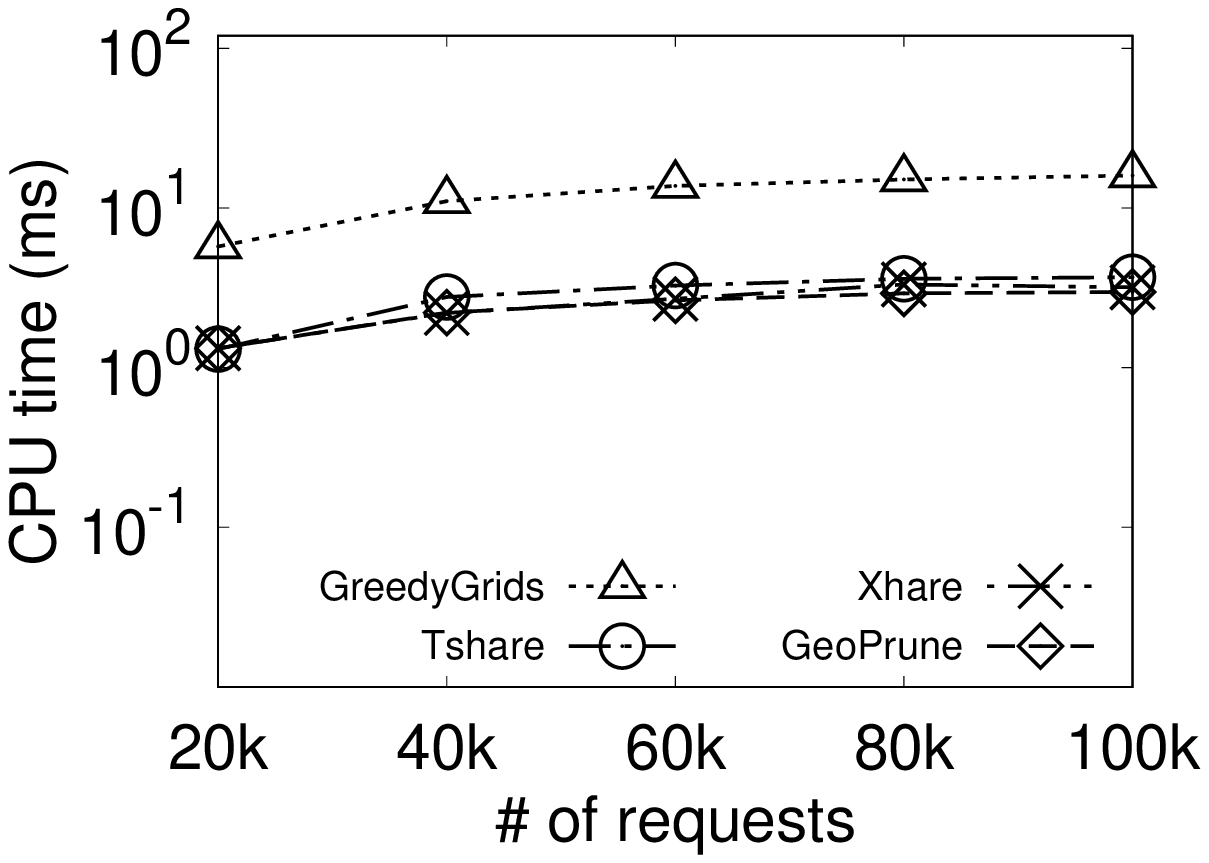}
			\vspace{-6mm}
			\caption{Overall match time (CD).}%
			\label{fig:ChengduVaryNumRequests_matchTime}
		\end{subfigure}
		\vspace{3mm}
		\begin{subfigure}[t]{0.24\textwidth}  
			\centering 
			\includegraphics[width=\textwidth]{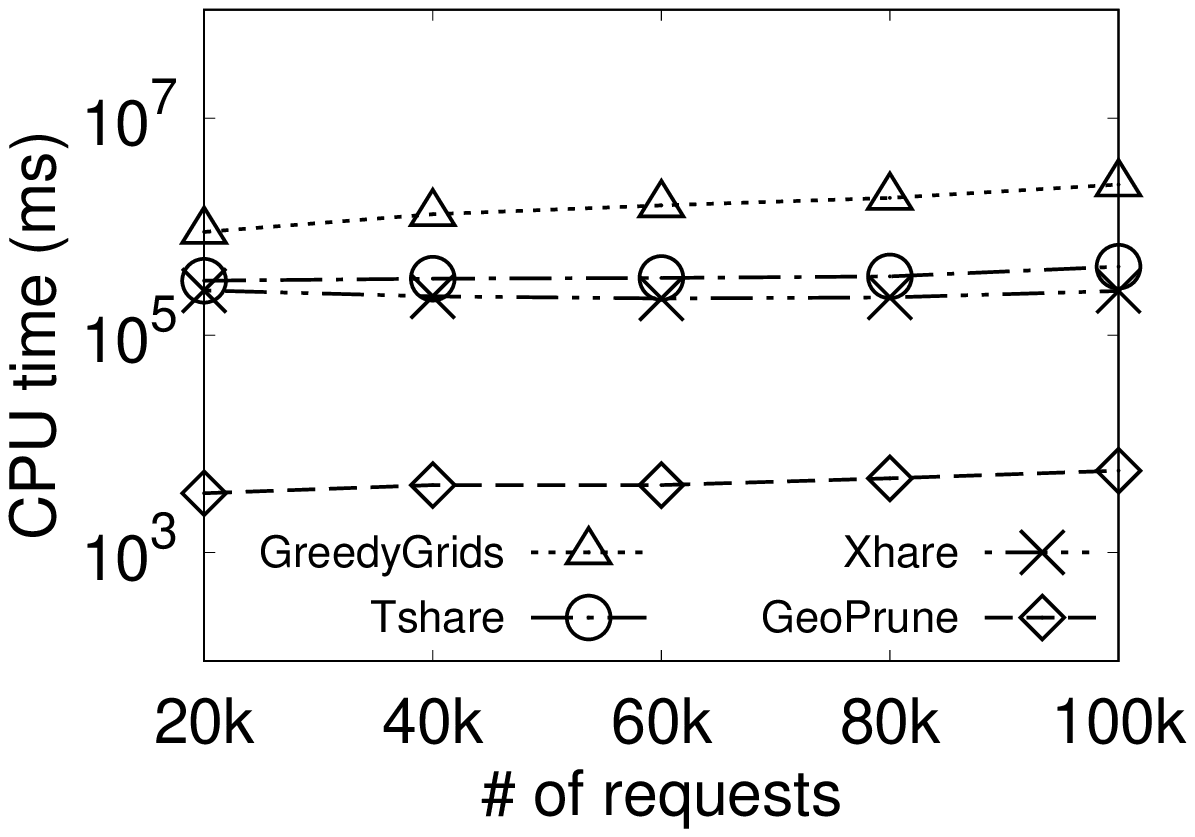}
			\vspace{-6mm}
			\caption{Overall update time (NYC).}%
			\label{fig:NYVaryNumRequests_overallUpdateTime}
		\end{subfigure}
		\begin{subfigure}[t]{0.24\textwidth}  
			\centering 
			\includegraphics[width=\textwidth]{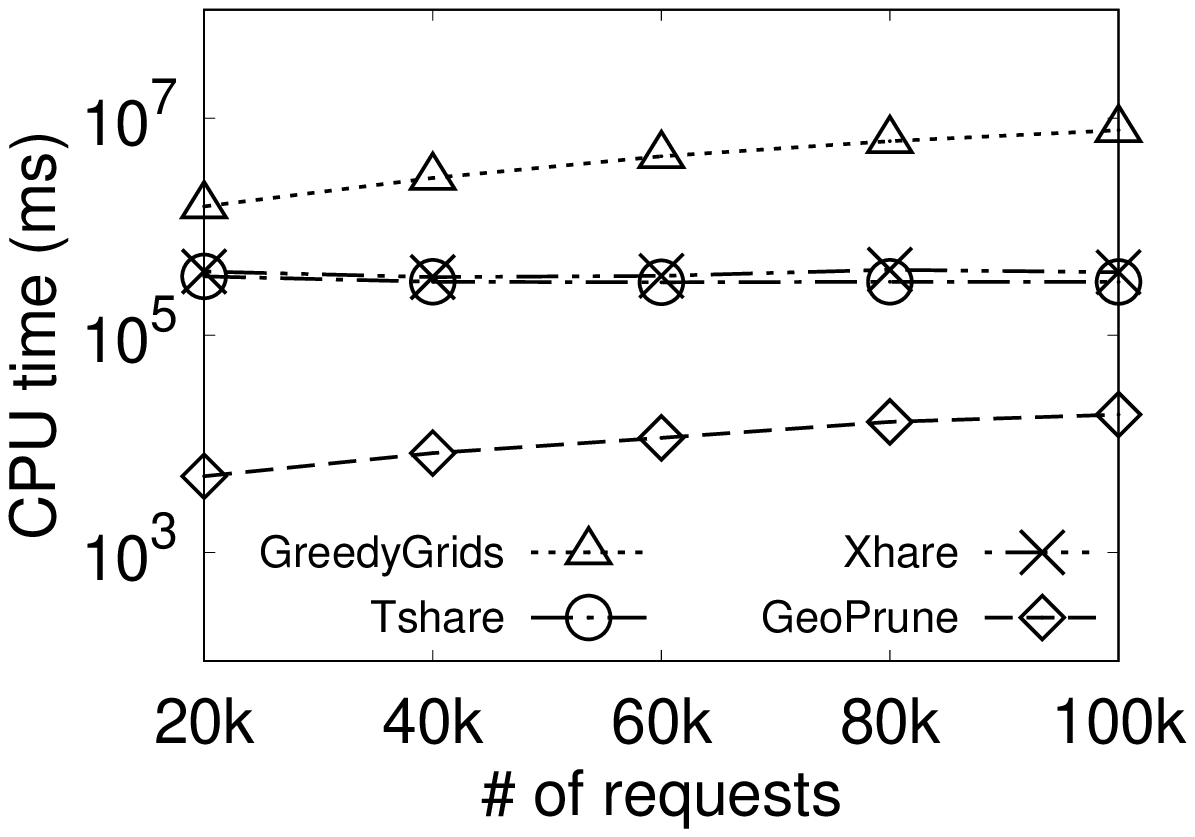}
			\vspace{-6mm}
			\caption{Overall update time (CD).}%
			\label{fig:ChengduVaryNumRequests_overallUpdateTime}
		\end{subfigure}
		\caption{Effect of the number of requests. }
		\label{fig:VaryNumRequests}
	\end{figure}

	Figure~\ref{fig:varyNumVehicles} shows the  results on varying the number of vehicles. \ouralgorithmCapital substantially reduces the number of remaining candidate vehicles compared to other algorithms.
	When the number of vehicles is $2^{13}$, the average number of candidates of  \ouralgorithm is only 5 on the NYC dataset, while the other algorithms return 57 $\sim$ 172 candidates per request.
	GreedyGrids returns the largest set of candidates as it simply retrieves all vehicles in the nearby grid cells, among which only a few are feasible.
	Tshare and Xhare find fewer candidates than GreedyGrids but may result in false negatives due to the approximate search.
	
	The number of remaining vehicles largely affects the running time of the selection stage and the overall match time.
	As shown in Figure~\ref{fig:NYVarynumVehicles_matchTime} and Figure~\ref{fig:chengduVarynumVehicles_matchTime}, \ouralgorithm reduces the overall match time by 71\% to 90\% on the NYC dataset and up to 80\% on the  Chengdu dataset.
	Consistent with experiments shown in the previous studies~\cite{tong2018unified}, \cite{Xu2019Efficient},\cite{pan2019ridesharing}, all algorithms exhibit longer pruning time with more vehicles as the number of vehicle candidates increases.
	The matching time of Tshare and Xhare is comparable with \ouralgorithm on the Chengdu dataset when the number of vehicles is small but continuously increases with more vehicles, showing that GeoPrune scales better with the increase in the number of vehicles.

	In terms of the update cost, \ouralgorithm is two to three orders of magnitude faster because GeoPrune mainly relies on circles and ellipses for pruning while other algorithms require real-time maintenance of indices over the road networks.
	
		\subsubsection{Effect of the capacity of vehicles}
	Figure~\ref{fig:varyCapacity} illustrates the algorithm performance when varying the capacity of vehicles.
	\ouralgorithmCapital outperforms other state-of-the-arts in all capacity settings on NYC dataset and shows comparable match time with Tshare and Xhare on Chengdu dataset.
	As shown in Figure~\ref{fig:NYVaryCapacity_matchTime} and Figure ~\ref{fig:ChengduVaryCapacity_matchTime}, the number of remaining vehicles and the overall match time keep stable when the vehicle capacity varies on both road networks, which may be caused by the limited shareability between trip requests under the parameter settings.
	
	The update cost of algorithms when varying the capacity of vehicles is shown in Figure~\ref{fig:NYVaryCapacity_overallUpdateTime} and Figure~\ref{fig:NYVaryCapacity_overallUpdateTime}.
	The overall update cost of  \ouralgorithm  is again observed to be two to three orders of magnitude faster than other algorithms.
	The reason that why overall update cost is barely affected by the vehicle capacity might be the stable length of the vehicle schedule, which is caused by the limited shareability between requests and the low capacity (at most 10).

	\subsubsection{Effect of the Maximum Waiting Time}
	Figure~\ref{fig:NYVaryWaitingTime} shows the experimental results when varying the maximum waiting time of trip requests.
	All algorithms exhibit longer match time with the increasing waiting time because of larger shareability between requests and more returned vehicle candidates.
	\ouralgorithmCapital again shows the best pruning performance  in almost all cases.
	Tshare requires less matching time than \ouralgorithm when the waiting time is  2\,min on the Chengdu dataset.
	However, longer waiting time requires Tshare to check more nearby grid cells and thus the matching time of Tshare increases continuously and becomes five times slower than \ouralgorithm when the waiting time is 10\,min.
	Xhare finds fewer vehicles and requires less matching time than GeoPrune when the waiting time is longer than 6\,min on the Chengdu dataset.
	This is because Xhare assumes vehicles travel on pre-defined routes, and new requests can only be served on the way of these routes.
	A long waiting time brings more feasible vehicles with append-append case and Xhare may miss these vehicles.

	\begin{figure}[t]
		\centering
		\begin{subfigure}[t]{0.24\textwidth}
			\centering
			\includegraphics[width=\textwidth]{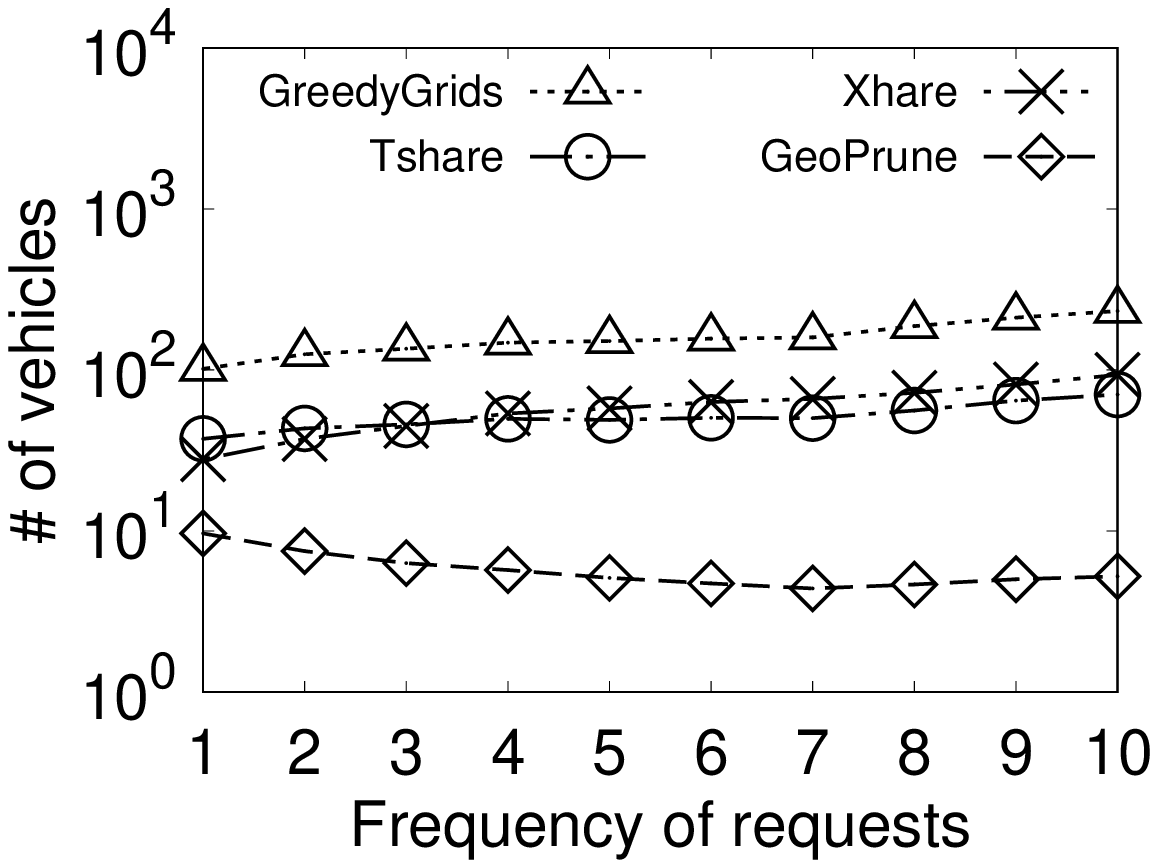}
			\vspace{-6mm}
			\caption{\# remaining vehicles (NYC).}%
			\label{fig:NYVaryFrequency_numPrune}
		\end{subfigure}
		\begin{subfigure}[t]{0.24\textwidth}
			\centering
			\includegraphics[width=\textwidth]{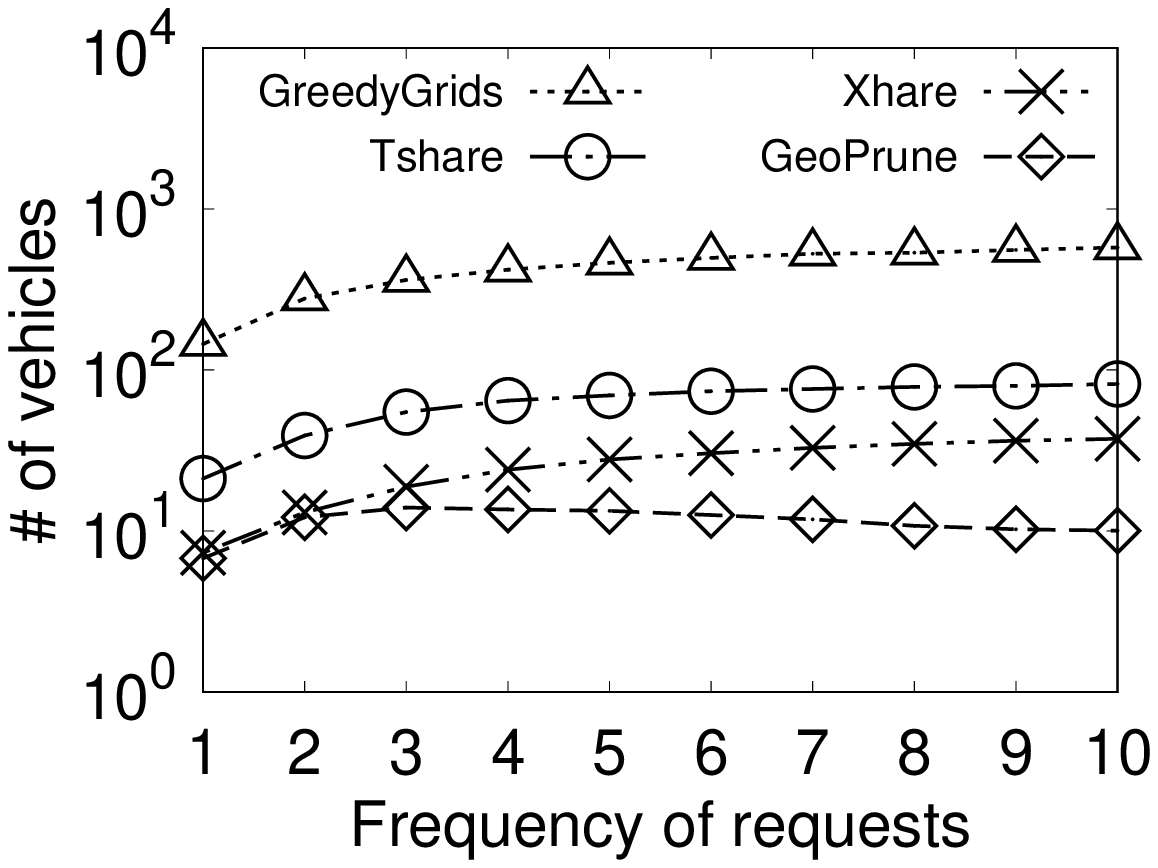}
			\vspace{-6mm}
			\caption{\# remaining vehicles (CD).}%
			\label{fig:ChengduVaryFrequency_numPrune}
		\end{subfigure}
		\vspace{3mm}
		\begin{subfigure}[t]{0.24\textwidth}
			\centering
			\includegraphics[width=\textwidth]{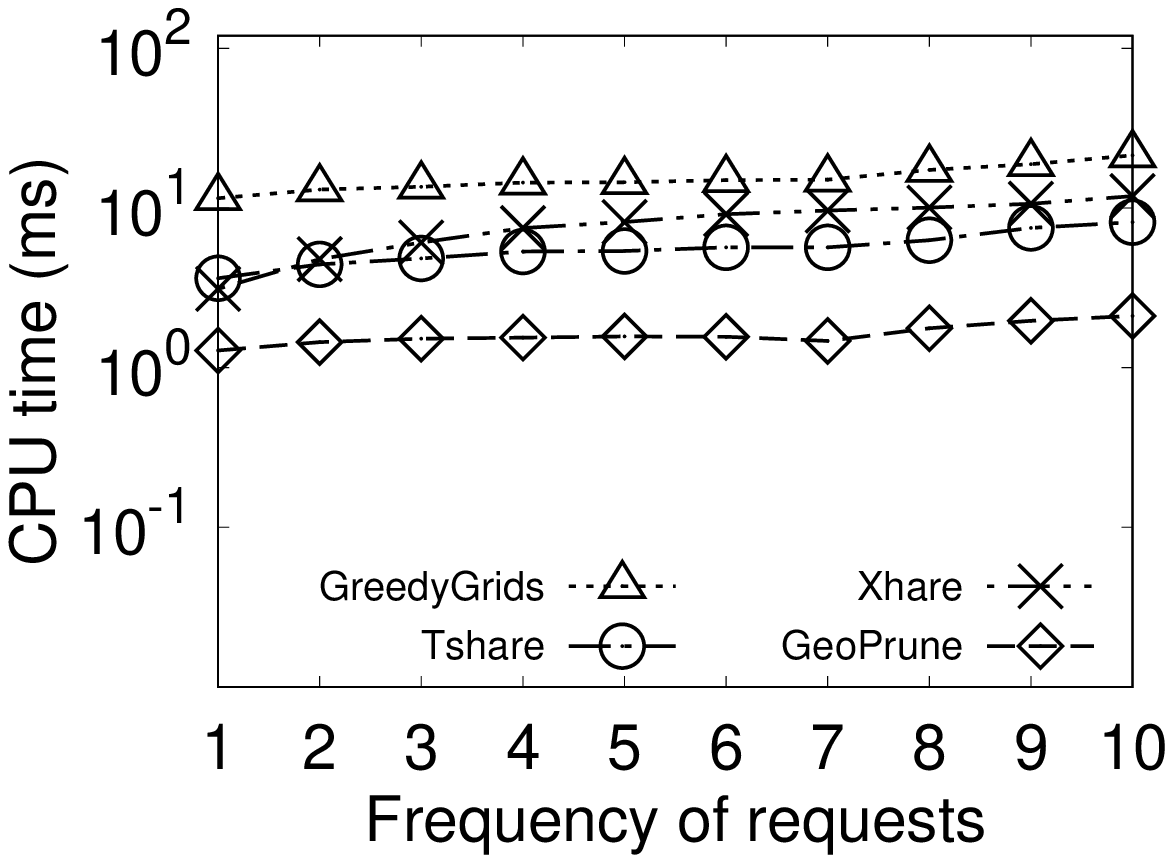}
			\vspace{-6mm}
			\caption{Overall match time (NYC).}%
			\label{fig:NYVaryFrequency_matchTime}
		\end{subfigure}
		\begin{subfigure}[t]{0.24\textwidth}
			\centering
			\includegraphics[width=\textwidth]{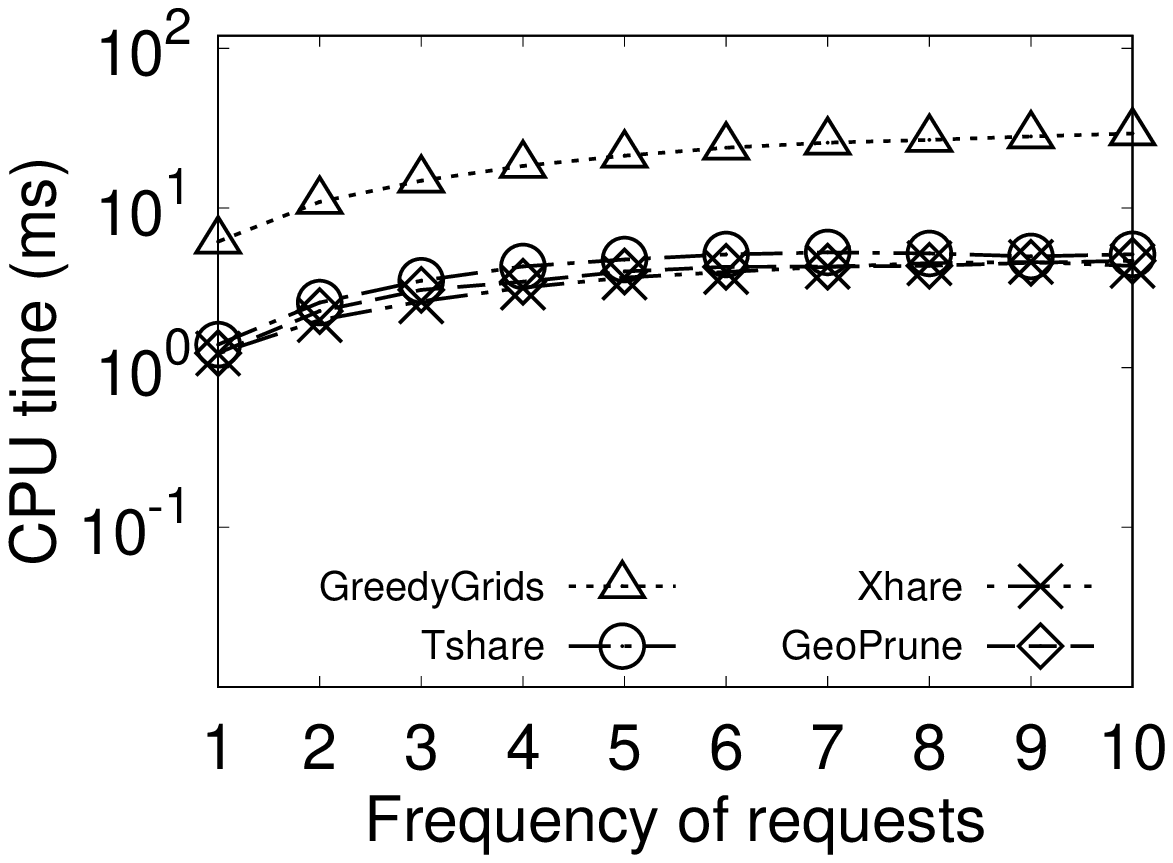}
			\vspace{-6mm}
			\caption{Overall match time (CD).}%
			\label{fig:ChengduVaryFrequency_matchTime}
		\end{subfigure}
		\vspace{3mm}
		\begin{subfigure}[t]{0.24\textwidth}  
			\centering 
			\includegraphics[width=\textwidth]{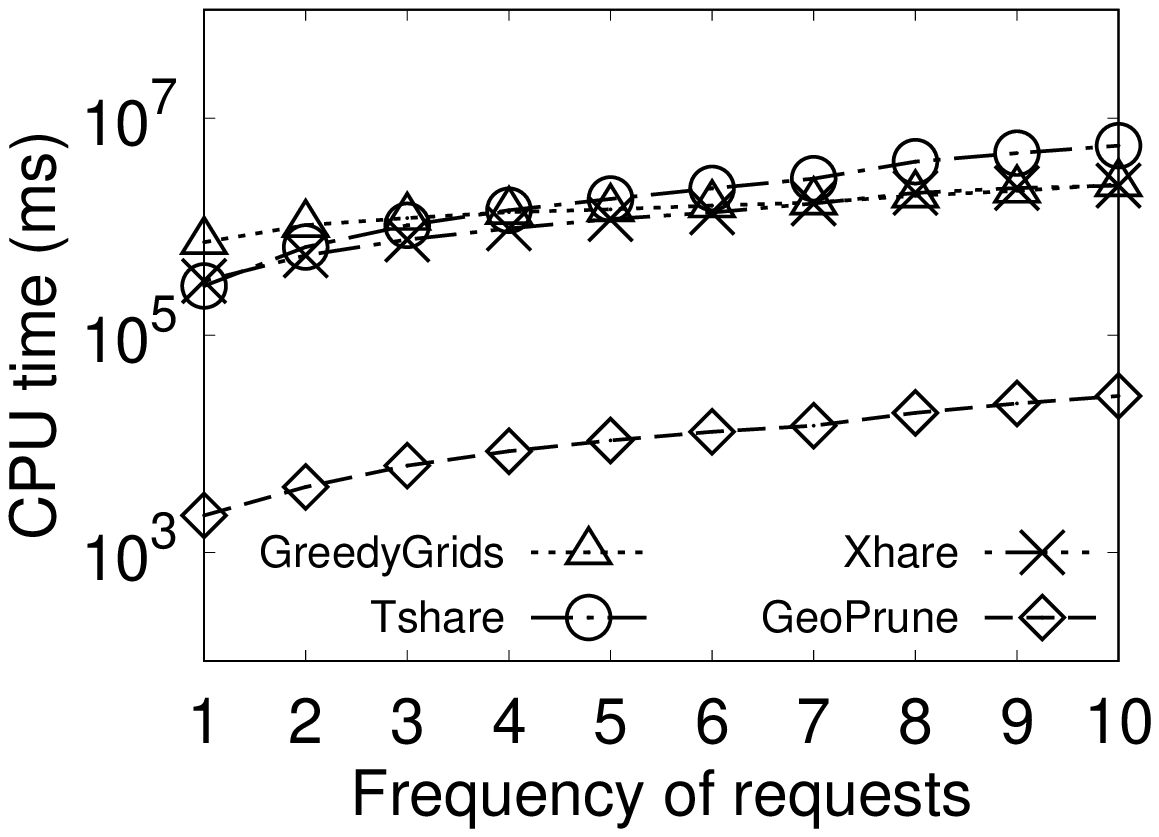}
			\vspace{-6mm}
			\caption{Overall update time (NYC).}%
			\label{fig:NYVaryFrequency_overallUpdateTime}
		\end{subfigure}
		\begin{subfigure}[t]{0.24\textwidth}  
			\centering 
			\includegraphics[width=\textwidth]{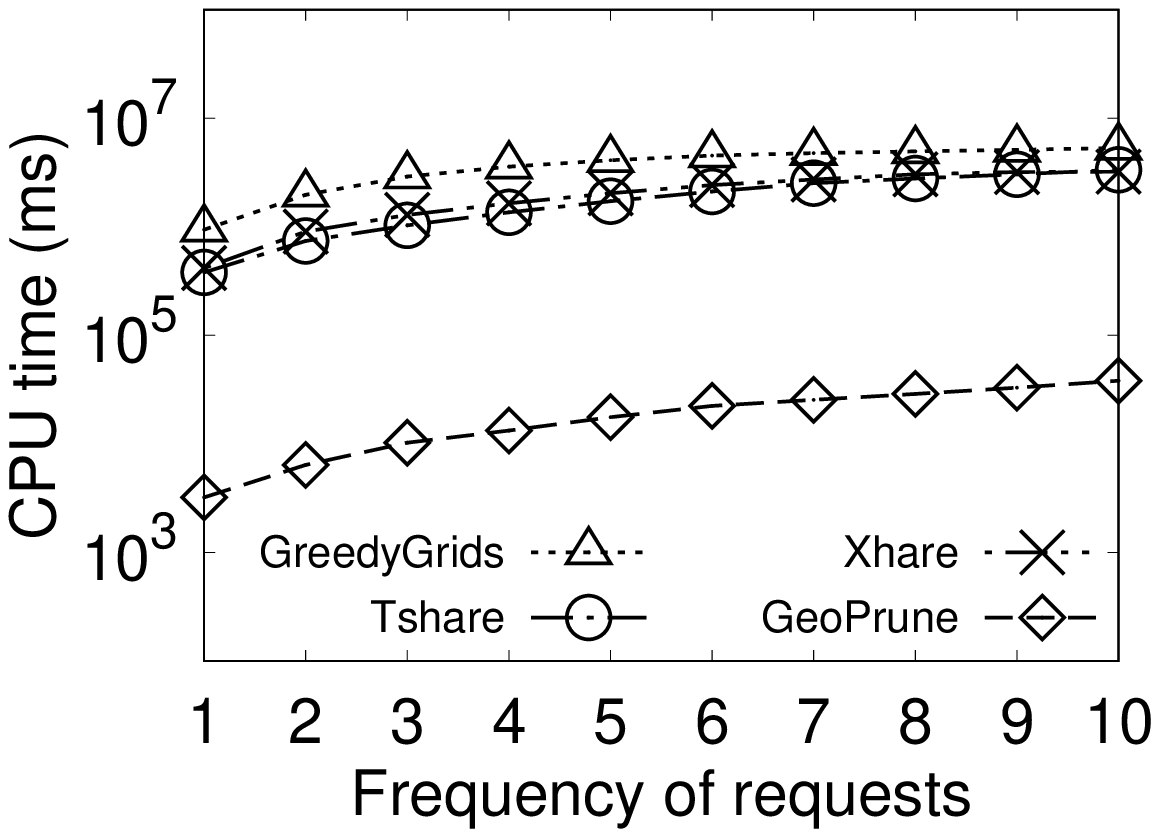}
			\vspace{-6mm}
			\caption{Overall update time (CD).}%
			\label{fig:ChengduVaryFrequency_overallUpdateTime}
		\end{subfigure}
		\caption{Effect of the frequency of requests. }
		\label{fig:VaryFrequency}
	\end{figure}
	Figure~\ref{fig:NYVaryWaitingTime_overallUpdateTime} and Figure~\ref{fig:ChengduVaryWaitingTime_overallUpdateTime} show the update cost, which increases for all algorithms with the larger waiting time, since the schedule of vehicles becomes longer and more requests can be shared.
	Still, \ouralgorithmCapital is two to three orders of magnitude faster on update compared to the state-of-the-art.
	
	\subsubsection{Effect of the Detour Ratio}
	Figure~\ref{fig:VaryDetourRatio} examines the sensitivity over the detour ratio.
	Again, \ouralgorithm prunes more infeasible vehicles, and its match time is three to ten times lower than the other algorithms on the NYC dataset and comparable with Tshare and Xhare on the Chengdu dataset.
	The number of remaining vehicles of all algorithms keeps almost stable due to the limited shareability. The update cost of all algorithms remains stable (and three orders of magnitude smaller for \ouralgorithmEnd) since the length of vehicle schedules is barely affected by different detour ratios.
	%

	\subsubsection{Effect of the Number of Trip Requests}
	Figure~\ref{fig:VaryNumRequests} shows the experiments when the number of requests varies.
	Interestingly, algorithms show different behavior on the two datasets.
	When the number of requests changes from 20\,k to 100\,k, the candidates returned by \ouralgorithm for each request decreases from 11 to 4 on the NYC dataset but increases from 6 to 13 on the Chengdu dataset, meaning that
	the shareability between requests decreases on the NYC dataset while increases on the Chengdu dataset.
	The trend of the overall match time is consistent with that of the number of remaining vehicles, which again validates that the overall match time is largely affected by the number of remaining vehicles.
	
	More trip requests correspond to longer simulation time and increase the total update cost (with \ouralgorithm still being two to three orders of magnitude cheaper in terms of update cost).

	\subsubsection{Effect of the Trip Request Frequency}
	Figure~\ref{fig:VaryFrequency} shows the scalability of algorithms with the frequency of trip requests varying from 1 to 10 requests per second over 3\,hours.
	Note that the frequencies of the original NYC and Chengdu datasets are  5.19  and 3 requests per second respectively.
	To generate trip requests less frequent than the original datasets, we uniformly sample trip requests from the original datasets.
	As for more frequent trip requests, we extract a certain number of trip requests  according to the frequency, e.g., 10,800$\times$7=75,600 trip requests when the frequency is 7.
	We then uniformly distribute the request issue time over 3\,hours.
	
	All algorithms return more vehicle candidates when the frequency increases while \ouralgorithm keeps almost stable.
	This shows that \ouralgorithm provides tighter pruning and is more scalable to highly dynamic scenarios.
	The overall matching time of \ouralgorithm consistently outperforms others  on the NYC dataset and is comparable with Tshare and Xhare on Chengdu dataset.
	The update cost of all algorithms grows with higher frequency due to more frequent updates while \ouralgorithm again outperforms others by two to three orders of magnitude.

	\begin{figure}[t]
		\centering
		\begin{subfigure}[t]{0.24\textwidth}  
			\centering 
			\includegraphics[width=\textwidth]{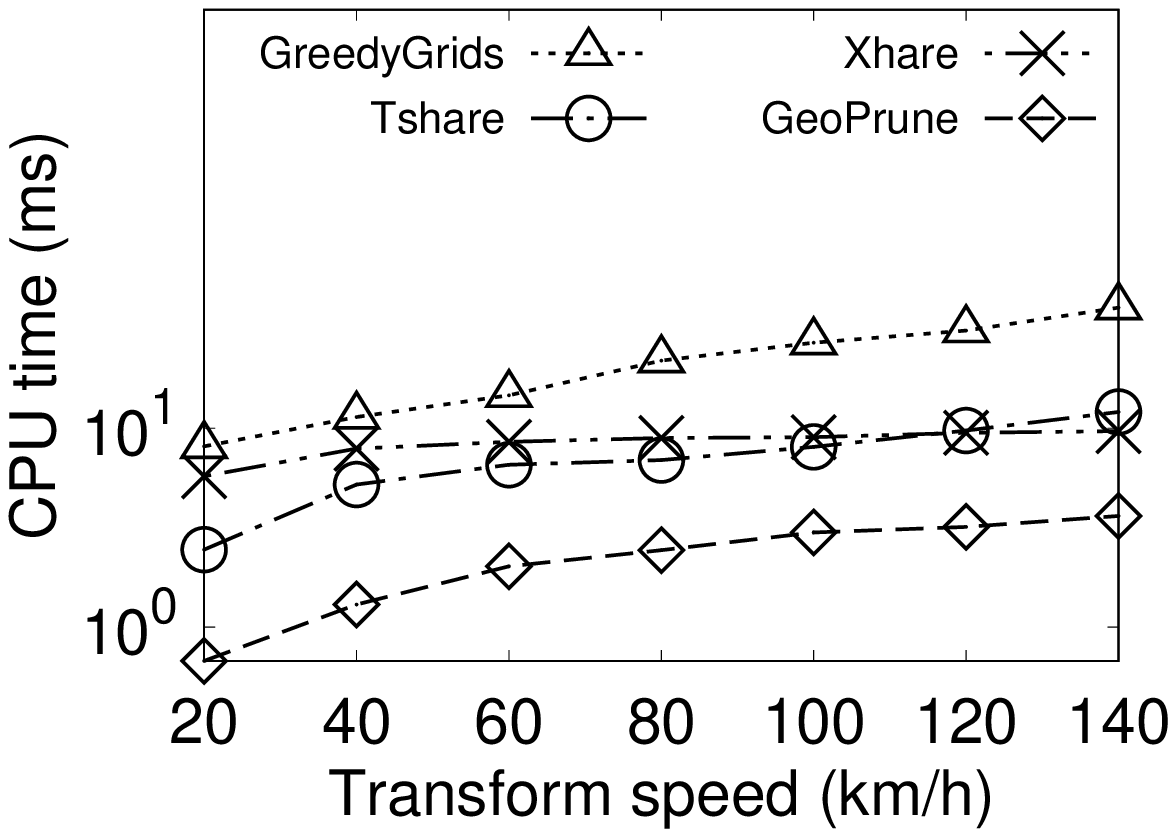}
			\vspace{-6mm}
			\caption{Overall match time.}%
			\label{fig:NYVarySpeed_overallMatchTime}
		\end{subfigure}
		\begin{subfigure}[t]{0.24\textwidth}  
			\centering 
			\includegraphics[width=\textwidth]{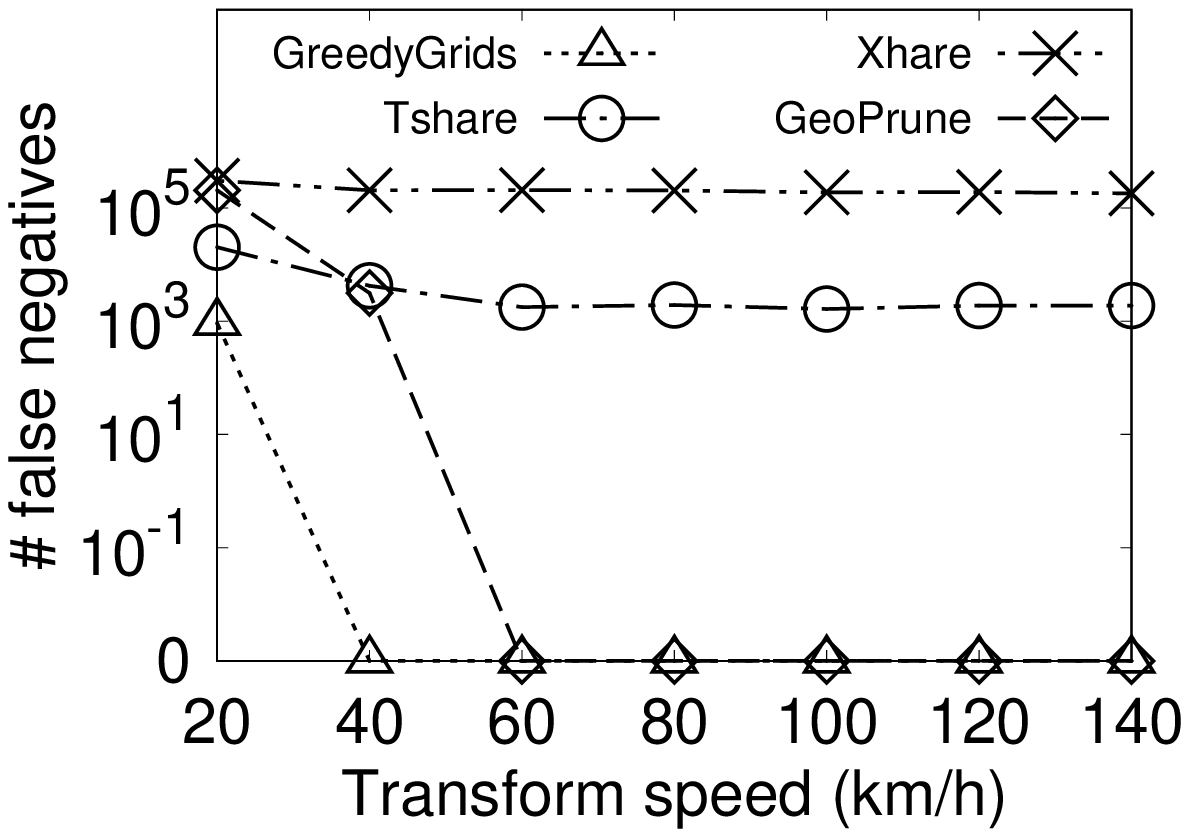}
			\vspace{-6mm}
			\caption{ \# of false negatives.}%
			\label{fig:NYVarySpeed_FNs}
		\end{subfigure}
		\caption{Effect of the transforming speed (NYC). }
		\label{fig:VarySpeed}
	\end{figure}

	\subsubsection{Effect of the Transforming Speed:}
	All algorithms need a speed value to transform the time constraint to distance constraint so that pruning based on geographical locations can be applied. 
	Figure~\ref{fig:VarySpeed} shows the effect of the transforming speed.
	A higher speed enlarges the search space and thus all algorithms show longer match time.
	However, GeoPrune consistently performs efficient pruning with all speed values.
	Figure~\ref{fig:NYVarySpeed_FNs} shows the total number of false negatives wrongly pruned for 60,000 requests.
	The same as GreedyGrids, GeoPrune ensures no false negatives when the speed is greater than the vehicle speed (48km/h), whereas Xhare and Tshare still result in false negatives even with a high transforming speed.
	This demonstrates the robustness of GeoPrune to real-time traffic conditions, where the transforming speed can be set as the maximum speed of vehicles such that the pruning is still correct and the processing time only increases slightly.

	\begin{table}[t]
		\caption{Memory consumption (MB) (\# vehicles = $2^{13}$).}
		\label{tab:memory}
		\setlength{\tabcolsep}{3pt}
		\centering
		\begin{tabular}{l r r  r r }
			\hline
			& \textbf{GreedyGrids} & \textbf{Tshare} & \textbf{Xhare} &\textbf{\ouralgorithm}  \\ \hline
			NYC   &   0.38    &   100.34   &   1546.40 & 6.56    \\ 
			Chengdu & 1.67 & 9965.37 & 21282.46 & 6.43  \\
			\hline
		\end{tabular}
	\end{table}
	
	\subsubsection{Memory Consumption}
	Table~\ref{tab:memory} illustrates the maximum memory usage of the algorithms under the default setting.
	All state-of-the-arts consume more memory on the Chengdu dataset as it has a large road network,  while \ouralgorithm keeps stable.
 All state-of-the-arts require maintaining an index over the road network, which is hence largely affected by the network size.
	For example, the grid size of Tshare in NYC is 46$\times$46 but increases to 174$\times$174 in Chengdu.
	\ouralgorithmEnd, however, only maintains several R-trees and thus is less affected by the road network size.
	GreedyGrids has the smallest memory footprint as it only records a list of in-cell vehicles for each grid cell, which is consistent with the observation in~\cite{tong2018unified}.
	Tshare and Xhare consume much more memory than \ouralgorithm due to the large road network index.

	\subsubsection{Cost Breakdown of Algorithm Steps}
	Figure~\ref{fig:Sensitivity analysis} compares the cost of different phases in the match process and update process when varying the number of requests on the NYC dataset.
	Figure~\ref{fig:NYVaryNumRequests_search} shows the cost of the pruning algorithms while Figure~\ref{fig:NYVaryNumRequests_selection} shows the selection cost based on their pruning results.
	\ouralgorithm requires slightly longer time for pruning than Tshare and Xhare but can reduce the selection time by more than 88\% due to the fewer remaining vehicles.
	The selection time of algorithms (Figure~\ref{fig:NYVaryNumRequests_selection}) is consistent with the number of remaining vehicles (Figure~\ref{fig:NYVaryNumRequests_numPrune}), which again demonstrates that the selection step is largely affected by the number of remaining vehicles.

	Figure~\ref{fig:NYVaryNumRequests_matchUpdateTime} shows the update cost when a request is newly assigned.
	\ouralgorithmCapital takes slightly longer time than GreedyGrids to update the R-trees.
	Xhare and Tshare, however, need much more time than \ouralgorithm as they need to  update the pass-through and reachable areas of the matched vehicle while \ouralgorithm can quickly bound the reachable areas by ellipses.
	
	Figure~\ref{fig:NYVaryNumRequests_moveUpdateTime} compares the update cost when vehicles are moving in the street. GreedyGrids is two orders of magnitude slower than the other three algorithms because it needs to track the located grid cells of continuously moving vehicles.
	
	\begin{figure}[t]
		\centering
		\begin{subfigure}[t]{0.24\textwidth}
			\centering
			\includegraphics[width=\textwidth]{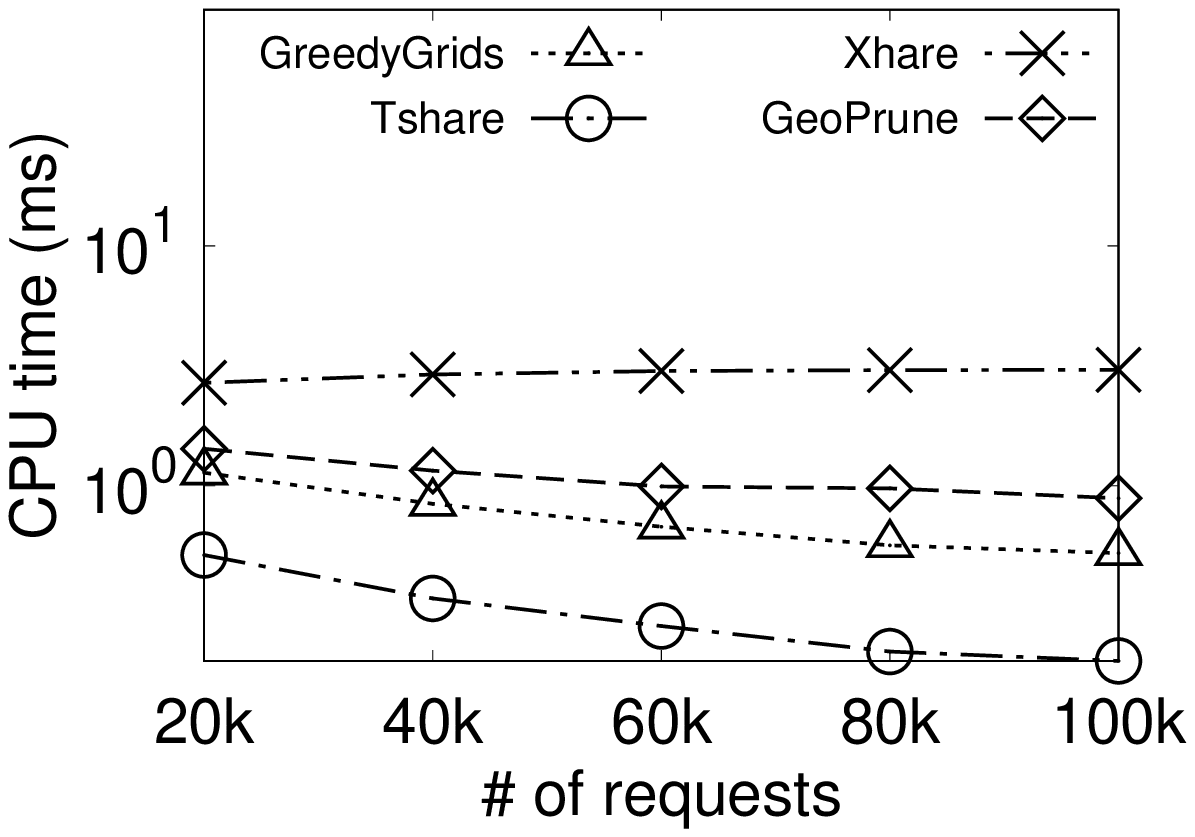}
			\vspace{-6mm}
			\caption{Pruning time.}%
			\label{fig:NYVaryNumRequests_search}
		\end{subfigure}
		\begin{subfigure}[t]{0.24\textwidth}
			\centering
			\includegraphics[width=\textwidth]{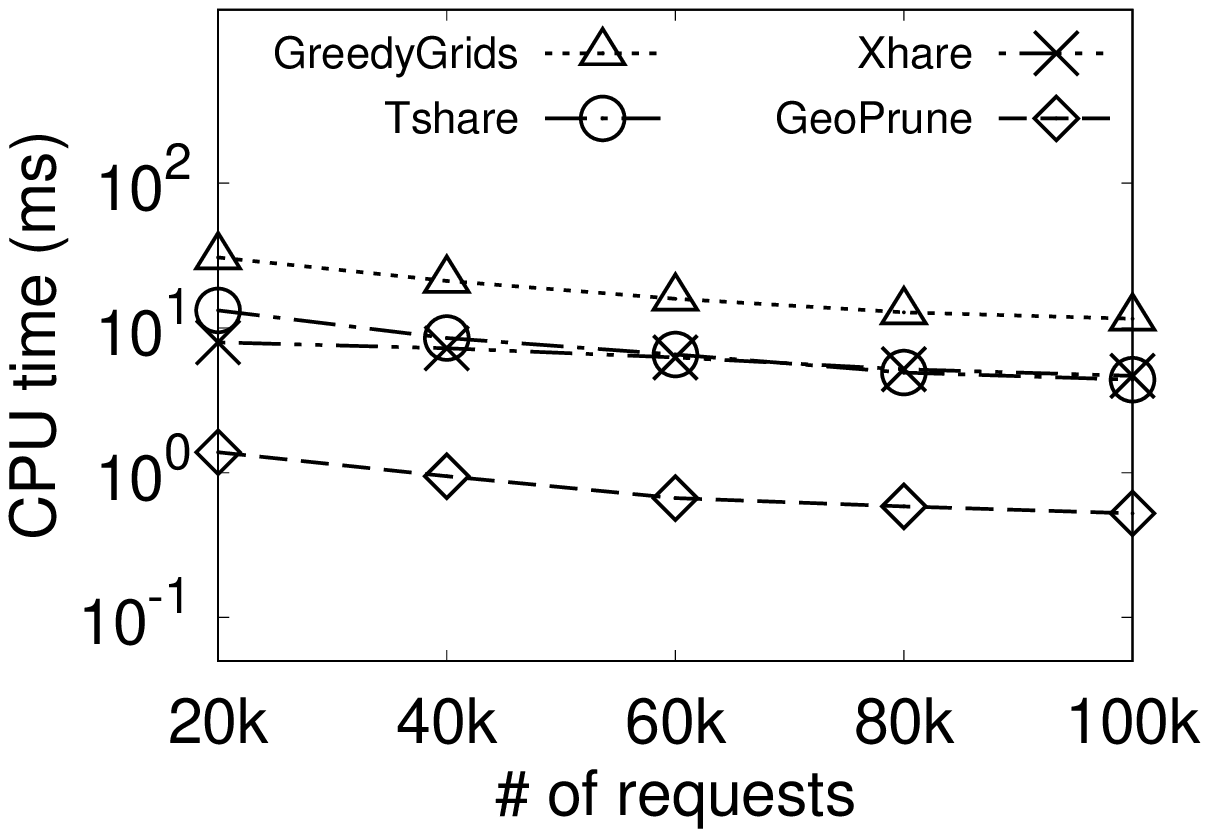}
			\vspace{-6mm}
			\caption{Selection time.}%
			\label{fig:NYVaryNumRequests_selection}
		\end{subfigure}
		\vspace{3mm}
		\begin{subfigure}[t]{0.24\textwidth}
			\centering
			\includegraphics[width=\textwidth]{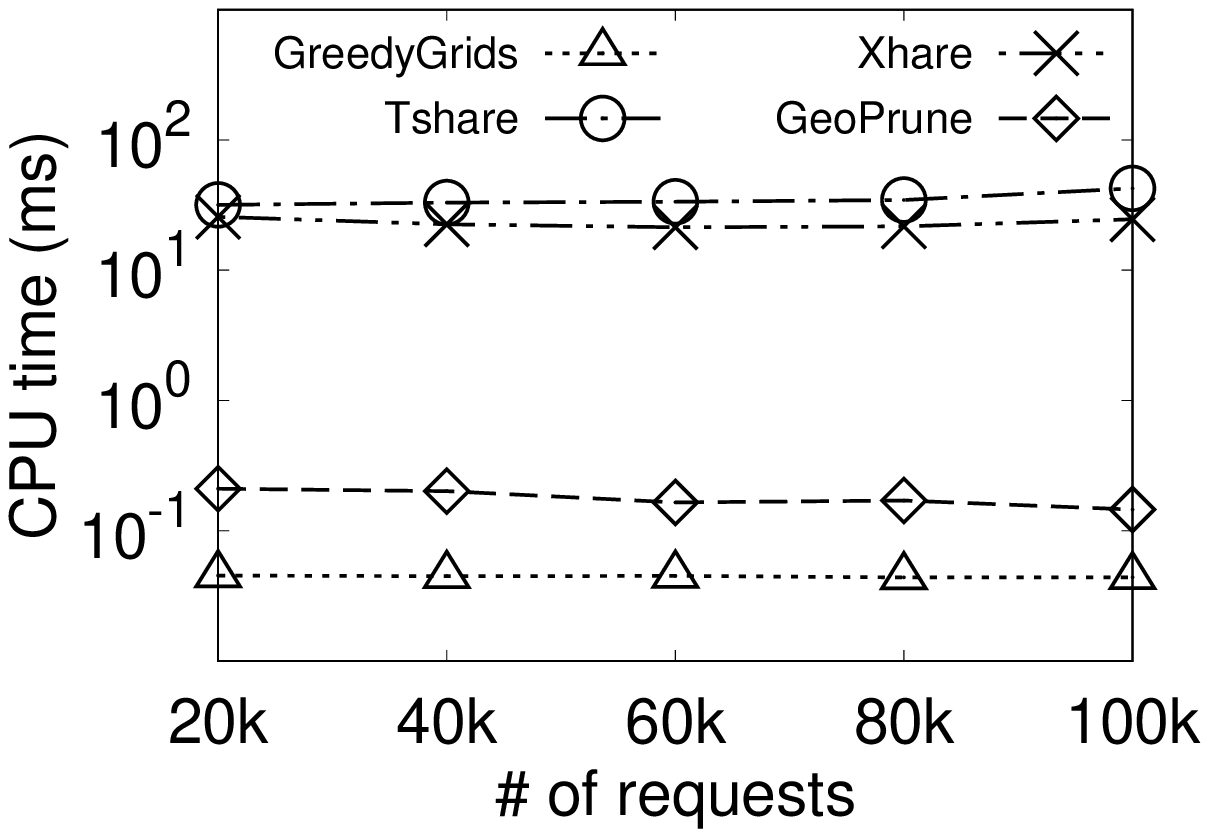}
			\vspace{-6mm}
			\caption{Match update time.}%
			\label{fig:NYVaryNumRequests_matchUpdateTime}
		\end{subfigure}
		\begin{subfigure}[t]{0.24\textwidth}
			\centering
			\includegraphics[width=\textwidth]{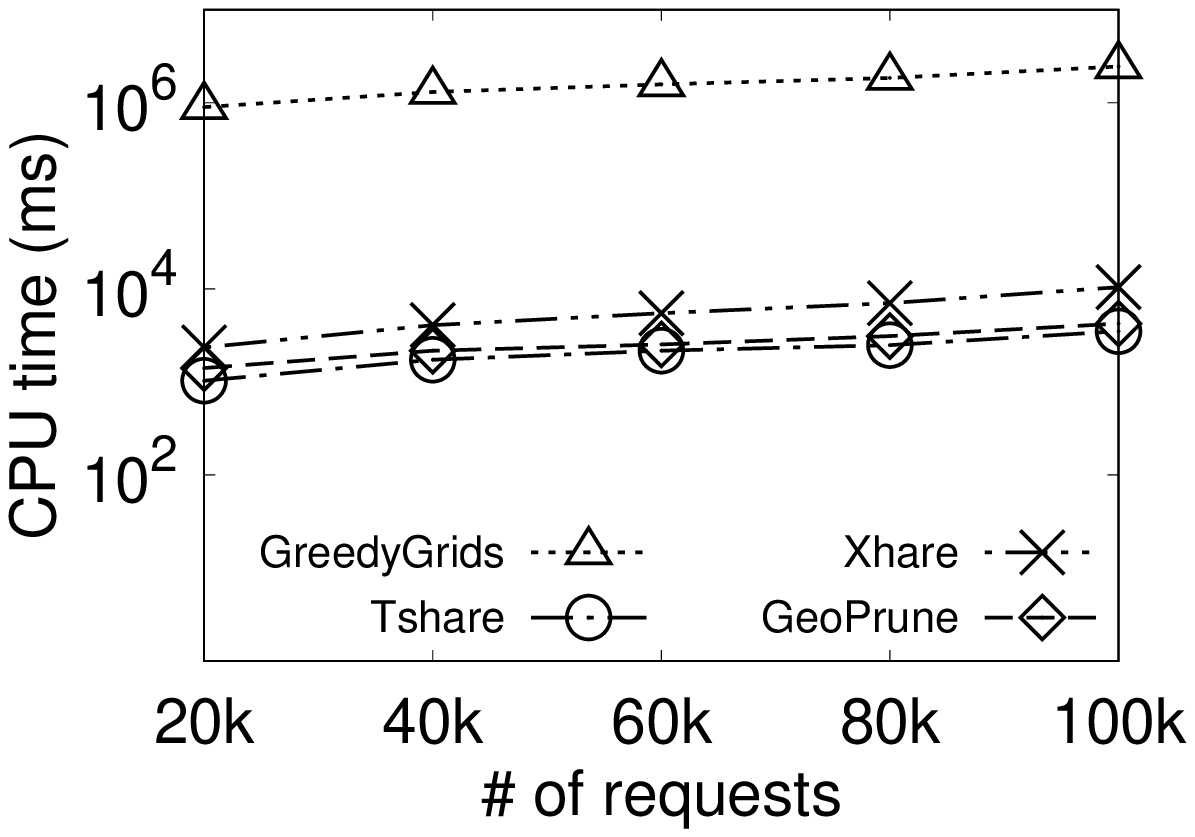}
			\vspace{-6mm}
			\caption{Move update time.}%
			\label{fig:NYVaryNumRequests_moveUpdateTime}
		\end{subfigure}
		\caption{The cost breakdown of algorithm steps. }
		\label{fig:Sensitivity analysis}
	\end{figure}

	\vspace{-5pt}
	\section{Related work}
	\label{sec:relatedWork}
	
	
	Dynamic ride-sharing matching has been studied with different optimization goals and constraints.
	A common optimization goal is to minimize the total travel cost of vehicles~\cite{ma2013t},\cite{tong2018unified},\cite{thangaraj2017xhare},\cite{huang2014large},\cite{alonso2017demand}.
	A few other studies aim to provide a better service experience to passengers~\cite{cheng2017utility},\cite{Xu2019Efficient}, \cite{duan2016real}.
	There are also studies that maximize the overall profit of the ride-sharing platform~\cite{asghari2016price},\cite{zheng2018order},\cite{zheng2019auction}.
	A common need in  these problems is to efficiently filter out infeasible vehicles that violate the service constraints, such that computing the optimal vehicles 
	from the remaining ones can be done with lower costs.
	
	%
	
	We next discuss the studies that aim to minimize the total travel cost as we use it to examine our algorithm.
	Huang et~al.~\cite{huang2014large} maintain a \emph{kinetic tree} for each vehicle to record all possible routes instead of a single optimal route.
	GeoPrune can be easily integrated into their setting by computing the detour ellipses of all  possible routes.
	Alonso et~al.~\cite{alonso2017demand} assign requests to vehicles in batches. They first compute the shareability between requests and vehicles and then construct a graph to connect shareable requests and vehicles.
	Their shareability computation requires an exhaustive check on all possibilities, which can be streamlined by  \ouralgorithm.
	The state-of-the-art selection algorithm~\cite{tong2018unified} for minimizing the total vehicle travel time   first filters infeasible vehicles by checking whether inserting the new request to the vehicle schedules is valid based on the Euclidean distance. 
	After filtering, it ranks all remaining vehicles using the increased distance calculated from the Euclidean distance insertions and sequentially checks these remaining vehicles using road network distances.
	Although this algorithm has a pruning step,  \ouralgorithm  can be applied before it to reduce the number of vehicles for individual checking to further improve the efficiency.

	Next, we discuss existing algorithms for pruning infeasible vehicles -- Tshare~\cite{ma2013t} and Xhare~\cite{thangaraj2017xhare}. 
	Tshare builds an index over the road network by partitioning the space into equi-sized grid cells.
	The distance between two objects (e.g., a request and a vehicle) is estimated using the centers of their corresponding cells. Such an estimation may miss feasible vehicles.
	The geometric objects applied in GeoPrune, in comparison, bound the reachable areas and ensure all feasible vehicles to be returned.
	Moreover, Tshare stores pairwise distances between all grid cells and is not scalable to large networks due to the high memory cost.
	Besides, Tshare maintains the pass-through grid cells of vehicles in real-time, which is costly for highly dynamic scenarios.
	In comparison,  GeoPrune  can quickly bound the reachable areas using ellipses and index these areas using R-trees, which saves computation and update costs.
	Xhare partitions the road network into three levels: grid cells, landmarks, and clusters. The reachable areas of vehicles are estimated using the distance between clusters.
	Similar to Tshare, Xhare is an approximate method. The index of Xhare may have a large memory footprint for large networks and high update cost for dynamic scenarios.

	
	
	

	\section{Conclusions}
	We studied the dynamic ride-sharing matching problem 
	and proposed an efficient algorithm named \ouralgorithm to prune infeasible vehicles to serve trip requests.
	Our algorithm applies geometric objects to bound the areas that vehicles can visit without violating the service constraints of passengers.
	The applied geometric objects are simple to compute and further indexed using efficient data structures, which makes \ouralgorithm highly efficient and scalable.
	Our experiments on real-world datasets confirm the advantages of \ouralgorithm in pruning effectiveness and matching efficiency. 
	In the future, it is worth exploring the advantages of \ouralgorithm in other ride-sharing settings or to other spatial crowd-sourcing problems.

	\bibliographystyle{IEEEtran}
	\bibliography{ridesharing_arxiv}

\end{document}